\documentclass{amsart}
\usepackage[margin=1in]{geometry}
\usepackage{amsmath, amsfonts, amssymb}
\usepackage{xcolor}
\usepackage{algpseudocode}
\usepackage[algo2e,linesnumbered,vlined,ruled]{algorithm2e}
\usepackage{siunitx}
\usepackage{graphicx}
\usepackage{subfig}
\usepackage{booktabs} 
\usepackage{hyperref,cleveref}
\hypersetup{
    colorlinks=true,
    linkcolor=blue,
    citecolor=magenta,
    filecolor=magenta,      
    urlcolor=cyan,
    }
\graphicspath{{images/}}

\newcommand{\ml}{\left(}
\newcommand{\mr}{\right)}

\newcommand\ii{\mathrm{i}}

\DeclareMathOperator\re{Re}
\DeclareMathOperator\im{Im}
\DeclareMathOperator*{\argmin}{arg\,min}
\newcommand{\SU}[1]{\operatorname{SU}(#1)}

\newcommand{\SL}[1]{\operatorname{SL}(#1,\mathbb{C})}
\newcommand{\tr}{\operatorname{tr}}

\newtheorem{theorem}{Theorem}[section]
\newtheorem{lemma}[theorem]{Lemma}
\newtheorem{proposition}[theorem]{Proposition}

\theoremstyle{remark}
\newtheorem*{remark}{Remark}

\title{Regularization of Complex Langevin Method}

\author{Zhenning Cai}
\address[Zhenning Cai]{Department of Mathematics, National University of Singapore,
  Level 4, Block S17, 10 Lower Kent Ridge Road, Singapore 119076}
\email{matcz@nus.edu.sg}

\author{Yang Kuang}
\address[Yang Kuang]{School of Mathematics and Statistics, Guangdong University of Technology,
  Guangzhou 510006, Guangdong, China}
\email{ykuang@gdut.edu.cn}

\author{Hong Kiat Tan}
\address[Hong Kiat Tan]{Department of Mathematics, University of California, Los Angeles, 520 Portola Plaza, Los Angeles, CA 90095, USA}
\email{maxtanhk@math.ucla.edu}

\thanks{Zhenning Cai was supported by the Academic Research Fund of the Ministry of Education of
Singapore under grant No. R-146-000-291-114.}


\begin{document}

\maketitle

\begin{abstract}
The complex Langevin method, a numerical method used to compute the ensemble average with a complex partition function, often suffers from runaway instability.
We study the regularization of the complex Langevin method via augmenting the action with a stabilization term.
Since the regularization introduces biases to the numerical result, two approaches, named $2R$ and $3R$ methods, are introduced to recover the unbiased result.
The $2R$ method supplements the regularization with regression to estimate the unregularized ensemble average, and the $3R$ method reduces the computational cost by coupling the regularization with a reweighting strategy before regression.
Both methods can be generalized to the $\SU{n}$ theory and are assessed from several perspectives.
Several numerical experiments in the lattice field theory are carried out to show the effectiveness of our approaches.

\end{abstract}

\section{Introduction}
The complex Langevin method is a numerical approach used to circumvent the numerical sign problem arising in the computation of ensemble averages with complex Boltzmann weights.
Such issues may appear in the real-time quantum field theories \cite{berges2007lattice, alexandru2016monte}, coupled quantum systems with chemical potentials  such as the Hubbard model \cite{white1989numerical, wynen2021machine} and the quantum chromodynamics at finite density \cite{muroya2003lattice, aarts2008stochastic}, and also the superstring theory \cite{anagnostopoulos2018complex, nishimura2019complex}.
In these applications, one usually encounters strong oscillations in high-dimensional functions, leading to significant cancellations when integrating the functions.
As a result, the classical Monte Carlo
method fails to work as the variance is large compared to the mean value, and such difficulty is known
as the numerical sign problem \cite{loh1990sign}.

The complex Langevin method, introduced in \cite{klauder1983langevin,parisi1983complex}, tries to tame the numerical sign problem by using a straightforward extension of a classical sampling method called the Langevin method.
The extension allows the samples to take complex values due to the complex Boltzmann weights.
Unfortunately, the application of the complex Langevin method had been severely limited for a long time due to one of its major drawbacks: the method often diverges or converges to incorrect solutions \cite{ambjorn1986complex}.
The justification of the method and the understanding of its failure were explored in several works \cite{gausterer1998complex,aarts2010complex, nagata2016argument, salcedo2016does}, but the precise reason for the biased results remained unclear until recently \cite{scherzer2019clandboundary, seiler2020complex}.
In general, the failure of the complex Langevin method is due to the lack of control of excursions away from the real axis. Many efforts have been made in the past decade to strict such excursions.
For instance, the use of adaptive time steps is studied in \cite{aarts2010adaptive} to avoid runaway trajectories; the method gauge cooling, which utilizes the gauge invariance to minimize the distances to the real axis, is proposed in \cite{seiler2013gauge} and has achieved many applications \cite{aarts2014simulating, Kogut2019applying, hirasawa2020complex}; the dynamical stabilization is introduced in \cite{Attanasio2019dynamical} and tested in \cite{attanasio2017}.
In the case where the method converges, the work \cite{scherzer2020controlling} proposes an approximation technique to quantify the bias.
Other attempts to improve the complex Langevin method include the coupling with Lefschetz thimbles \cite{nishimura2017combining}, the deformation technique \cite{nagata2018complex}, etc. We invite the readers to refer to \cite{berger2020complex, attanasio2020complex} for a comprehensive review of the recent advances.

In general, the complex Langevin method is still a numerical tool under construction. In this paper, we are going to carry out a deeper study of the aforementioned  dynamical stabilization.
The idea of dynamical stabilization is to add converging velocity fields to the complex Langevin equation to restrict the excursion of samples.
Instead of working on the original approach introduced in \cite{Attanasio2019dynamical}, we will investigate a slightly improved version considered in \cite{loheac2017third}, called the method of modified action.
Here we will follow \cite{scherzer2019clandboundary} and name this approach as the  ``regularization'' of the complex Langevin method.
Compared with the original approach,  this method is easier to be justified theoretically.
Our focus will be on possible modifications upon regularization. These include
coupling regularization with the reweighted complex Langevin method \cite{jacques2017reweight, bloch2017reweighted} and attempts to recover the unbiased result using regression.
Our study is to be carried out via a deep look into a motivating example in the $\operatorname{U}(1)$ one-link case, after which the method will be generalized to the $\SU{n}$ theories and applied to several lattice field theories, including the 3D XY model \cite{aarts2010convergence}, the Polyakov model \cite{Polyakov1978thermal}, and the heavy dense QCD (quantum chromodynamics) \cite{seiler2013gauge}.

The rest of the paper is organized as follows. In \Cref{sec:CL}, we briefly review the complex Langevin method and its general theory. In \Cref{sec:regU1}, the $2R$ method and $3R$ method are presented for a motivating one-dimensional example in the $\operatorname{U}(1)$ one-link case. Then these methods are generalized to multi-dimensional integrals in the $\operatorname{U}(1)$ theory in \Cref{sec:GenU1} and to the $\SU{n}$ theories in \Cref{sec:RegSUN}. In \Cref{sec:QCD}, the applications of these methods to several lattice filed theories are discussed.  Finally, the paper ends with some concluding remarks in \Cref{sec:con}.

\section{A review of the complex Langevin method} \label{sec:CL}

Following \cite{seiler2013gauge}, we use the notation $\{\cdot\}$ to denote the discrete field defined on a lattice.
For instance, suppose $\{\phi\}$ is a three-dimensional real scalar lattice field.
Then, $\{\phi\}$ contains a set of variables $\phi_x \in \mathbb{R}$ with $x$ being a three-dimensional multi-index representing the lattice point.
If the lattice has $N$ points, $\{\phi\}$ is essentially a vector in $\mathbb{R}^N$.
For simplicity, we will also use $\phi_k$, $k=1,\ldots,N$ to denote the components of $\{\phi\}$.
In this section, we will provide a brief introduction to the complex Langevin method and its regularization.
For introductory purposes, we will temporarily restrict ourselves to the scalar fields where $\phi_k \in \mathbb{R}$ or $\mathbb{T}$, where $\mathbb{T}$ stands for the torus $\mathbb{T} = \mathbb{R} \bmod 2\pi$.

With the notations defined above, we are interested in computing the following ensemble average:
\begin{equation}\label{eq:integral}
  \left\langle O \right\rangle = \frac{1}{Z} \int_{\Omega}
  O(\{\phi\})e^{-S(\{\phi\})} \,\mathrm{d}\{\phi\}, \quad Z = \int_{\Omega} e^{-S(\{\phi\})} \,\mathrm{d}\{\phi\},
\end{equation}
where $\Omega = \mathbb{R}^N$ or $\mathbb{T}^N$.
When $S(\{\phi\})$ is real, we can regard $Z$ as the partition function, so that the integral can be evaluated by the Langevin method \cite{parisi1981perturbation}. Specifically, the Langevin equation associated with \eqref{eq:integral} is given by
\begin{equation} \label{eq:real-Langevin}
\mathrm{d}\phi_k = K_k(\{\phi\}) \,\mathrm{d}t + \mathrm{d}w_k, \quad K_k = -\frac{\partial S}{\partial \phi_k}, \qquad k = 1,\ldots,N.
\end{equation}
Here $w_k$, $k = 1,\ldots,N$ are independent Wiener processes satisfying $\mathrm{d}w_k^2 = 2\mathrm{d}t$ for each $k$.
The Fokker-Planck equation of this stochastic process is
\begin{equation}\label{eq:real-FP}
\frac{\partial P}{\partial t} + \sum_{k=1}^N \frac{\partial}{\partial \phi_k} ( K_k P) =
  \sum_{k=1}^N \frac{\partial^2 P}{\partial \phi_k^2},
\end{equation}
where $P(\{\phi\},t)$ represents the probability distribution of the field $\{\phi\}$ at time $t$.
If $e^{-S({\phi})}$ is integrable and the stochastic process \eqref{eq:real-Langevin} is ergodic, then $P(\{\phi\},t)$ will converge to the equilibrium distribution $\frac{1}{Z} e^{-S(\{\phi\})}$ as $t \rightarrow \infty$. As a result, we can approximate \eqref{eq:integral} by
\begin{equation} \label{eq:approx_O}
\langle O \rangle \approx \frac{1}{N_{\mathrm{sample}}} \sum_{m=1}^{N_{\mathrm{sample}}} O(\{\Phi^{(m)}\}),
\end{equation}
where $\{\Phi^{(m)}\}$, $m = 1,\ldots, N_{\mathrm{sample}}$ are the samples generated by simulating the Langevin equation \eqref{eq:real-Langevin} and choosing $\Phi_k^{(m)} = \phi_k(T + m \Delta T)$ for a sufficiently large $T$ and sufficiently long time difference $\Delta T$.

However, when $S(\{\phi\})$ is complex, the Langevin method is no longer valid since $Z$ is not a partition function. To handle such complex actions, the complex Langevin method \cite{parisi1983complex, klauder1983stochastic} postulates stochastic equations of the same form as \eqref{eq:real-Langevin} with the trajectories of the process wandering in the complexified space $\Omega_{\mathbb{C}}$.
Here
\begin{equation}
\Omega_{\mathbb{C}} = \begin{cases}
  \mathbb{C}^N, & \text{if } \Omega = \mathbb{R}^N, \\
  (\mathbb{T} + \ii \mathbb{R})^N, & \text{if } \Omega = \mathbb{T}^N.
\end{cases}
\end{equation}
Now we assume that both $O(\cdot)$ and $S(\cdot)$ can be extended to $\Omega_{\mathbb{C}}$ holomorphically.
Thus, the stochastic process \eqref{eq:real-Langevin} is again well defined, and the complex Langevin method again approximates $\langle O \rangle$ using \eqref{eq:approx_O}. Since $K_k(\cdot)$ can take complex values, the field $\{\phi\}$ becomes a complex field, which can also be represented by two real fields $\{\phi^R\}$ and $\{\phi^I\}$ with $\phi_k = \phi_k^R + \ii \phi_k^I$.
The evolution of these two fields follows the \emph{complex Langevin equation}:
\begin{equation}\label{eq:complex-Langevin}
  \left\{
    \begin{array}{@{}lll}
      \mathrm{d}\phi_k^R = K_k^R(\{\phi^R\}, \{\phi^I\}) \,\mathrm{d}t + \mathrm{d}w_k,& K_k = -\re K_k, \\[6pt]
      \mathrm{d}\phi_k^I = K_k^I(\{\phi^R\},\{\phi^I\}) \,\mathrm{d}t,& K_k^I = -\im K_k,
    \end{array}
  \right.
\end{equation}
where $K_k$ is again the partial derivative of $S$ as defined in \eqref{eq:real-Langevin}.
The corresponding Fokker-Planck equation for the probability density function has the form $P(\{\phi^R\}, \{\phi^I\}, t)$, which evolves according to  
\begin{equation}\label{eq:complex-FP}
    \frac{\partial P}{\partial t} + \sum_{k=1}^N \left( \frac{\partial}{\partial \phi_k^R} (K_k^R P) +
    \frac{\partial}{\partial \phi_k^I} (K_k^I P) \right)
    = \sum_{k=1}^N \frac{\partial^2 P}{\partial (\phi_k^R)^2}.
\end{equation}

The correctness of the complex Langevin method requires the following two conditions:
\begin{itemize}
\item The stochastic process \eqref{eq:complex-Langevin} is ergodic.
\item The following equation holds:
\begin{equation} \label{eq:1d-justification}
\lim_{t \rightarrow \infty} \int_{\Omega_{\mathbb{C}}} O(\{\phi\}) P(\{\phi^R\}, \{\phi^I\}, t) \mathrm{d}\{\phi\} = \frac{1}{Z} \int_{\Omega} O(\{\phi\}) e^{-S(\{\phi\})} \mathrm{d}\{\phi\}.
\end{equation}
\end{itemize}
As mentioned in the introduction, observations indicate that when the stochastic process is ergodic, the complex Langevin method still produces wrong results, meaning that \eqref{eq:1d-justification} fails to hold. This occurs especially when $P(\cdot, \cdot, \infty)$ decays slowly, and the details have been studied in \cite{aarts2013localised, scherzer2019clandboundary, cai2021validity}. As a remedy, the method of \emph{dynamical stabilization} proposed in \cite{Attanasio2019dynamical} adds an artificial term to the imaginary part of the drift velocity $K_k$ to suppress the tail of $P$. Specifically, in \eqref{eq:complex-Langevin}, $K_k$ is chosen as
\begin{equation} \label{eq:sun-ds}
K_k(\{\phi\}) = -\partial_{\phi_k} S(\{\phi\}) - \ii \alpha_{DS} (\im \phi_k)^r, \qquad \forall k = 1,\cdots,N.
\end{equation}
where $r$ is an odd positive integer and $\alpha_{DS}$ is a positive parameter balancing the stabilizing effect and the bias introduced by this regularization. Since \eqref{eq:sun-ds} is no longer the derivative of an analytic function, the justification of this approach remains open.

In this paper, we will focus on another type of regularization introduced in \cite{scherzer2019clandboundary}, in which a specific problem is studied.
In the next section, we will conduct a deeper study of the regularization technique based on this motivating example and consider its possible extensions.

\section{Motivating Example - Regularization of complex Langevin for the   \texorpdfstring{$\operatorname{U}(1)$}{U(1)} one-link model}\label{sec:regU1}

We will motivate the use of regularization and dynamic stabilization by applying our proposed method on the one-dimensional $\operatorname{U}(1)$ one-link model studied in \cite{berges2008real, scherzer2019clandboundary}, where $\Omega = \mathbb{T}$.
Following \cite{scherzer2019clandboundary}, we use $x$ to denote the integral variable. The action is a $2\pi$-periodic function:
\begin{displaymath}
S(x) = \ii \beta \cos x,
\end{displaymath}
with $\beta \in \mathbb{R}^+$, and $x \in \mathbb{T}$.

\subsection{Regularization of complex Langevin}
Inspired from \cite{scherzer2019clandboundary}, the regularized action for this model is given by
\begin{equation}\label{U(1)3}
S_s(x) = S(x) + \frac{sx^2}{2} = \ii \beta \cos x + \frac{sx^2}{2}
\end{equation}
for any given $s > 0$. As discussed in \cite{scherzer2019clandboundary}, for large values of $s$, the value of $\langle O \rangle$ agrees with its corresponding true value with a modified action. However, a divergence is observed for values of $s$ close to $0.4$ as we attempt to set $s$ close to $0$ to retrieve the true value of $\langle O \rangle$ with unregularized action. The results are also reproduced in \Cref{fig:U(1)NumFailFig}. The authors have commented that the use of appropriate regression functions might have by extrapolating the results from $s > 0.4$ to obtain a decent estimate at $s = 0$. We will thus be following a similar argument while supplementing it with relevant regression functions with the appropriate mathematical justification. 

Note that due to the presence of the regularizing term, the periodicity for $S(x)$ in $x$ is destroyed.
Therefore, in the definition of the observable, we will ``unroll'' the torus $\mathbb{T}$ and change the integral domain to $\mathbb{R}$. Thus, under the modified action, we can rewrite equations \eqref{eq:integral} as
\begin{equation}\label{U(1)1}
\langle O \rangle_s = \frac{1}{Z_s} \int_{\mathbb{R}} O(x) \exp(-S_s(x)) \,\mathrm{d}x,
\end{equation}
with 
\begin{equation}\label{U(1)2}
Z_s = \int_{\mathbb{R}} \exp(-S_s(x)) \,\mathrm{d}x
\end{equation}

Upon complexification, we obtain the complex action
\begin{equation}\label{U(1)4}
S_s(z) = S(z) + \frac{sz^2}{2}, \quad z \in \mathbb{C}.
\end{equation}
The corresponding drift terms can be computed as follows:\footnote{Following the convention in \eqref{eq:real-Langevin} and in \eqref{eq:3dxy-drift}, we will represent $K_{x,s}$ as the drift term for the scalar field $\phi_x$ with regularized action. However, as in the one-dimensional case, as it is understood that we are only dealing with one field variable (which is written as $x$), we will drop the comma that separates $x$ and $s$ and simply write it as $K_s$. } 
\begin{equation}\label{U(1)5}
\begin{aligned}
K_{s}^R(x,y) &= - \re{(S_s'(x+ \ii y))} = -\beta \cos x \sinh y - sx, \\
K_{s}^I(x,y) &= - \im{(S_s'(x+ \ii y))} = \beta \sin x \cosh y - sy.
\end{aligned}
\end{equation}

As for the regularized action, two questions need to be answered:
\begin{itemize}
\item What is the relation between the regularized observable $\langle O \rangle_s$ and the original observable $\langle O \rangle$?
\item Can we apply the complex Langevin method to obtain the correct value of $\langle O \rangle_s$?
\end{itemize}
The following two sections will be devoted to the exploration of their answers.

\subsubsection{Correct convergence under regularized action}

Note that is natural to expect that the regularized observable will converge to the original observable as the regularizing parameter vanishes:
\begin{equation}\label{U(1)6}
\lim_{s\rightarrow 0^+} \langle O \rangle_s = \langle O \rangle.
\end{equation}
However, this is not immediately clear since we have changed the integration domain from $[0,2\pi)$ to $\mathbb{R}$ as we apply the regularization. Fortunately, the result above still holds under some mild conditions. To prove the limit \eqref{U(1)6}, we need the following lemma:
\begin{lemma} \label{LemmaIntegral}
For any $m, \beta \in \mathbb{R}$,
\begin{displaymath}
\int_{\mathbb{R}} e^{\ii mx}e^{-\ii \beta \cos x -s\frac{x^2}{2}} \mathrm{d}x = \sum_{n\in \mathbb{Z}} \ii^n J_n(-\beta) \sqrt{\frac{2\pi}{s}}e^{-\frac{(m+n)^2}{2s}},
\end{displaymath}
where $J_n$ denotes the Bessel function of the first kind.
\end{lemma}

\begin{proof}
Applying the Jacobi-Anger expansion 
\begin{equation}\label{U(1)9}
e^{- \ii z\cos\phi} = \sum_{n \in \mathbb{Z}} \ii^n J_n(-z) e^{\ii n\phi} \text{ for all } \phi \in \mathbb{R}.
\end{equation}
to \eqref{U(1)7}, we have
\begin{equation}\label{U(1)10}
\int_{\mathbb{R}} e^{\ii mx}e^{-\ii \beta \cos x -s\frac{x^2}{2}} \mathrm{d}x
= \sum_{n\in \mathbb{Z}} \ii^n J_n(-\beta) \int_{\mathbb{R}} e^{\ii (n+m)x-s\frac{x^2}{2}} \mathrm{d}x
= \sum_{n\in \mathbb{Z}} \ii^n J_n(-\beta) \sqrt{\frac{2\pi}{s}}e^{-\frac{(m+n)^2}{2s}}.
\end{equation}
Here, we have interchanged the infinite sum and the integral, which can be justified using Dominated Convergence Theorem by computing
\begin{equation}\label{U(1)11}
\begin{aligned}
\left| \sum_{n = -N}^{N} \ii^n J_n(-\beta)e^{\ii (n+m)x-s\frac{x^2}{2}} \right| &= e^{-s\frac{x^2}{2}}\left| \sum_{n=-N}^{N} \ii^n J_n(-\beta)e^{\ii nx} \right| \leq 4e^{-s\frac{x^2}{2}}.
\end{aligned}
\end{equation}
for a given $N \in \mathbb{Z}^+$ that is large enough, as the finite sum inside the absolute sign tends to $e^{-\ii \beta\cos x}$ with modulus $1$ if \eqref{U(1)9} is applied. The resulting upper bound in \eqref{U(1)11} is clearly integrable on $\mathbb{R}$ for a fixed $\beta$ and $s$.
\end{proof}

\begin{proposition}\label{U(1)Prop1}
For the $\operatorname{U}(1)$ one-link model, suppose the observable $O(x)$ is $2\pi-$periodic and absolutely continuous on $[0,2\pi)$. In addition, if we demand that $O$ is a $(1+\alpha)$-H\"{o}lder class function for some $\alpha > 0$, then, we have that \eqref{U(1)6} holds.
\end{proposition}

\begin{proof}
First, we consider the Fourier series expansion of $O(x)$ given by
\begin{equation}\label{U(1)7}
O(x) = \sum_{m \in \mathbb{Z}} \widehat{O}_m e^{\ii mx}.
\end{equation}
Furthermore, from a standard result in Harmonic Analysis, we know that the convergence of the infinite series on the right hand side of \eqref{U(1)5} is uniform, which thus implies that the Fourier series on the right can be used to represent $O$. From here, we apply \Cref{LemmaIntegral}:
\begin{gather}
\label{eq:O_Ss}
\int_{\mathbb{R}} O(x) \exp(-S_s(x)) \mathrm{d}x =
\sum_{m\in \mathbb{Z}} \widehat{O}_m \int_{\mathbb{R}} e^{\ii m x} \exp(-S_s(x)) \mathrm{d}x = 
\sum_{m\in \mathbb{Z}} \sum_{n\in \mathbb{Z}}
\ii^n \widehat{O}_m J_n(-\beta) \cdot \sqrt{\frac{2\pi}{s}} e^{-\frac{(m+n)^2}{2s}}, \\
\begin{aligned}
\label{eq:O_S}
\int_{\mathbb{T}} O(x) \exp(-S(x)) \mathrm{d}x &=
\sum_{m\in \mathbb{Z}} \widehat{O}_m \int_0^{2\pi} e^{\ii m x} e^{-\ii \beta \cos x} \mathrm{d}x =
\sum_{m\in \mathbb{Z}} \widehat{O}_m \int_0^{2\pi} e^{\ii m x} \sum_{n \in \mathbb{Z}}\ii^n J_n(-\beta)e^{\ii n x} \mathrm{d}x \\
&= 2\pi \sum_{m\in \mathbb{Z}} \sum_{n\in \mathbb{Z}} \ii^n \widehat{O}_m J_n(-\beta) \delta_{n,-m} = 
2\pi \sum_{m\in \mathbb{Z}} \ii^m \widehat{O}_{-m} J_m(-\beta),
\end{aligned}
\end{gather}
where the last equality of \eqref{eq:O_Ss} is due to \Cref{LemmaIntegral} and \eqref{eq:O_S} utilizes \eqref{U(1)9} and the property that $J_n(x) = J_{-n}(x)$ for all integer $n$. Here, we note that in \eqref{eq:O_Ss} and \eqref{eq:O_S}, we have swapped the relevant infinite series and integration. This can be justified using the Dominated Convergence Theorem by considering the following partial sums for any $N \in \mathbb{N}$:
\begin{equation}\label{U(1)43}
\begin{aligned}
\left| \sum_{m = - N}^N \widehat{O}_m  e^{\ii m x} \exp(-S_s(x)) \right| &= \left| \sum_{m = - N}^N \widehat{O}_m  e^{\ii m x - \ii \beta \cos(x) - s\frac{x^2}{2}} \right| \\
&\leq 2\sum_{m = 1}^N \frac{K}{|m|^{1+\alpha}} e^{-s\frac{x^2}{2}} +  |\widehat{O}_0|e^{-s\frac{x^2}{2}}
\leq e^{-s\frac{x^2}{2}}(2K \zeta(1+\alpha) + |\widehat{O}_0|),
\end{aligned}
\end{equation}
where $\zeta(\cdot)$ is the Riemann zeta function, and we have used the assumption that $O$ is a $(1+\alpha)$-H\"{o}lder class function for some $\alpha > 0$. This means that there exists a constant $K$ such that
\begin{equation} \label{Holder}
|\widehat{O}_m| \leq \frac{K}{|m|^{1+\alpha}}.
\end{equation}
Thus, from \eqref{U(1)43}, we can see that the upper bound is clearly integrable on $\mathbb{R}$. Thus, by Dominated Convergence Theorem, the aforementioned interchange is justified.

The equation \eqref{eq:O_Ss} implies that
\begin{equation}\label{U(1)12}
\lim_{s\rightarrow 0^+ }\sqrt{\frac{s}{2\pi}}\int_{\mathbb{R}} O(x) \exp(-S_s(x)) \mathrm{d}x
= \sum_{n \in \mathbb{Z}} \ii^n \widehat{O}_{-n} J_n(-\beta).
\end{equation}
For $Z_s$ and $Z$, one can use the same technique to deduce that
\begin{equation} \label{eq:Zs_Z}
\lim_{s\rightarrow 0^+} \sqrt{\frac{s}{2\pi}} Z_s = J_0(-\beta), \qquad
Z = 2\pi J_0(-\beta).
\end{equation}
It is now clear from \eqref{eq:O_S}, \eqref{U(1)12}, and \eqref{eq:Zs_Z} that
\begin{equation}\label{U(1)13}
\lim_{s\rightarrow 0^+} \langle O \rangle_s
= \frac{\displaystyle \lim_{s\rightarrow 0^+}\sum_{m\in\mathbb{Z}} \widehat{O}_m \sqrt{\frac{s}{2\pi}}\int_{\mathbb{R}} e^{\ii mx}e^{-\ii \beta \cos x -s\frac{x^2}{2}} \mathrm{d}x}{\displaystyle \lim_{s\rightarrow 0^+}\sqrt{\frac{s}{2\pi}} Z_s}
= \frac{1}{J_0(-\beta)} \sum_{m\in \mathbb{Z}}\widehat{O}_{-m} \ii^m J_m(-\beta) = \langle O \rangle,
\end{equation}
which concludes the proof.
\end{proof}

The result above justifies the regularization of the action - if $s$ is chosen small and $\langle O \rangle_s$ can be correctly computed by the complex Langevin method, the value $\langle O \rangle_s$ can be regarded as an approximation of $\langle O \rangle$. 

\subsubsection{Correct numerical convergence for complex Langevin method}
Despite the guarantee for correct convergence given in Proposition \ref{U(1)Prop1}, numerical results from the complex Langevin method suggest otherwise. The numerical experiments on the regularized action have been carried out in \cite{scherzer2019clandboundary}, and we have repeated the same experiments for $\beta = 0.5$. The results are plotted in Figure \ref{fig:U(1)NumFailFig} for $O(x) = e^{\ii x}$, where we can observe a divergence between the true values represented by the red curve and the numerical results represented by the data points. The data points are obtained via numerical simulations with a fixed time step of $\Delta t = 3\times10^{-4}$ for values of $s$ closer to $0$ and $\Delta t = 1\times10^{-3}$ if otherwise, with each sample obtained after every $2000$ steps for a total of $10^6$ samples for each value of $s \in \{0.05k \hspace{3pt}| \hspace{3pt} 0 \leq k \leq 30, k \in \mathbb{Z} \}.$
\begin{figure}[htbp]
\centering
\includegraphics[width=0.5\textwidth]{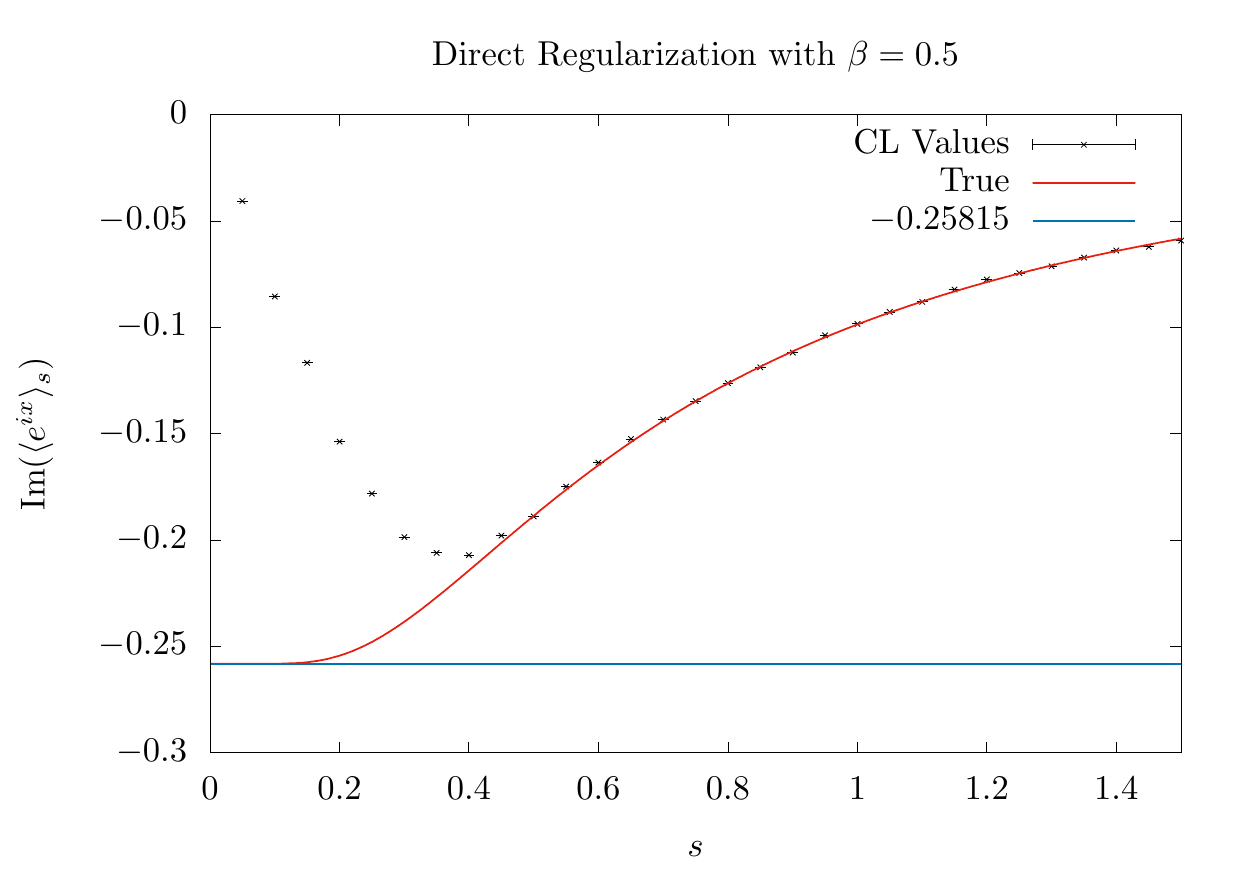}
\caption{The graph depicting the numerical results and true values of $\im \ml \langle e^{ix} \rangle_s \mr$ against $s$ for $\beta = 0.5$.}
\label{fig:U(1)NumFailFig}
\end{figure}

In view of Proposition \ref{U(1)Prop1}, we can deduce that such divergence between the true values and the numerical results must be due to the corresponding complexification of the Langevin dynamics. It is thus instructive to investigate the correct values of $s$ in which the numerical results from our complex Langevin method agree with that from the original Langevin dynamics.
The phenomenon that the correctness of complex Langevin changes with the parameter in the action has been observed and explained in a number of previous works \cite{aarts2013localised,nagata2016argument,cai2021validity}. In \cite{cai2021validity}, it is demonstrated in another example that the correctness of the complex Langevin results can be guaranteed only when the probability density function is localized, meaning that for all $t > 0$, the solution of \eqref{eq:complex-Langevin} always satisfies $y(t) \in [Y^-, Y^+]$ for some $Y^- < Y^+$.
Our problem has a close similarity to the example in \cite{cai2021validity}, and it can be expected that we also require the localization of the $y(t)$ to guarantee the correctness of complex Langevin.
To confine the value of $y(t)$, we need that the imaginary velocity $K^I$ to satisfy $K_{s}^I(x,Y^-) > 0$ and $K_{s}^I(x,Y^+) < 0$ for all $x$. Note that the choice of $0$ here is due to the fact that in all simulations of the complex Langevin dynamics, we will always set the initial coordinates to be at the origin. This thus motivates the following proposition:

\begin{proposition}\label{U(1)Prop2}
For the $\operatorname{U}(1)$ one-link model, given a fixed $s > 0$ and $\beta > 0$, if $s > 1.509\beta$, then there exist $Y^+ > 0$ and $Y^- < 0$ such that
\begin{equation}\label{U(1)17}
K_{s}^I(x,Y^+) < 0 \text{ and } K_{s}^I(x,Y^-) > 0.
\end{equation}
\end{proposition}

\begin{proof}
We first consider the case for $K_s^I(x,Y^+) < 0$. For any $(x,y) \in \mathbb{R}^2$, using the expression from \eqref{U(1)5}, we are looking to solve the following inequality
\begin{equation}\label{U(1)18}
K_{s}^I(x,y) =  \beta \sin(x) \cosh(y) - sy < 0
\end{equation}
in the sense that there exists a $y = Y^+ > 0$ such that for all $x \in \mathbb{R}$, $K_{s}^I(x,Y^+) < 0$.

First, for this to hold for all $x$, it must thus hold at a point in which $\sin x$ is maximum, that is, takes the value of $1$, as $\cosh(y) > 0$ for all $y \in \mathbb{R}$. We define the new expression of $K_{s}^I$ in which we replace $\sin(x)$ by 1 as $\bar{K}^I$. Thus, we are looking to solve for a region in the parameter space ($\beta,s$) such that such a $Y^+$ would be guaranteed. The strategy is as follows. First, we fix the parameters $\beta$ and $s$ and solve for the minimum value of this function $\bar{K}^I$ at $y_0$ in terms of $\beta$ and $s$. Since this minimum value is a function of $\beta$ and $s$, we can in fact find such a region in the parameter space such that $\bar{K}^I(y_0) < 0$. Thus, since $\bar{K}^I$ is minimized at $y_0$ and is negative, we then have for all $y \in [0,Y^+]$ with $Y^+ = y_0$ that $\bar{K}^I(y) < 0$ and thus $K^I_s(x,y) < 0$ for all $x \in \mathbb{R}$ and $y \in [0,Y^+]$. Therefore, we have $y_0$ as the required $Y^+$ that we are looking for.
To apply this strategy, we first look at the corresponding function for $\bar{K}^I$:
\begin{equation}\label{U(1)19}
 \bar{K}^I(y) = \beta \cosh(y) - s y.
\end{equation}
Using standard one-variable optimization techniques, we see that global minimum is attained at 
\begin{equation}\label{U(1)20}
y_0 = \sinh^{-1}\ml \frac{s}{\beta}\mr.
\end{equation}
Now we demand that the minimum value of $\bar{K}^I$ be negative:
\begin{equation}\label{Klessthan0}
\beta \cosh(y_0) - s y_0 < 0.
\end{equation}
Inserting \eqref{U(1)20} into the equation above and letting $\chi = s/\beta$, we can simplify the inequality \eqref{Klessthan0} to
\begin{equation}\label{U(1)23}
e^{\sqrt{1 + \ml \frac{1}{\chi} \mr^2}} - \chi - \sqrt{1+\chi^2} < 0
\end{equation}
which can be solved numerically to obtain:
\begin{equation}
\label{U(1)24}
\chi > 1.509, \quad s > 1.509 \beta,
\end{equation}
and the proof is thus complete for this case.

For the other case in \eqref{U(1)17}, we can use the same strategy to obtain the same sufficient condition $s > 1.509 \beta$, which completes the proof of the proposition.
\end{proof}

Indeed, as we can see from Figure \ref{fig:U(1)NumFailFig}, for points after $s = 0.8 \approx 1.6\beta > 1.509\beta$, we can observe that the true values are coherent with the numerical values obtained. For $s \in (0.5, 0.7)$, although the numerical results from complex Langevin appear to be on the red curve, we believed that a small systematic bias has occurred.

In view of Propositions \ref{U(1)Prop1} and \ref{U(1)Prop2}, it seems unlikely that we can obtain good numerical values of $\langle O \rangle$ solely with the use of a regularized action, as seen in \Cref{fig:U(1)NumFailFig}. This thus motivates the following subsection, in which we will consider the fix of the regularized values.

\subsection{Reweighted complex Langevin method with regularized action}
In this subsection, we will introduce the reweighted complex Langevin method aimed at obtaining numerical results for $\langle O \rangle_s$.

In \cite{jacques2017reweight}, the authors consider the action $S_{\xi}$ with a parameter $\xi$. It then holds for any $\xi$ and $\xi_0$ that
\begin{equation} \label{eq:S_xi}
\frac{\displaystyle\int O(x) S_{\xi}(x) \mathrm{d}x}{\displaystyle\int S_{\xi}(x) \mathrm{d}x} = \frac{\displaystyle \int \frac{O(x) S_{\xi}(x)}{S_{\xi_0}(x)} S_{\xi_0}(x) \mathrm{d}x}{\displaystyle \int \frac{S_{\xi}(x)}{S_{\xi_0}(x)} S_{\xi_0}(x) \mathrm{d}x}.
\end{equation}
Both the numerator and the denominator on the right-hand side can be approximated using the complex Langevin method with action $S_{\xi_0}(x)$. By choosing an appropriate $\xi_0$, one may get a better approximation of $\langle O \rangle$ as compared to applying the complex Langevin method directly to the left-hand side of \eqref{eq:S_xi}. In our case, the regularized action includes a regularizing parameter $s$, in which we know that the true value could be generated at $s = 0$. This inspires us to develop our algorithm according to the following proposition:

\begin{proposition}\label{U(1)Prop3}
The following equality holds for all $s,s_0 \geq 0$:
\begin{equation}\label{U(1)29}
 \langle O(x) \rangle_s = \frac{\left\langle O (x) \exp\ml \frac{(s_0-s)x^2}{2} \mr \right\rangle_{s_0}}{\left\langle \exp\ml \frac{(s_0-s)x^2}{2} \mr \right\rangle_{s_0}}
\end{equation}
\end{proposition}

This proposition is a direct result of \eqref{eq:S_xi} by setting $\xi$ to be $s$.
Thus, from Proposition \ref{U(1)Prop2}, as long as we pick $s_0 > 1.509\beta$ and $s > 0$, we are guaranteed that the numerical values of two integrals in the ratio obtained using the complex Langevin method for the right hand side of \ref{U(1)Prop3} have no biases. By equality \eqref{U(1)29}, we can thus obtain an accurate numerical value of $\langle O \rangle_s$ even for $s < s_0$. Setting $s \rightarrow 0$ in \eqref{U(1)29}, we thus have an accurate numerical value of $\langle O \rangle$.

Following \Cref{U(1)Prop3}, we carry out numerical experiments by fixing $s_0 = 0.8$ and compute $\langle O \rangle_s$ for $s \in (0, 1.5)$. The numerical values were generated using a fixed time step of $\Delta t = 10^{-3}$, with each sample obtained after every $2000$ steps for a total of $10^7$ samples at $s_0 = 0.8$. Values of $\langle O \rangle_s$ for $s \neq 0.8$ were obtained from this set of points generated via the equation \eqref{U(1)29} above. Furthermore, the corresponding error bars were generated using a $M$ out of $N$ naive bootstrap method at each $s$, with $M = 20000$ and $N = 10^7$, repeated for $n = 10000$ times. The results and the estimated error bars are plotted in \Cref{fig:U(1)RewDivFig}. Indeed, we observe that for $s$ around $0.5$, we have obtained a better approximation of $\langle O \rangle$. However, two worrying phenomena have also surfaced from this experiment. Namely,
\begin{itemize}
    \item The numerical value of $\langle O \rangle_s$ deviates from the true value when $s$ gets smaller than $0.5$.
    \item As $s$ reduces, the estimated standard error start to grow dramatically from $s$ around $0.4$.
\end{itemize}

\begin{figure}[htbp]
\centering
\includegraphics[width=0.5\textwidth]{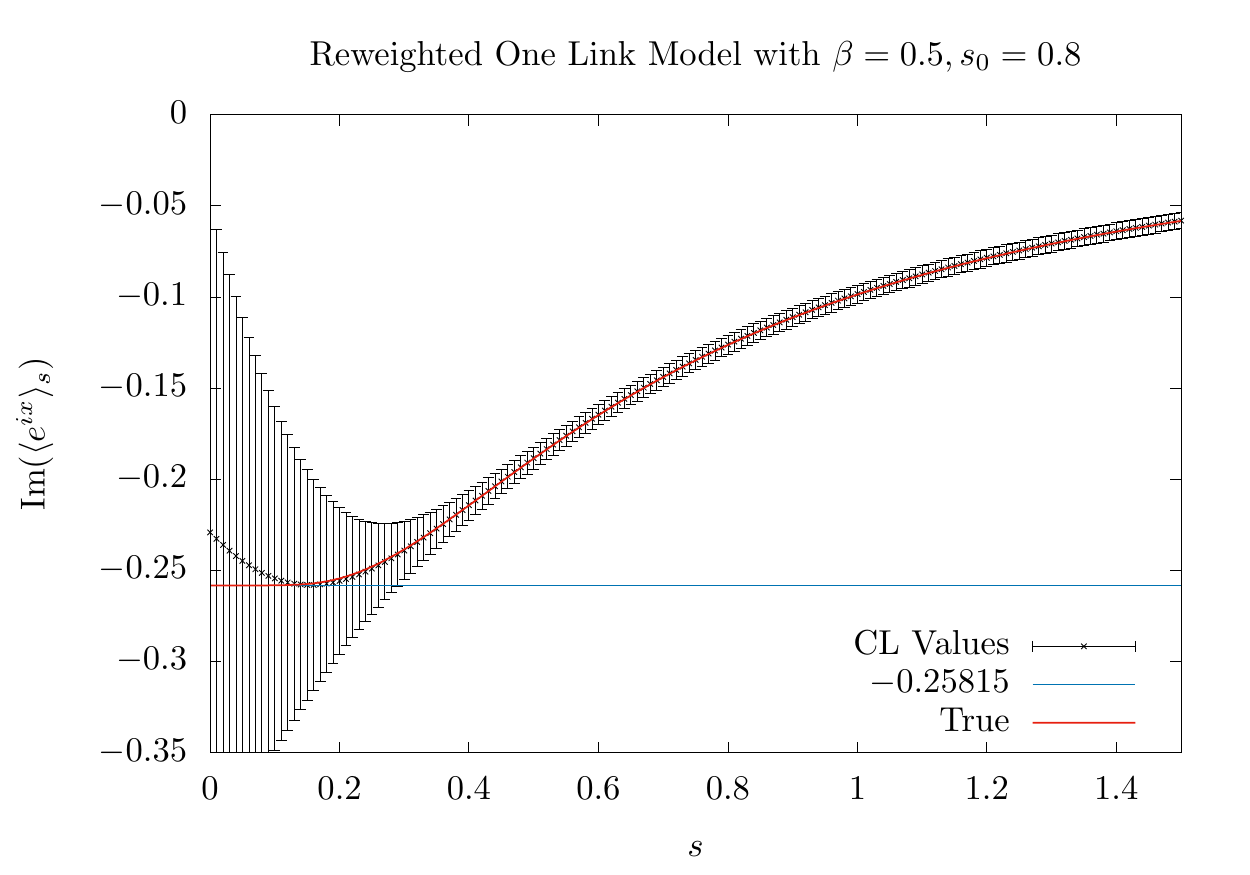}
\caption{The graph depicting the divergence of numerical results obtained from reweighted complex Langevin method and true values of $\im \ml \langle e^{\ii x} \rangle_s \mr$ against $s$ for $\beta = 0.5$ due to large standard errors. $-0.25815$ represents the true value at $s = 0$.}
\label{fig:U(1)RewDivFig}
\end{figure}
Nonetheless, it can be shown that the divergence for $\langle O \rangle$ is due to a large standard error, in which the standard error for $\langle O \rangle_s$ grows as the value of $s$ decreases from $s_0$ to $0$. This thus provides motivation for the following section, in which the introduction of a mathematically-motivated regression model aims to obtain an improved numerical estimate for $\langle O \rangle$.

\subsection{Coupling reweighted complex Langevin method with regularized action, with regression}
As mentioned at the start of this section, an important question to address would be the choice of the regressors that we should use to perform regression. Will a simple polynomial regression work? What would be considered as appropriate regressors? To answer these questions, we refer back to \Cref{fig:U(1)NumFailFig}. The graph above shows the graph of the true curve of $\langle O \rangle_s$ for $O = e^{\ii x}$ with $\beta = 0.5$ in red. As observed, the curve becomes very flat when $s$ is close to zero, which implies that the higher-order derivatives of $\im{\ml \langle e^{\ii x} \rangle_s \mr}$ might be $0$ at $s = 0$. This is not a fact captured by arbitrary polynomial regressors. Thus, if such an observation is true, we would have to turn to other regressors. This motivates the proposition below.

\begin{proposition}\label{U(1)Prop4}
For the $\operatorname{U}(1)$ one-link model with regularized action, for any observable $O$ satisfying the conditions in Proposition \ref{U(1)Prop1} and for any given $k \in \mathbb{Z}^+$, we have 
\begin{equation}\label{U(1)32}
\frac{\mathrm{d}^k}{\mathrm{d}s^k} \langle O \rangle_s|_{s = 0} = 0.
\end{equation}
\end{proposition}

\begin{proof}
The proof of this proposition continues from the proof of Proposition \ref{U(1)Prop1}. From \eqref{U(1)10}, the numerator for $\langle O \rangle_s$ constitutes a sum over $m$ of the expression in \eqref{U(1)10}. The denominator however, consists of the $m=0$ term in \eqref{U(1)10}. Multiplying both the numerator and the denominator by a factor of $\sqrt{\frac{s}{2\pi}}$, we have that
\begin{equation}\label{U(1)33}
\langle O \rangle_s = \frac{\sum_{m\in \mathbb{Z}}\sum_{n \in \mathbb{Z}} \ii^n J_n(-\beta) \widehat{O}_m e^{-\frac{(m+n)^2}{2s}} }{\sum_{n \in \mathbb{Z}} \ii^n J_n(-\beta) e^{-\frac{n^2}{2s}}} := \frac{f_1(s)}{f_2(s)}.
\end{equation}
As the given function above is clearly infinitely differentiable, we take the derivative with respect $s$ on both sides to obtain
\begin{equation}\label{U(1)34}
\frac{\mathrm{d}}{\mathrm{d}s}\langle O \rangle_s = \frac{f_1'(s)f_2(s) - f_2'(s)f_1(s)}{(f_2(s))^2}
\end{equation}
where
\begin{equation}\label{U(1)35}
\begin{aligned}
f_1'(s) & = \sum_{(n + m)\in \mathbb{Z}\setminus\{0\}} \ii^n J_n(-\beta) \widehat{O}_m \ml \frac{(m+n)^2}{s^2}\mr e^{-\frac{(m+n)^2}{2s}}, \\
f_2'(s) &= \sum_{n\in\mathbb{Z}\setminus\{0\}} \ii^n J_n(-\beta) \ml \frac{n^2}{s^2}\mr e^{-\frac{n^2}{2s}}.
\end{aligned}
\end{equation}
Here, we used $\mathbb{Z} \setminus \{0\}$ since if $n$ or $n + m$ is equals to $0$ before differentiating, the corresponding term in the infinite series is a constant due to the absence of the exponential factor and disappears upon differentiation. By writing down \eqref{U(1)35}, we have explicitly swapped the derivative and the infinite sum. This can be justified using a standard result in analysis (see Theorem 7.17 in \cite{rudin}) as follows. First, we restrict our attention to $[0,s_r]$ for $s_r$ large enough.\footnote{Large enough can be understood in the sense that it is sufficient for our numerical simulations and that $s_r > s_0$} Then, we will proceed to show that the derivative of the sequence of partial sums converges uniformly on $[0,s_r]$. Below, we will verify the conditions for $f_1'(s)$, in which a simpler case will thus hold for $f_2'(s)$. The uniform convergence can be verified using Weierstrass M-test by first computing
\begin{equation}\label{U(1)40}
\begin{aligned}
&\sum_{(n+m) \in \mathbb{Z}\backslash \{0\}} \left| \ii^n J_n(-\beta) \widehat{O}_m \ml \frac{(m+n)^2}{s^2}\mr e^{-\frac{(m+n)^2}{2s}} \right| \\
\leq{} & \sum_{(m+n) \in \mathbb{Z}\backslash \{0\}} |J_n(-\beta)|\ml (1 - \delta_{0m}) \frac{K}{|m|^{1+\alpha}} +  \delta_{0m} |\widehat{O}_0|\mr \frac{4}{(m+n)^2}\ml\frac{(m+n)^2}{2s}\mr^2 e^{-\frac{(m+n)^2}{2s}} \\
\leq{} & 4\left(\sum_{n\in \mathbb{Z}} |J_n(-\beta)|\right) \sum_{m\in \mathbb{Z}} \ml (1 - \delta_{0m})\frac{K}{|m|^{1+\alpha}} +  \delta_{0m} |\widehat{O}_0|\mr
= M(\beta) [2K\zeta(1+\alpha) + |\widehat{O}_0|]< +\infty,
\end{aligned}
\end{equation}
where we have used the following facts:
\begin{itemize}
    \item $O$ is a $(1+\alpha)$-H\"{o}lder class function for some $\alpha > 0$, for Fourier coefficients $\widehat{O}_n$, where $n \neq 0$, and that $|\widehat{O}_0|$ is bounded. These two cases are separated by using the Kronecker delta symbol $\delta_{0m}$.
    \item Since $m+n\neq 0$, then we have $\frac{1}{(m+n)^2} \leq 1$ for all $m,n \in \mathbb{Z}$.
    \item $x^2 e^{-x} \leq \frac{4}{e^2} \leq 1$ for all $x \geq 0$. 
    \item $M(\beta) := \sum_{n=-\infty}^\infty |J_n(-\beta)| < + \infty$ for any $\beta > 0.$\footnote{This can be observed from its asymptotic behaviour for large $|n|$, such that for a fixed $\beta$, $|J_n(-\beta)| \sim \frac{1}{\Gamma(|n|+1)}(\frac{|\beta|}{2})^{|n|}$.} 
\end{itemize}
With \eqref{U(1)40}, the aforementioned interchange in \eqref{U(1)35} is justified. Furthermore, since  $\lim_{s\rightarrow 0^+} f_1'(s) = 0$ and $\lim_{s\rightarrow 0^+} f_2'(s) = 0$, and $\lim_{s\rightarrow 0^+} f_1(s) < + \infty$ and $\lim_{s\rightarrow 0^+} f_2(s) < + \infty$, then we have $\frac{\mathrm{d}}{\mathrm{d}s}\langle O \rangle_s |_{s=0} = 0.$

Now, assume that \eqref{U(1)32} has been proven for all $k = 1,\cdots,K$ for some $K > 0$. By the general Leibniz rule,
\begin{equation}\label{eq:higher_derivative}
  f_1^{(K+1)}(s) = \sum_{k=0}^{K+1} f_2^{(K+1-k)}(s) \frac{\mathrm{d}^k}{\mathrm{d}^k s}\langle O \rangle_s.
\end{equation}
Using a similar logic as in \eqref{U(1)35}, we can write down the higher order derivatives of $f_1$ and $f_2$ below:
\begin{equation}\label{U(1)37}
\begin{aligned}
f_1^{(q)}(s) & = \sum_{(n+m)\in \mathbb{Z}\setminus\{0\}} \ii^n J_n(-\beta) \hat O(m) h_{1,q}(s) e^{-\frac{(m+n)^2}{2s}} \text{ and }\\
f_2^{(r)}(s) & = \sum_{n\in\mathbb{Z}\setminus\{0\}} \ii^n J_n(-\beta) h_{2,r}(s) e^{-\frac{n^2}{2s}},
\end{aligned}
\end{equation}
where both $h_{1,k}(s)$ and $h_{2,k}(s)$ refer to polynomials in $\frac{1}{s}$ of degree $2k$. Note that the interchanges between the infinite sums and the $q$ and $r$-th order derivatives are still justified. This is because similar to \eqref{U(1)40}, the presence of $e^{-\frac{(m+n)^2}{2s}}$ will always be able to overwhelm any polynomials in $\frac{1}{s}$ and create $\frac{1}{(m+n)^\gamma}$ for $\gamma$ sufficiently large.
Taking the limit $s \rightarrow 0^+$ on both sides of \eqref{eq:higher_derivative}, we obtain
\begin{equation}
0 = f_2(0) \lim_{s\rightarrow 0^+} \frac{\mathrm{d}^{K+1}}{\mathrm{d}^{K+1} s}\langle O \rangle_s.
\end{equation}
Thus \eqref{U(1)32} holds for $k = K+1$ since $f_2(0) \neq 0$. By the principle of mathematical induction, \eqref{U(1)32} holds for all positive integer $k$.
\end{proof}

Upon acknowledging the information presented in Proposition \ref{U(1)Prop4}, we can investigate the structure of $\langle O \rangle_s$ presented in \eqref{U(1)33}. From here, we can consider regressors in the form of a ratio of sum of exponential functions as summarised below:

\begin{proposition}\label{U(1)Prop5}
For any observable $O$ satisfying the conditions as stated in Proposition \ref{U(1)Prop1}, an appropriate rational representation would be:
\begin{equation}\label{U(1)38}
\begin{aligned}
    \langle O \rangle_s &= \frac{\sum_{k=0}^\infty a_k e^{-\frac{k^2}{2s}}}{\sum_{k=0}^\infty b_k e^{-\frac{k^2}{2s}}}
\end{aligned}
\end{equation}
\end{proposition}

\begin{proof}
This follows directly by considering all possible integer combinations of both the numerator and the denominator in \eqref{U(1)33}.
\end{proof}

Henceforth, we can use the rational representation in \eqref{U(1)38} and consider the following regression model:

\begin{equation}\label{U(1)39}
\langle O \rangle_s = \frac{\sum_{k=0}^M a_k e^{-\frac{k^2}{2s}}}{1 + \sum_{k=1}^M b_k e^{-\frac{k^2}{2s}}} := \frac{A_M(s)}{B_M(s)}.
\end{equation}
In the regression model above, the coefficients $a_k$ and $b_k$ are obtained by minimizing the objective function:
\begin{equation}\label{eq:lossfunction}
\argmin_{a_k,b_k} \int_{s_{\min}}^{s_{\max}} |A_M(s) - f(s)B_M(s)|^2 \mathrm{d}s,
\end{equation}
where $s_{\min}$ and $s_{\max}$ are the lower and upper bounds of $s$ for which we can obtain the value of $\langle O \rangle_s$ via simulation, and $f(s)$ refers to a certain approximation of $\langle O \rangle_s$ for $s \in [s_{\min}, s_{\max}]$. In our experiments below, $f(s)$ is chosen as a polynomial of $s$ and is obtained via least squares approximation.

\subsection{Numerical Results}

In this subsection, we will attempt to include simulations and regressions conducted for $O(x) = e^{\ii x}$ with $\beta = 0.5$. Here, we note that from \cite{scherzer2019clandboundary}, the exact value of $\langle e^{\ii x} \rangle$ at $\beta = 0.5$ is given by 
$$ \langle e^{\ii x} \rangle = \frac{I_1(-0.5\ii)}{I_0(-0.5\ii)} = 0  -0.25815\ii.$$

First, we will attempt to obtain an estimate for $\langle e^{\ii x} \rangle$ through the use of regularization and regression. Note that we have shown that for $s\geq 0.8$, the numerical values obtained from the complex Langevin method are accurate. In addition, as inspired from Figure \ref{fig:U(1)NumFailFig}, we start to see a divergence between the true curve and the complex Langevin values at $s \approx 0.4$. Thus, we will attempt to include the complex Langevin values for two different cases, $s \geq 0.4$ and $s \geq 0.8$, and apply the regression model in Proposition \ref{U(1)Prop5} with different values of $M$ for each case. The numerical values were generated using a fixed time step of $\Delta t = 10^{-4}$, with each sample obtained after every $2000$ steps for a total of $10^6$ samples for each value of $s \in \{0.01k \hspace{3pt} | \hspace{3pt} 40 \leq k \leq 150, k \in \mathbb{Z} \}$. Due to fluctuations present in the raw data set, we have employed an interpolation using a quartic polynomial in $s$ to average out the fluctuations prior to solving the optimization problem \eqref{eq:lossfunction}. Here, we note that this is consistent with the original regression model in Proposition \ref{U(1)Prop5}, in which for $s$ not close to $0$, we do not have an issue with a flat curve as $s$ approaches $0$ and can therefore approximate such as expression with an appropriate polynomial. The relevant data points and regression curves are summarized in \Cref{fig:U(1)DirectCurves}.

\begin{figure}[htbp]
\centering
\includegraphics[width=0.48\textwidth]{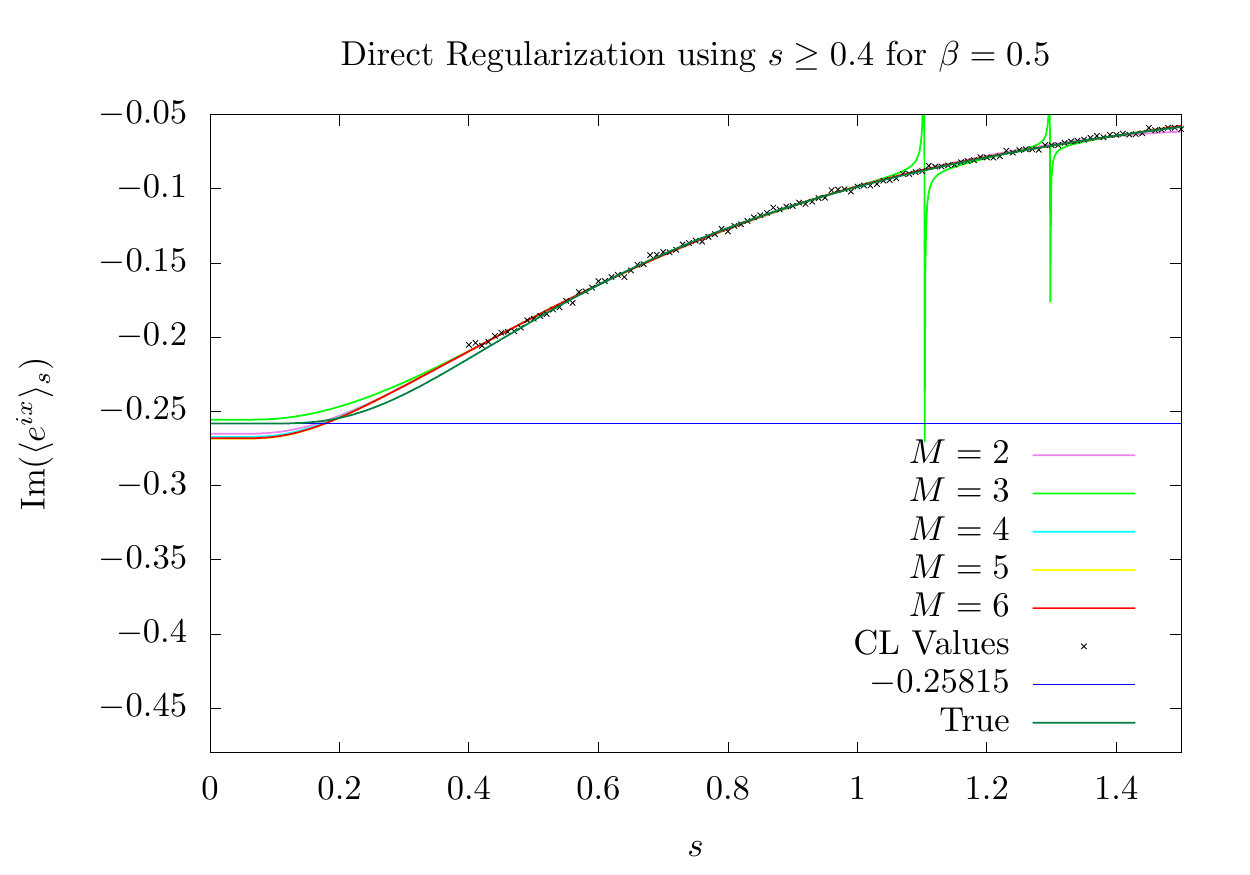}
\includegraphics[width=0.48\textwidth]{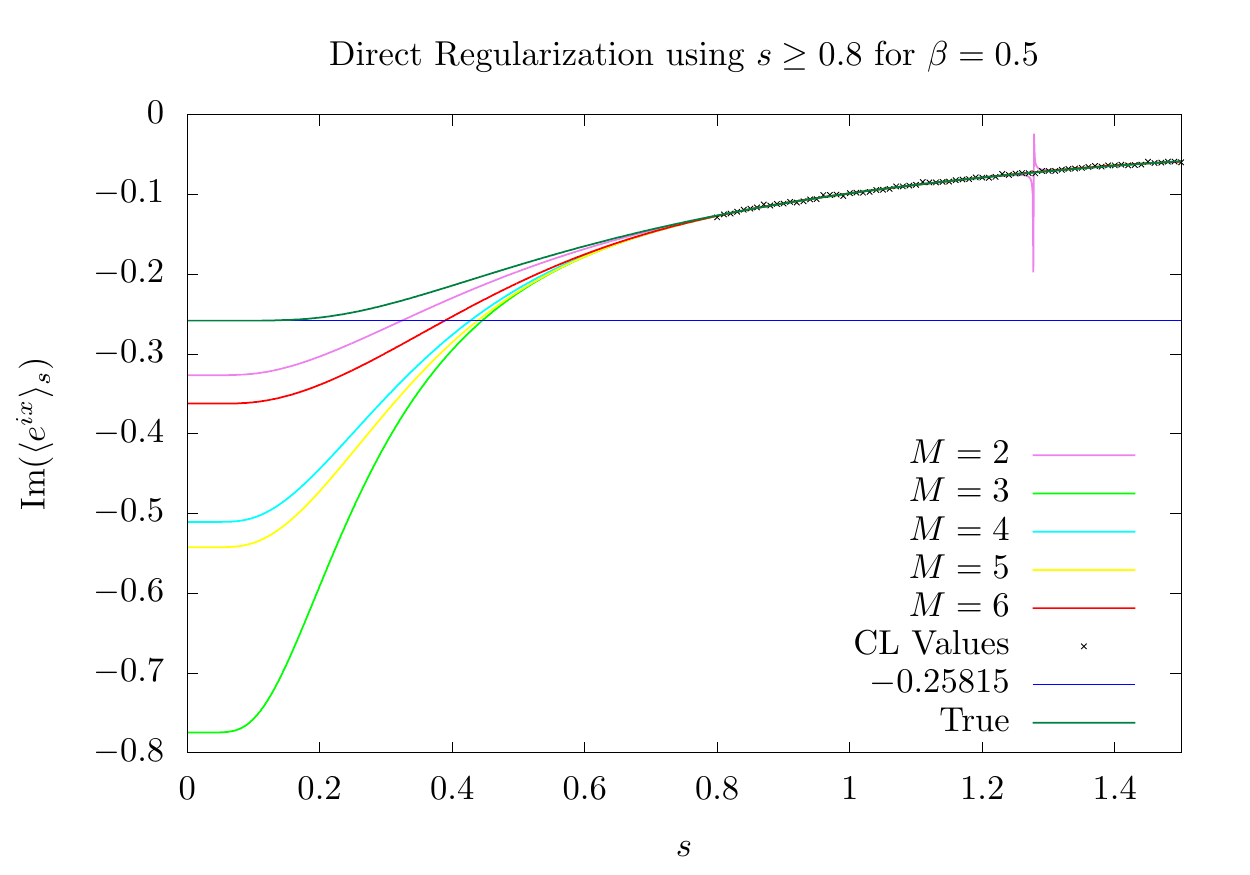}
\caption{Numerical results obtained by performing the relevant regression on observables via direct regularization. The left plot utilizes data points for $s\geq 0.4$ and the right plot utilizes data points for $s \geq 0.8$.}
\label{fig:U(1)DirectCurves}
\end{figure}

Next, we will illustrate the possible advantage obtained by reweighting our observables in accordance to Proposition \ref{U(1)Prop3}. As observed in \Cref{fig:U(1)RewDivFig}, 
starting from $s = 0.4$, the corresponding standard error grows rapidly as $s$ goes to $0$. Thus, this portion of data may not be suitable to be used in the regression. We would therefore like to remove part of the information from our data set. The criterion for this is based on the p-value and will be described in the following paragraph.

\begin{figure}[htbp]
\centering
\includegraphics[width=0.48\textwidth]{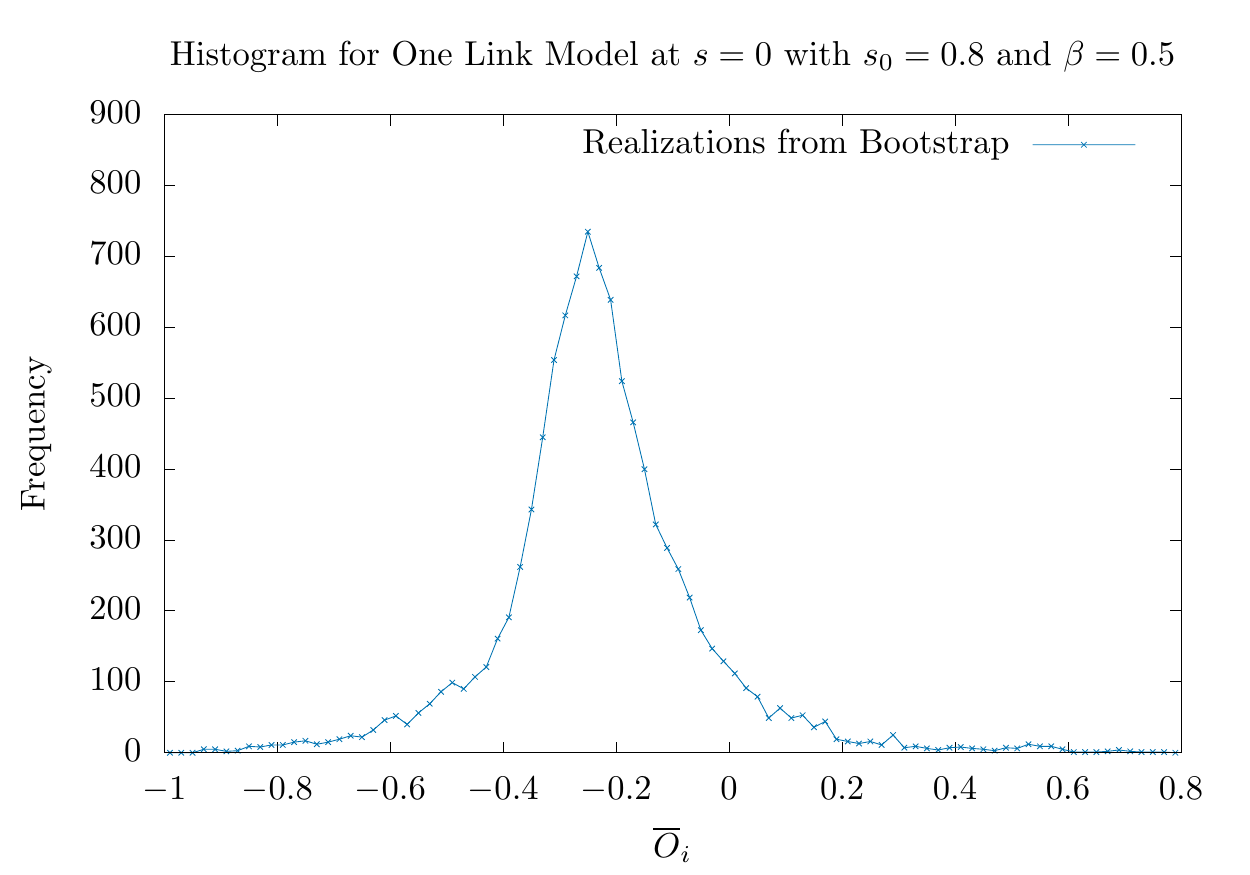}
\includegraphics[width=0.48\textwidth]{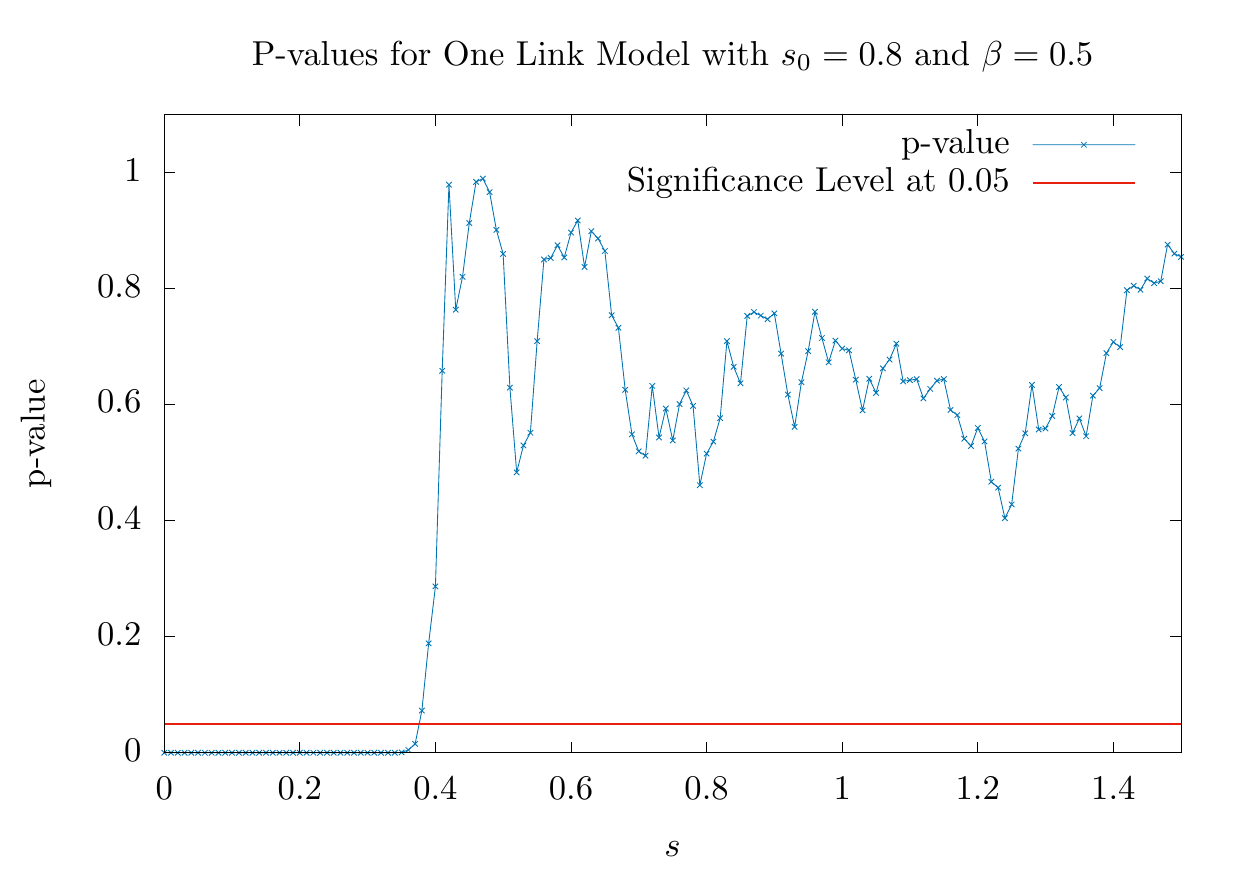}
\caption{The figure on the left depicts the histogram of the obtained distribution for ${\overline{O}_i}$ at $s = 0$. The figure on the right corresponds to the corresponding p-values obtained for $0 \leq s \leq 1.5$.}
\label{fig:U(1)Hist&Pvalues}
\end{figure}

The growth of the standard error as $s$ approaches $0$ can be explained as follows. First, we label each realization of the mean of the observable $e^{\ii x}$ as $\overline{O}_i$ for each iteration of the bootstrap method. Next, we obtain a histogram for the ${\overline{O}_i}$, as shown in the left diagram of Figure \ref{fig:U(1)Hist&Pvalues}. From there, we can observe that although the distribution looks somewhat symmetric and normal, the distribution of the ${\overline{O}_i}$ seems to be concentrated more at its mean. We can support this with the use of a Kolmogorov-Smirnov test, conducted against a normal distribution at each value of $s$. If the p-value at a given $s$ happens to be below $0.05$, we will reject the null hypothesis that the underlying distribution for $\langle O\rangle_s = \langle e^{\ii x} \rangle_s$ is normal and concluding that the underlying distribution at that value of $s$ is non-normal. Under the Generalized Central Limit Theorem, an instance in which a mean distribution converges to a non-normal distribution must corresponds to the fact that the underlying population has an infinite variance. Thus, as seen from the right diagram of Figure \ref{fig:U(1)Hist&Pvalues}, we will perform regression using points generated for $s \geq 0.39$ which corresponds to values of $s$ with $p-$value greater than or equals to $0.05$. The results are summarized below.

\begin{figure}[htbp]
\centering
\includegraphics[width=0.48\textwidth]{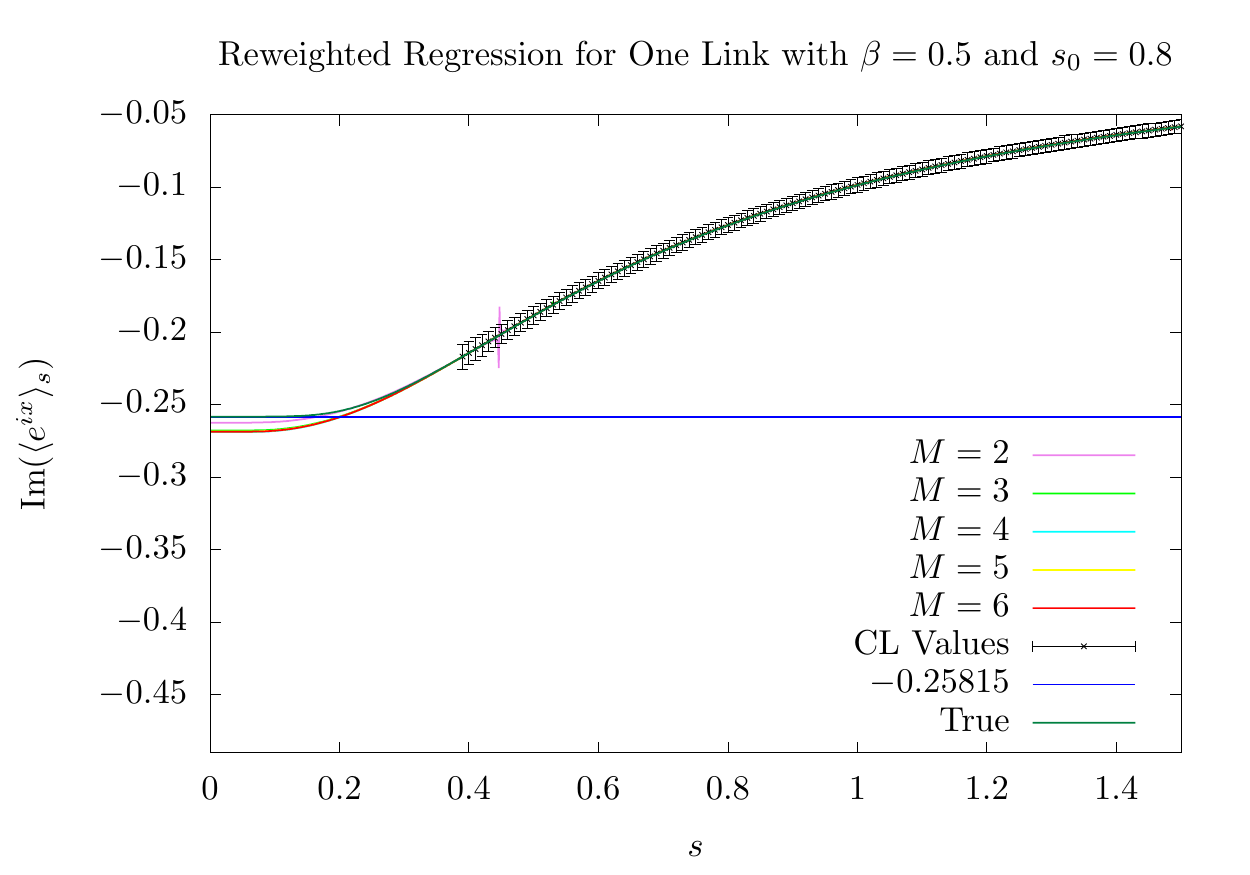}
\includegraphics[width=0.48\textwidth]{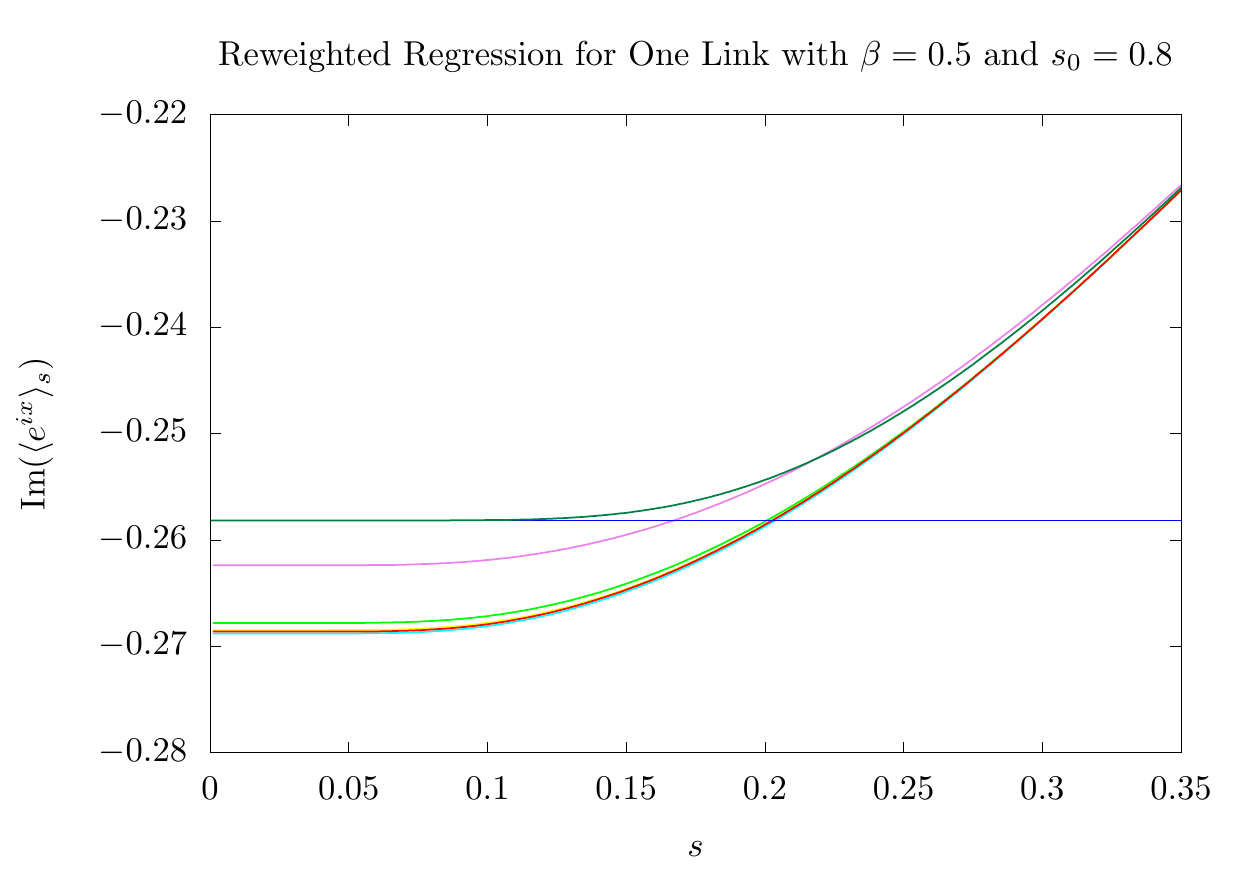}
\caption{Numerical results obtained by performing the relevant regression on reweighted observables. The plot on the right zooms in to display the behaviour for curves of different $M$ for small values of $s$ and values of relevant observable close to $-0.25815$.}
\label{fig:U(1)RewCurves}
\end{figure}

The numerical results for both methods, direct regression with regularized action ($2R$ method) and regressing reweighted observables ($3R$ method)\footnote{The three ``$R$''s mentioned here correspond to regularization, regression and reweighting. The missing $R$ in the first method corresponds to regression done without any sort of reweighting.}, are summarized in Table \ref{tab:U1NumResults}. Note that we have excluded the estimates obtained by the $2R$ method using data points with $s \geq 0.8$ as we can see from \Cref{fig:U(1)DirectCurves} that the values of $\im{\langle e^{ix} \rangle}$ predicted for the given values of $M$ are both inaccurate and imprecise. 
\begin{table}[htbp]
\centering
\caption{Estimates of $\im{\langle e^{\ii x}} \rangle$ for various methods at $\beta = 0.5$. Note that the \% Discrepancy is calculated with respect to the true value of $-0.25815$.}
\begin{tabular}{|c|c|c|c|c|}
\hline
    & Direct Reg $s\geq 0.4$ & \%  & Reweighted  & \%  \\
    $M$ & (\Cref{fig:U(1)DirectCurves}; $2R$ Method) & Disp &  (\Cref{fig:U(1)RewCurves}; $3R$ Method) & Disp \\
\hline
    $2$ & $-0.265004$ & $2.7$ &  $-0.262370$ & $1.6$ \\
    $3$ & $-0.255572$ & $1.0$ &  $-0.267792$ & $3.7$ \\
    $4$ & $-0.267196$ & $3.5$ &  $-0.268786$ & $4.1$ \\
    $5$ & $-0.268420$ & $4.0$ &  $-0.268468$ & $4.0$ \\
    $6$ & $-0.268089$ & $3.9$ &  $-0.268611$ & $4.1$ \\
\hline
\end{tabular}
\label{tab:U1NumResults}
\end{table}

We summarize some of the relevant key observations from Figures \ref{fig:U(1)DirectCurves}, \ref{fig:U(1)Hist&Pvalues}, and  \ref{fig:U(1)RewCurves} and Table \ref{tab:U1NumResults} below.
\begin{itemize}
    \item From Table \ref{tab:U1NumResults}, for a fixed $M \geq 3$, we can see that both methods are on par in terms of their accuracy in estimating the value of $\im \langle e^{ix} \rangle$. 
    \item However, the matched performance of the $2R$ Method for $s\geq 0.4$ as mentioned is on top of the fact that we have used \textit{a priori} information on the divergence of the numerical values generated using our complex Langevin algorithm as in \Cref{fig:U(1)NumFailFig}. This piece of information might not be available for general $\operatorname{U}(1)$ models such as the 3D XY Model.
    \item In the absence of \textit{a priori} information, as mentioned, the results obtained via the $2R$ Method for $s \geq 0.8$ are both inaccurate and imprecise. This can be already be seen from the right plot in \Cref{fig:U(1)DirectCurves}. Therefore, to get better results using the $2R$ method, one may consider including some biased results with acceptable errors (e.g. $\langle O \rangle_s$ with $s \in (0.4,0.8)$ in this example).
    \item In addition, we can see from \Cref{fig:U(1)NumFailFig} that the true curve (for $\im{\langle e^{\ii x}\rangle}_s$ as a function of $s$) diverges from the numerical value at $s \approx 0.4$, which is consistent with our non-normality test as explained for \Cref{fig:U(1)Hist&Pvalues}. As of now, we are not sure if this is a coincidence or if there are sufficient mathematical grounds to justify such a phenomenon.
\end{itemize}

\begin{remark}
In theory, only samples drawn from populations of infinite variance will give rise to stable distributions (instead of normal distributions) under the limit of a large sample size. However, depending on the type of bootstrapping method used, there are other types of distribution in which a simple naive bootstrapping method, like the one we have employed, might fail. Furthermore, as shown in \cite{InfVarBootstrap}, even for population distributions with infinite population variance that are suitably well-behaved, we note that the resulting distribution might not be that of a stable distribution, but rather, a random distribution. Nonetheless, the success of the naive bootstrapping method for large values of $s$ indicates that we might not face issues that we will with standard counterexamples to naive bootstrapping methods, such as the extreme order statistics but rather, the issue can be attributed to infinite population variance.
\end{remark}

\section{Generalization to multi-dimensional integrals in the \texorpdfstring{$\operatorname{U}(1)$}{U(1)} theory} \label{sec:GenU1} 
In this section, we will attempt to extend the method used for the $\operatorname{U}(1)$ one-link model to integrals on $\mathbb{T}^N$, $N \geqslant 1$ with more general actions.
This includes the $\operatorname{U}(1)$ lattice field theory, where $N$ equals the number of links.
As in \Cref{sec:CL}, we will use  $\{\phi\}$ to denote the collection of all the $N$ variables $\phi_1, \cdots, \phi_N$, and each $\phi_k$ is a variable in $\mathbb{T}$.
The action and the observable will be denoted by $S(\{\phi\})$ and $O(\{\phi\})$ respectively.
To extend the regularization to multi-dimensional models, we adopt the uniform regularization in all directions, that is, the regularized observable $\langle O \rangle_s$ is defined by
\begin{equation}\label{3DXY25}
\langle O \rangle_s = \frac{1}{Z_s} \int_{\mathbb{R}^N} O(\{\phi\}) \exp \left( -S(\{\phi\}) - \frac{s}{2} \sum_{j=1}^N \phi_j^2 \right) \mathrm{d}\phi_1 \cdots \mathrm{d}\phi_N,
\end{equation} with $Z_s$ defined by
\begin{equation}\label{3DXY26}
Z_s = \int_{\mathbb{R}^N} \exp \left( -S(\{\phi\}) - \frac{s}{2} \sum_{j=1}^N \phi_j^2 \right) \mathrm{d}\phi_1 \cdots \mathrm{d}\phi_N.
\end{equation}
The extension to multi-dimensional models is done by generalizing Propositions \ref{U(1)Prop1} to \ref{U(1)Prop5} whenever possible, giving the necessary proofs unless a given proposition generalizes clearly.

We start off by generalizing Proposition \ref{U(1)Prop1} as follows:
\begin{proposition}\label{3DXYProp1}
Suppose that both $e^{-S(\{ \phi \})}$ and $O(\{ \phi\})$ are $C^{\infty}$ functions on $\mathbb{T}^N$. Then it holds that
\begin{displaymath}
\lim_{s\rightarrow 0^+} \langle O \rangle_s = \langle O \rangle.
\end{displaymath}
\end{proposition}
\begin{remark}
The smoothness of $O(\{\phi\})$ and $e^{-S(\{ \phi\})}$ guarantees the interchangeability of the infinite sum and the integrals.
Similar to the one-dimensional case (\Cref{U(1)Prop1}), the $C^{\infty}$ requirement can be weakened to a certain H\"older class.
However, this $C^{\infty}$ condition is satisfied in most applications of the complex Langevin method, due to the analytic extensibility of both functions to the complexified space.
\end{remark}
\begin{proof}
The proof will largely mimic that of Proposition \ref{U(1)Prop1}. We can write down the multi-dimensional Fourier series of $O(\{\phi\})$ and $e^{-S(\{\phi\})}$ in the following form:
\begin{equation}\label{3DXY1}
O(\{\phi\}) = \sum_{\vec{m} = (m_1, \cdots, m_N) \in \mathbb{Z}^N} \widehat{O}_{\vec{m}} \exp \left( \ii \sum_{j=1}^N m_j \phi_j \right), \quad
e^{-S(\{\phi\})} = \sum_{\vec{k} = (k_1, \cdots, k_N) \in \mathbb{Z}^N} \widehat{\beta}_{\vec{k}} \exp \left( \ii \sum_{j=1}^N m_j \phi_j \right).
\end{equation}
In particular, analogous to Proposition \ref{U(1)Prop1}, the $C^{\infty}$ regularity allows for the interchange of the relevant infinite sum and the $N$-dimensional integral. Thus, it is sufficient to only consider integrals in the following form:
\begin{equation}\label{3DXY3}
\begin{aligned}
&\int_{\mathbb{R}^N} 
\exp \left( \ii \sum_{j=1}^N m_j \phi_j \right) \exp \left( -S(\{\phi\}) - \frac{s}{2} \sum_{j=1}^N \phi_j^2 \right) \mathrm{d} \phi_1 \cdots \mathrm{d} \phi_N \\
=& \int_{\mathbb{R}^N} \left[ \sum_{\vec k = (k_1,\cdots,k_N) \in \mathbb{Z}^N} \widehat{\beta}_{\vec k} \exp \ml \ii \sum_{j=1}^N k_j\phi_j \mr \right] \ml \prod_{j=1}^N e^{\ii m_j \phi_j -s\frac{\phi_j^2}{2}} \mr \mathrm{d} \phi_1 \cdots \mathrm{d} \phi_N \\
=& \sum_{\vec k = (k_1,\cdots,k_N) \in \mathbb{Z}^N} \widehat{\beta}_{\vec k} \int_{\mathbb{R}^N} \prod_{j=1}^N e^{ \ii (m_j + k_j) \phi_j -s\frac{\phi_j^2}{2} } \mathrm{d} \phi_1 \cdots \mathrm{d} \phi_N \\
=& \left( \frac{2\pi}{s} \right)^{N/2} \sum_{\vec{k} = (k_1, \cdots, k_N) \in \mathbb{Z}^N} \widehat{\beta}_{\vec k} \exp \left( \frac{1}{2s} \sum_{j=1}^N (m_j + k_j)^2 \right).
\end{aligned}
\end{equation}
Here we have again interchanged the summation and the integral, which can be justified by an argument similar to \eqref{U(1)11} using the smoothness of $e^{-S(\{\phi\})}$.
Thus, we have
\begin{equation}\label{3DXY14}
\langle O \rangle_s = \frac{1}{Z_s}  \ml \frac{2\pi}{s} \mr^{N/2} \sum_{\vec{m} \in \mathbb{Z}^N} \widehat{O}_{\vec{m}} \sum_{\vec k \in \mathbb{Z}^N}\widehat{\beta}_{\vec k}  \exp \left( -\frac{1}{2s} \sum_{j=1}^N (m_j+k_j)^2 \right), 
\quad Z_s = \ml \frac{2\pi}{s} \mr^{N/2} \sum_{\vec{k} \in \mathbb{Z}^N}\widehat{\beta}_{\vec{k}} \exp \left(-\frac{1}{2s} \sum_{j=1}^N k_j^2 \right),
\end{equation}
from which we can obtain the limit
\begin{equation} \label{eq:lim_Os}
\lim_{s\rightarrow 0^+} \langle O \rangle_s = \frac{1}{\widehat{\beta}_{\vec{0}}}\sum_{\vec{m} \in \mathbb{Z}^N} \widehat{O}_{\vec{m}} \widehat{\beta}_{-\vec m}.
\end{equation}

It remains to show that the right-hand-side of \eqref{eq:lim_Os} equals $\langle O \rangle$.
To this end, we write $\langle O \rangle$ as
\begin{equation}\label{3DXY34}
\langle O \rangle = \frac{(2\pi)^{-N} \int_{\mathbb{T}^N} O(\{\phi\}) e^{-S(\{\phi\})}\mathrm{d}\{\phi\}}{(2\pi)^{-N} \int_{\mathbb{T}^N} e^{-S(\{\phi\})} \mathrm{d}\{\phi\}}.
\end{equation}
By definition, it is clear that $\beta_{\vec{0}}$ is the (unweighted)\footnote{Here, unweighted mean of an expression $f$ refers to $\langle f\rangle_{\text{unweighted}} = \frac{\int_D f(x) \mathrm{d}x}{\int_D 1 \mathrm{d}x}$, where $D$ is the appropriate domain for consideration. One can compare this expression to the expression for weighted means as in \eqref{3DXY25} and \eqref{3DXY26}.}
mean value of $e^{-S(\{\phi\})}$, which is equals the denominator in \eqref{3DXY34}. Since
\begin{equation}\label{3DXY33}
O(\{\phi\})e^{-S(\{\phi\})} = \left( \sum_{\vec{k} \in \mathbb{Z}^N} \widehat{O}_{\vec k} \prod_{j=1}^N e^{\ii k_j \phi_j} \right)
\left( \sum_{\vec{k} \in \mathbb{Z}^N} \beta_{\vec k} \prod_{j=1}^N e^{\ii k_j \phi_j} \right) = \sum_{\vec{k} \in \mathbb{Z}^N} \left( \sum_{\vec{m} \in \mathbb{Z}^N} \widehat{O}_{\vec{m}} \beta_{\vec{k}-\vec{m}} \right)\prod_{j=1}^N e^{\ii k_j \phi_j},
\end{equation}
we see that the right-hand side of \eqref{eq:lim_Os} is the Fourier coefficient of the zero-frequency mode in the expansion of $O(\{\phi\})e^{-S(\{\phi\})}$,
and thus corresponds to the numerator in \eqref{3DXY34}.
We have thus completed the proof.
\end{proof}

Next, we note the range of $s(\beta)$ depends largely on the functional form of $S(\{ \phi\})$. Thus, Proposition \ref{U(1)Prop2} does not generalize easily. Instead, we will proceed to generalize Proposition \ref{U(1)Prop3}. However, this is straightforward, as we only have to replace  $|x|^2$ with the sum of $|\phi_j|^2$, as elaborated below:
\begin{proposition}\label{3DXYProp3}
The following equality holds for all $s,s_0 \geq 0$:
\begin{equation}\label{3DXY12}
 \langle O(\{\phi\}) \rangle_s = \frac{\left\langle O (\{\phi\}) \exp\ml \frac{s_0-s}{2} \sum_{j=1}^N |\phi_j|^2 \mr \right\rangle_{s_0}}{\left\langle \exp\ml \frac{s_0-s}{2} \sum_{j=1}^N |\phi_j|^2 \mr \right\rangle_{s_0}}.
\end{equation}
\end{proposition} As the proof is completely analogous to that in Proposition \ref{U(1)Prop3}, we shall skip the proof.

We can also observe a similar phenomenon as Proposition \ref{U(1)Prop4} for all possible $S(\{ \phi\})$ satisfying the relevant conditions for uniform convergence of their corresponding Fourier series. In particular, the following proposition generalizes this.

\begin{proposition}\label{3DXYProp4}
Let the action $S(\{\phi\})$ and any observable $O$ satisfying the conditions in Proposition \ref{3DXYProp1} be given. For any given $k \in \mathbb{Z}^+$, we have
\begin{equation}\label{3DXY13}
\frac{\mathrm{d}^k}{\mathrm{d}s^k} \langle O \rangle_s \bigg|_{s = 0} = 0.
\end{equation}
\end{proposition}
\begin{proof}
As the proof for this proposition is largely similar to that in Proposition \ref{U(1)Prop4}, we will sketch the key differences here. From \eqref{3DXY14}, we have $\langle O \rangle_s = f_1(s) / f_2(s)$, where
\begin{equation}\label{3DXY15}
f_1(s) = \sum_{\vec{m} \in \mathbb{Z}^N} \widehat{O}_{\vec{m}} \sum_{\vec k \in \mathbb{Z}^N}\beta_{\vec k}  \exp \left( -\frac{1}{2s} \sum_{j=1}^N (m_j+k_j)^2 \right), 
\qquad f_2(s) = \sum_{\vec{k} \in \mathbb{Z}^N}\beta_{\vec{k}} \exp \left(-\frac{1}{2s} \sum_{j=1}^N k_j^2 \right).
\end{equation}
The structure of this function is similar to \eqref{U(1)33}, and we can use exactly the same approach to prove that all the derivative at $s = 0$ vanishes. Here the interchangeability of the derivative and the series follows the fact that both $O$ and $S$ are infinitely, so that the coefficients $\widehat{O}_{\vec k}$ and $\beta_{\vec k}$ decay faster than $1/|\vec k|^{\alpha}$ for any $\alpha > 0$.
\end{proof}

Based on the expression \eqref{3DXY15}, it is straightforward to derive the following proposition, which is similar to \Cref{U(1)Prop5}:

\begin{proposition}\label{3DXYProp5}
For any observable $O$ satisfying the conditions as stated in \Cref{3DXYProp1}, there exist coefficients $a_k$ and $b_k$ for all $k \in \mathbb{N}$ such that
\begin{equation}\label{3DXY18}
\langle O \rangle_s = \frac{\sum_{k=0}^\infty a_k e^{-\frac{k}{2s}}}{\sum_{k=0}^\infty b_k e^{-\frac{k}{2s}}}.
\end{equation}
\end{proposition}

Note that a key difference between \eqref{3DXY15} and \eqref{U(1)33} is that the sum over the components $\sum_{j=1}^N$ does not exist in the one-link model.
In fact, the integer $\sum_{j=1}^N k_j^2$ can take any value from $0$ to $N$, so that when $N$ is large, the zero terms in both summations in \eqref{3DXY18} will appear only for a large $k$.
To accommodate for such cases, we have included all possible integer coefficients of $-\frac{1}{2s}$ instead as shown in \eqref{3DXY18}.

In view of the proposition above, an appropriate regression model is given by 
\begin{equation}\label{3DXY22}
\langle O \rangle_s = \frac{\sum_{k=0}^M a_k e^{-\frac{k}{2s}}}{1 + \sum_{k=1}^M b_k e^{-\frac{k}{2s}}}.
\end{equation}
Note that for any model under the $\operatorname{U}(1)$ lattice field theory framework, similar to the regression model in  \eqref{U(1)39}, we can minimize a similar objective function as described in \eqref{eq:lossfunction} to obtain the corresponding regression coefficients.

\section{Extension to the \texorpdfstring{$\operatorname{SU}(n)$}{SU(n)} theory}\label{sec:RegSUN}
This section introduces how we can extend results from the previous sections to the $\SU{n}$ theory, and presents results for the one-dimensional problems. Before we introduce the regularization method in $\SU{n}$ theory, we will briefly review the complex Langevin method under the $\SU{n}$ gauge theory. 

Consider the set $\{U_k : k=1,\dots,N\}\in [\SU{n}]^{N}$, where  $N$ is the total number of lattice points. Given the action $S(\{U\})$, the expectation value for the observable $O(\{U\})$ is given by
\begin{equation}\label{eq:sun-obs}
\langle O \rangle = \frac{\int_{[\SU{n}]^N} O (\{U\}) \exp(-S(\{U\}) \,\mathrm{d}\{U\}}{\int_{[\SU{n}]^N} \exp(-S(\{U\}) \,\mathrm{d}\{U\}}.
\end{equation}

Let $w_{a,k}$ for $a = 1,\dots,n^2-1$ be independent Brownian motions. Upon complexifying the configuration space $[\SU{n}]^N$ to $[\operatorname{SL}(n,\mathbb{C})]^N$, the complex Langevin method for $\SU{n}$ theory can be described by the complex stochastic process:
\begin{equation} \label{eq:sun-cl}
\mathrm{d} U_{k} = -\sum_{a=1}^{n^2-1} \mathrm{i} \lambda_a \Big[U_{k} D_{a,k}S(\{U\}) \mathrm{d}t + U_{k} \circ \mathrm{d}w_{a,k} \Big], \quad U_{k}\in\operatorname{SL}(n,\mathbb{C}), k = 1,\dots,N,
\end{equation} 
where $\circ$ stands for the Stratonovich interpretation of the stochastic integral and  $\lambda_a$, $a = 1,\dots,n^2-1$ are the infinitesimal generators of the $\SU{n}$ group satisfying the orthogonality
\begin{equation*}
\text{tr}(\lambda_a \lambda_b) = 2 \delta_{ab}, \quad \forall a,b=1,\dots,n^2-1,
\end{equation*}
and $D_{a,k}$ denotes the left Lie derivative operator defined as
\begin{equation}\label{eq:sun-llied}
D_{a,k} S(\{U\}) = \lim_{\epsilon \rightarrow 0} \frac{S(\{\widetilde{U}^{\epsilon}\}) - S(\{U\})}{\epsilon},
\end{equation}
where $\{\widetilde{U}^{\epsilon}\}$ denotes the field with links defined by
\begin{displaymath}
\widetilde{U}_l^{\epsilon} = \exp(\ii \epsilon \delta_{kl} \lambda_a) U_l, \qquad l = 1,\ldots,N.
\end{displaymath}
To solve \eqref{eq:sun-cl}, we mimic standard methods and  apply the following scheme to update the links as follows:
\begin{equation} \label{eq:Euler}
U_{k}^{(j+1)} = \exp \left( -\sum_{a=1}^{n^2-1} \mathrm{i} \lambda_a \left( D_{a,k}S^{(j)} \Delta t + \eta_{a,k} \sqrt{\Delta t} \right)  \right) U_{k}^{(j)}, \quad k=1,\dots,N,
\end{equation}
where $U_k^{(j)}$ is the link at time instance $t_j$,   $\Delta t=t_{j+1}-t_j$ is time step and each $\eta_{a,k}$ is normally distributed with mean $0$ and variance $2$. The scheme is similar to the Euler-Maruyama method, while the exponential map is applied to keep the solution from leaving the Lie group.

Below, we will generalize the techniques of regularization, reweighting, and regression to the $\SU{n}$ group theory. They will be introduced in the following three subsections.
 
\subsection{Regularization}
To demonstrate how the regularization can be generalized to the $\SU{n}$ theory, we would like to restate the method for the $\operatorname{U}(1)$ theory in the language of group theories.
Here we regard $\operatorname{U}(1)$ as the unit circle on the complex plane, so that we can establish the one-to-one map between $U \in \operatorname{U}(1)$ and $x \in \mathbb{T}$ by $U = \exp(\ii x)$.
Thus, we can write down the integrals over $\mathbb{T}$ as integrals over $\operatorname{U}(1)$.
For instance,
\begin{equation} \label{eq:int_U1}
Z = \int_{\mathbb{T}} \exp(-S(x)) \mathrm{d}x =
  \int_{\operatorname{U}(1)} \exp \Big({-S(-\ii \log U)} \Big) \mathrm{d}U,
\end{equation}
where we have assumed that $S(\cdot)$ is periodic with period $2\pi$, so that the value of $S(-\ii \log U)$ is unique. In the equation \eqref{eq:int_U1}, we can also consider $\log U$ as an element in the Lie algebra of $\operatorname{U}(1)$, i.e., $\mathfrak{g} = \ii \mathbb{R}$.
When we apply the regularization, an extra term $sx^2/2$ is added to the action, whose counterpart should be $-s(\log U)^2/2$ if we represent the regularization term using the variable in $\operatorname{U}(1)$. However, since $sx^2/2$ is no longer periodic with respect to $x$, the expression $-s(\log U)^2/2$ becomes multi-valued, so simply adding this term to the action will cause ambiguities.

To address this problem, we can rewrite \eqref{U(1)1} and \eqref{U(1)2} in the following form:
\begin{equation} \label{eq:lieU1-reg}
\begin{aligned}
\langle O \rangle_s &= \frac{1}{Z_s} \int_{\operatorname{U}(1)} O(-\ii \log U) \sum_{\substack{g \in \mathfrak{g} \\ \text{ s.t.} \exp(g) = U}} \exp \left(-S(-\ii \log U) + \frac{s g^2}{2} \right) \mathrm{d}U, \\
Z_s &= \int_{\operatorname{U}(1)} \sum_{\substack{g \in \mathfrak{g} \\ \text{ s.t.} \exp(g) = U}} \exp \left(-S(-\ii \log U) + \frac{s g^2}{2} \right) \mathrm{d}U.
\end{aligned}
\end{equation}
Here the summation is taken over all logarithms of $U$, corresponding to unrolling the torus $\mathbb{T}$ to the real axis $\mathbb{R}$. Our generalization of the regularization to the $\SU{n}$ theory will be based on the form of \eqref{eq:lieU1-reg}.

Formally, for the $\SU{n}$ theory, the regularized action can be written as
\begin{equation}\label{eq:sun-regLieA}
S_s(\{U\}) = S(\{U\}) - \frac{s}{2} \sum_{k=1}^N \tr[(\log U_k)^2],
\end{equation}
where $\log U_k$ is an element in the Lie algebra $\mathfrak{su}(n)$.
Here the regularization term is chosen as the trace of the matrix square since its square root defines a norm on $\mathfrak{su}(n)$.
However, the ambiguity again comes from the non-uniqueness of the matrix logarithm.
Therefore, we mimic \eqref{eq:lieU1-reg} to write down the regularized observable as
\begin{equation} \label{eq:sun-regLieA-obs}
\begin{aligned}
\langle O \rangle_s &= \frac{1}{Z_s} \int_{[\SU{n}]^N} O(\{U\}) \sum_{\substack{g_1 \in \mathfrak{su}(n) \\ \text{ s.t.} \exp(g_1) = U_1}} \cdots\sum_{\substack{g_N \in \mathfrak{su}(n) \\ \text{ s.t.} \exp(g_N) = U_N}} \exp \left(-S(\{U\}) + \sum_{k=1}^N \frac{s}{2} \tr(g_k^2) \right) \mathrm{d}U, \\
Z_s &= \int_{[\SU{n}]^N} 
\sum_{\substack{g_1 \in \mathfrak{su}(n) \\ \text{ s.t.} \exp(g_1) = U_1}} \cdots\sum_{\substack{g_N \in \mathfrak{su}(n) \\ \text{ s.t.} \exp(g_N) = U_N}} \exp \left(-S(\{U\}) + \sum_{k=1}^N \frac{s}{2} \tr(g_k^2) \right) \mathrm{d}U.
\end{aligned}
\end{equation}
Following the idea in \Cref{U(1)Prop1}, it is expected that when $s\to 0^+$, the regularized observable $\langle O \rangle_s$ will converge to $\langle O \rangle$ in \eqref{eq:sun-obs}. We leave the rigorous proof as our further work and we focus on the implementation of the regularized method for $\SU{n}$ theory in this section.

To formulate the complex Langevin method for the complexified action \eqref{eq:sun-regLieA}, we need to first formulate the Langevin method for real actions.
The general idea of the numerical scheme follows \eqref{eq:Euler}, while the derivative of the action $D_{a,k} S$ should be replaced with $D_{a,k} S_s$ defined by
\begin{equation}\label{eq:sun-regLieD}
\begin{aligned}
D_{a,k}S_{s}(\{U\}) =D_{a,k}S(\{U\})-\frac{s}{2}\lim_{\epsilon \to 0} \frac{\tr\left[(\log(e^{\ii\lambda_a\epsilon}U_k))^2 - (\log U_k)^2\right]}{\epsilon}.
\end{aligned}
\end{equation}
The numerical scheme is fully determined once the Lie derivative is determined. Meanwhile, the generalization to the complex action follows naturally as the formula of the numerical scheme is unchanged, and each link $U_k$ automatically falls into the complexification of $\SU{n}$, i.e., the special linear group $\SL{n}$. Below, we will focus on the computation of $\eqref{eq:sun-regLieD}$ with $U_k \in \SL{n}$ and resolve the complication caused by the multi-valued logarithmic function.

In \eqref{eq:sun-regLieD}, $g_k := \log U_k$ can take any matrix logarithm of $U_k$, while $\log(e^{\ii\lambda_a\epsilon}U_k)$ must be the matrix logarithm that is closest to $g_k$ such that the limit exists.
The value of the limit is given in the following proposition:
\begin{proposition}
For any $U_k \in \SL{n}$, it holds that
\begin{displaymath}
\lim_{\epsilon \to 0} \frac{\tr\left[(\log(e^{\ii\lambda_a\epsilon}U_k))^2 - (\log U_k)^2\right]}{\epsilon} = 2\ii \tr(\lambda_a \log U_k).
\end{displaymath}
\end{proposition}
\begin{proof}
Since $\tr(A^2-B^2) = \tr(A+B)(A-B)$ for any matrices $A$ and $B$, we can simplify the limit as follows:
\begin{equation} \label{eq:limit}
\begin{split}
\lim_{\epsilon \to 0} \frac{\mbox{tr}\left[(\log(e^{\ii\lambda_a\epsilon}U_k))^2 - (\log U_k)^2\right]}{\epsilon} &=
\lim_{\epsilon \to 0}  \tr \left([\log(e^{\ii\lambda_a\epsilon}U_k) + \log U_k] \frac{\log(e^{\ii\lambda_a\epsilon}U_k) - \log U_k}{\epsilon} \right) \\
&= 2 \lim_{\epsilon \to 0}  \tr \left(\log U_k \frac{\log(e^{\ii\lambda_a\epsilon}U_k) - \log U_k}{\epsilon} \right).
\end{split}
\end{equation}
Let $\Delta_1^{\epsilon} = \log(e^{\ii\lambda_a\epsilon}U_k) - \log U_k$ and
\begin{equation}
\Delta_{i+1}^{\epsilon} = [\log U_k, \Delta_i^{\epsilon}] = (\log U_k) \Delta_i^{\epsilon} - \Delta_i^{\epsilon} (\log U_k), \qquad i = 1,2,\ldots.
\end{equation}
It can be derived from the cyclic property of the matrix trace that
\begin{equation} \label{eq:zero_trace}
\tr [(\log U_k) \Delta_{i+1}^{\epsilon}] = \tr[ (\log U_k)^2 \Delta_i^{\epsilon} - (\log U_k) \Delta_i^{\epsilon} (\log U_k)] = 0, \qquad i = 1,2,\ldots.
\end{equation}
According to the differential formula of the exponential map
\cite[Theorem~5]{rossmann2006lie}, we have
\begin{equation*}
\begin{split}
    \sum_{i=1}^{+\infty} \frac{(-1)^{i+1}}{i!} \Delta_i^{\epsilon} &= U_k^{-1} e^{\ii \lambda_a \epsilon} U_k - I + O(\epsilon^2) = \ii\epsilon U_k^{-1} \lambda_a U_k + \mathcal{O}(\epsilon^2).
\end{split}
\end{equation*}
Using \eqref{eq:zero_trace}, we can left-multiply the above equation by $\log U_k$ and then take the trace to obtain
\begin{equation}
\tr[(\log U_k) \Delta_1^{\epsilon}] = \ii \epsilon \tr[(\log U_k) U_k^{-1} \lambda_a U_k] + \mathcal{O}(\epsilon^2) = \ii \epsilon \tr (\lambda_a \log U_k) + \mathcal{O}(\epsilon^2).
\end{equation}
Inserting this equation into \eqref{eq:limit} concludes the proof. 
\end{proof}
By this proposition, the drift term \eqref{eq:sun-regLieD} becomes
\begin{equation}\label{eq:sun-regLieD-2}
    D_{a,k}S_{s}(\{U\}) 
    = D_{a,k}S(\{U\}) -\ii s\tr(\lambda_a \log U_k),
\end{equation}
based on which the update of the links \eqref{eq:Euler} becomes
\begin{equation} \label{eq:regularized_evolution}
U_{k}^{(j+1)} = \exp \left( -\sum_{a=1}^{n^2-1} \mathrm{i} \lambda_a \left( D_{a,k}S^{(j)} \Delta t  - \underline{\ii s \tr(\lambda_a \log U_k^{(j)}) \Delta t} + \eta_{a,k} \sqrt{\Delta t} \right)  \right) U_{k}^{(j)}, \quad k=1,\dots,N,
\end{equation}
where the underlined term produces a pull-back velocity so that the excursion away from $\SU{n}$ can be restricted.

We now consider the determination of $g_k = \log U_k$. We first assume that $U_k$ is diagonalizable, i.e., $U_k = R \Theta R^{-1}$ for some $R \in \SL{n}$ and $\Theta = \operatorname{diag} (\theta_1, \cdots, \theta_n)$. Then
\begin{equation}\label{eq:sun-liea}
  g_k = \log U_k = \log(R\Theta R^{-1})
  = R (\log \Theta ) R^{-1}.
\end{equation}
The non-uniqueness of $g_k$ comes from the non-uniqueness of $\log \Theta$.
If $\Xi=\operatorname{diag}(\xi_1,\dots,\xi_n)$ is a diagonal matrix satisfying
\begin{equation}
\exp(\xi_i) =\theta_i, \quad i = 1,\cdots,n \qquad \text{and} \qquad \sum_{i=1}^n \xi_i = 0.
\end{equation}
Then for any $K_1, \cdots, K_n \in \mathbb{Z}$ satisfying $K_1 + \cdots + K_n = 0$, we have
\begin{equation} \label{eq:matrix_logarithm}
\exp(\Xi + 2\pi \ii \Gamma) = \Theta,
\end{equation}
where $\Gamma = \operatorname{diag} (K_1, \cdots, K_n)$.
Hence, the matrix $R(\Xi + 2\pi \ii \Gamma) R^{-1}$ is a candidate of $g_k$. Note that here we require $\tr \Xi = \tr \Gamma = 0$ since $g_k$ is required to be an element in $\mathfrak{sl}(n,\mathbb{C})$, which contains all traceless $n\times n$ matrices. If $U_k$ is non-diagonalizable, then $\Theta$ will contain Jordan blocks and $\Xi$ becomes an upper-triangular matrix. Nonetheless, the relation \eqref{eq:matrix_logarithm} still holds and still plays the role that leads to multiple values of matrix logarithms.

To resolve this issue, we would like to determine a specific $\Xi_k^{(j)}$ for each $j$ and $k$ such that the diagonal of $\exp(\Xi_k^{(j)})$ consists of all the eigenvalues of $U_k$, which will further determine the matrix logarithm.
In order to maintain the continuity of the dynamics, we choose $\Xi_k^{(j+1)}$ by minimizing its difference from $\Xi_k^{(j)}$. In other words, we determine $\Xi_k^{(j+1)}$ via
\begin{equation} \label{eq:sun-chooseg}
\begin{aligned}
\Xi_k^{(j+1)} = {} & \operatorname*{argmin}_{\Xi \in \mathbb{C}^{n\times n}} \|\Xi - \Xi_k^{(j)}\|_F, \\
& \text{s.t. } \Xi \text{ is a Jordan normal form and } \exists R\in \mathbb{C}^{n\times n} \text{ such that } \exp(R \Xi R^{-1}) = U_k^{(j+1)}.
\end{aligned}
\end{equation}
When $j = 0$, we simply choose $\Xi_k^{(0)}$ such that the imaginary parts of all its diagonal elements locate in $[-\pi, \pi)$.
Here, we comment that due to the existence of the stochastic term, $\Xi_k^{(j+1)}$ may be distant from $\Xi_k^{(j)}$. Thus, our strategy does not produce the ``correct'' choice of $\Xi_k^{(j+1)}$. However, such a probability decreases exponentially as $\Delta t$ decreases, which will not affect the order of accuracy for the numerical method. Note that in this approach, we need to keep track of the evolution of $\Xi_k^{(j)}$, which guarantees that the summations in \eqref{eq:sun-regLieA-obs} are taken into account.

To summarize, we describe the algorithm of the complex Langevin method for the $\SU{n}$ theory below:
\begin{algorithm2e}[ht]
	\caption{Complex Langevin method for the regularized action}
	\label{alg:cl-sun-reg}
	\SetKwInOut{Input}{Input}
	\SetKwInOut{Output}{Output}
	\SetKwComment{Comment}{}{}
	\BlankLine 
	\Input{Initial field $\{U^{(0)}\}$, time step $\Delta t$, total number of time steps $J$, number of time steps to reach the invariant measure $J_0$, number of time steps between two samples $\Delta J$}
	\BlankLine
    Set $j \gets 0$, $N_{\mathrm{sample}} \gets 0$ and $\langle O \rangle_s \gets 0$\;
    For each $k$, let $\Xi_k$ be the diagonal matrix with diagonal elements being the logarithms of the eigenvalues of $U_k^{(0)}$, and compute its logarithm $\log U_k^{(0)} = R \Xi_k^{(0)} R^{-1}$\;
    \For{$j \gets 0$ \KwTo $J$}{
        Use the result of $\log U_k^{(0)}$ to evolve the solution by \eqref{eq:regularized_evolution}\;
        \If{$j > J_0$}{
            $\langle O \rangle_s \gets \langle O \rangle_s + O(\{U^{(j+1)}\})$\;
            $J_0 \gets J_0 + \Delta J$, $N_{\mathrm{sample}} \gets N_{\mathrm{sample}}+1$\;
        }
        Compute $\Xi_k^{(j+1)}$ for each $k$ according to \eqref{eq:sun-chooseg}, and find $\log U_k^{(j+1)}$ based on the result of $\Xi_k^{(j+1)}$\;
    }
    $\langle O\rangle_s \gets \langle O \rangle_s / N_{\mathrm{sample}}$\;
    \BlankLine
    \Output{$\langle O\rangle_s$}
\end{algorithm2e}

Like in the regularized $\operatorname{U}(1)$ theory, when the regularization parameter $s$ is small, the regularized method may still converge to biased results. Our improvements, including reweighting and regularization, will be studied in the following two subsections.

\subsection{Reweighting}
The idea of reweighting follows from out motivating example, as in \eqref{U(1)29}. We represent the observable $\langle O \rangle_s$ as
\begin{equation} \label{eq:sun-reweighted}
\langle O \rangle_s = \frac{\langle O \exp(S_{s_0} - S_s)\rangle_{s_0}}{\langle \exp(S_{s_0} - S_s)\rangle_{s_0}},
\end{equation}
where
\begin{equation} \label{eq:s0-s}
S_{s_0} - S_s = -\frac{s_0-s}{2}\sum_{k=1}^N\tr[(\log U_k)^2]
\end{equation}
according to the definition \eqref{eq:sun-regLieA}.
To compute the numerator and the denominator of \eqref{eq:sun-reweighted}, \Cref{alg:cl-sun-reg} can still be applied, and the logarithms in \eqref{eq:s0-s} can still be found by using $\Xi_k^{(j)}$ at each time step (see \eqref{eq:sun-chooseg}). However, when $N$ is large, the value of $S_{s_0} - S_s$ in \eqref{eq:s0-s} might be a large positive number if $s_0 > s$ (note that $\tr[(\log U_k)^2] < 0)$. Consequently, its exponent $\exp(S_{s_0} - S_s)$ will be so huge that it will be difficult to handle using double-precision floating-point numbers.
Therefore, we shift $S_{s_0} - S_s$ by its average, so that \eqref{eq:sun-reweighted} can be computed as
\begin{equation} \label{eq:sun-reweighted-shifted}
\langle O \rangle_s = \frac{\Big\langle O \exp\Big(S_{s_0} - S_s - \langle S_{s_0} - S_s \rangle_{s_0}\Big)\Big\rangle_{s_0}}{\Big\langle \exp\Big(S_{s_0} - S_s - \langle S_{s_0} - S_s \rangle_{s_0}\Big)\Big\rangle_{s_0}}.
\end{equation}
This requires us to record the values of $O$ and $S_{s_0} - S_s$ for each sample, so that we can first compute $\langle S_{s_0} - S_s \rangle$, and then use the result to evaluate \eqref{eq:sun-reweighted-shifted}.
In our tests, the magnitude of the shifted exponents turns out to be acceptable after applying such a trick.

\subsection{Regression} \label{sec:sun-reg}
For the $\operatorname{U}(1)$ theory, the expression we used in the regression has a fractional form \eqref{3DXY18}, which is derived based on the Fourier expansion of functions on $\operatorname{U}(1)$.
Similarly, for the $\SU{n}$ theory, suppose $\psi_k$, $k = 1,2,\ldots$ form an orthonormal set of basis:
\begin{equation}
\int_{\SU{n}} \psi_k(U) \psi_l(U) \mathrm{d}U = \delta_{kl}.
\end{equation}
Then the expression used in the regularization can be determined by computing 
\begin{equation} \label{eq:sun-integral}
\int_{\SU{n}} \psi_k(U) \sum_{\substack{g \in \mathfrak{su}(n) \\ \text{s.t.} \exp(g) = U}}
\exp \left( \frac{s}{2} \tr g^2 \right) \mathrm{d}U, \qquad k = 1,2,\ldots.
\end{equation}

For the $\SU{2}$ theory, the integral \eqref{eq:sun-integral} can be calculated using the isomorphism between $\SU{2}$ and the 3-sphere $\mathbb{S}^3$.
The basis functions $\psi_k(U)$ can be chosen as the generalized spherical harmonics defined on $\mathbb{S}^3$, which are denoted by $Y_{l_1 l_2 l_3}(\psi, \theta, \varphi)$ with $(\psi, \theta, \varphi)$ being the hyperspherical coordinates. The precise form of $Y_{l_1 l_2 l_3}$ can be found in \cite{higuchi1987symmetric}.
To compute the integral, we write \eqref{eq:sun-integral} as an integral on the three-dimensional linear space $\mathfrak{su}(2)$.
With proper changes of variables, the integral \eqref{eq:sun-integral} can be transformed into
\begin{equation}\label{eq:su2-s3}
I_{l_1 l_2 l_3} := \sqrt{\frac{2l_2+1}{4\pi} \frac{(l_2 + l_1)!}{(l_2 - l_1)!}}
\int_{\mathbb{R}} \int_{\mathbb{R}} \int_{\mathbb{R}}
  \left( \frac{g_1 + \ii g_2}{\sqrt{g_1^2 + g_2^2}} \right)^{l_1}
  P_{l_2}^{-l_1} \left( \frac{g_3}{|g|} \right)
  \frac{Q_{l_3}^{l_2}(\sin |g|, \cos |g|)}{|g|\sqrt{g_1^2+g_2^2}}
  \exp \left(-s |g|^2 \right) \mathrm{d}g_1 \mathrm{d}g_2 \mathrm{d}g_3,
\end{equation}
where $|g| = \sqrt{g_1^2+g_2^2+g_3^2}$, and $P_{l_2}^{-l_1}(\cdot)$ is the associated Legendre function.
Here $g_1, g_2, g_3$ can be considered as the coefficients of the Pauli matrices $\sigma_1, \sigma_2, \sigma_3$ when representing the elements in $\mathfrak{su}(2)$.
Since $\tr(\sigma_i \sigma_j)=2\delta_{ij}$, we get $s$ instead of $s/2$ as the parameter in the exponent of \eqref{eq:su2-s3}, which differs slightly from the $\operatorname{U}(1)$ theory.
The polynomial $Q_{l_3}^{l_2}(\cdot, \cdot)$, $l_3 \geqslant l_2 \geqslant 0$ has degree $l_3$, and is defined by the recurrence relations
\begin{displaymath}
  Q_{l_3+1}^{l_2}(x,y) = 2\sqrt{\frac{(l_3+2)(l_3+1)}{(l_3-l_2+1)(l_3+l_2+2)}} y Q_{l_3}^{l_2}(x,y) - \sqrt{\frac{(l_3+2)(l_2+l_3+1)(l_3-l_2)}{l_3(l_3-l_2+1)(l_3+l_2+2)}} Q_{l_3-1}^{l_2}(x,y)
\end{displaymath}
with the initial condition
\begin{displaymath}
  Q_{l_3}^{l_3} = \sqrt{\frac{(2l_3+2)!!}{(2l_3+1)!!}} \frac{x^{l_3}}{\sqrt{\pi}}.
\end{displaymath}
The recurrence relation shows that $Q_{l_3}^{l_2}(x,y)$ has the form $x^{l_2} R_{l_3}^{l_2}(y)$, where $R_{l_3}^{l_2}(\cdot)$ is a polynomial of degree $l_3 - l_2$.

For the three-dimensional integral \eqref{eq:su2-s3}, one can use spherical coordinates to further simplify it. Upon integrating out the spherical angles, we obtain 
\begin{displaymath}
  I_{l_1 l_2 l_3} = \begin{cases}
  \displaystyle \frac{\pi^2}{2^{2l_2}} \sqrt{\frac{2l_2+1}{4\pi}} \begin{pmatrix} l_2 \\ l_2/2 \end{pmatrix}^2
   \int_0^{+\infty} (\sin g)^{l_2} R_{l_3}^{l_2}(\cos g) \exp(-s g^2) \mathrm{d}g, & \text{if } l_1 = 0 \text{ and } l_2 \text{ is even}, \\
   0, & \text{otherwise}.
  \end{cases}
\end{displaymath}
When $l_1 = 0$ and $l_2$ is even, we define the polynomial $\overline{R}_{l_3}^{l_2}(x) = (1 - x^2)^{l_2/2} R_{l_3}^{l_2}(x)$. Then, we have
\begin{displaymath}
  I_{l_1 l_2 l_3} = \frac{\pi^2}{2^{2l_2}} \sqrt{\frac{2l_2+1}{4\pi}} \begin{pmatrix} l_2 \\ l_2/2 \end{pmatrix}^2 \int_0^{+\infty} \overline{R}_{l_3}^{l_2}(\cos g) \exp(-sg^2) \mathrm{d}g, \qquad l_1 = 0 \text{ and } l_2 \text{ is even}.
\end{displaymath}
The integral above is the linear combination of $s^{-1/2} e^{-\frac{k^2}{4s}}$, $k = 0,1,\cdots,l_3$. Thus, similar to \Cref{3DXYProp1} and \Cref{3DXYProp3}, we conclude that in the $\SU{2}$ theory, the regularized observable $\langle O \rangle_s$ has the form
\begin{equation}\label{eq:su2-regression}
    \langle O \rangle_s = \frac{\sum_{k=0}^{+\infty} a_k e^{-\frac{k}{4 s}}}{\sum_{k=0}^{+\infty} b_k e^{-\frac{k}{4 s}}}.
\end{equation}
This expression can then be used in the regression model via a truncation of both infinite series.

For the $\SU{n}$ theory with $n > 2$, we have not found a straightforward way to evaluate the integral \eqref{eq:sun-integral}. Instead, we conjecture that the form \eqref{eq:su2-regression} holds for all the $\SU{n}$ theories, and we will use the approximation
\begin{equation}\label{eq:sun-regression}
    \langle O \rangle_s = \frac{\sum_{k=0}^M a_k e^{-\frac{k}{4 s}}}{1 + \sum_{k=1}^M b_k e^{-\frac{k}{4 s}}}
\end{equation}
in our regression model when carrying out numerical tests.

\section{Applications in the lattice field theory}\label{sec:QCD}
We are now ready to carry out numerical simulations for the lattice field theories. For the $D$-dimensional lattice, the each lattice point is denoted by a periodic multi-index
\begin{equation} \label{eq:X}
x \in X := (\mathbb{Z}/l_0 \mathbb{Z}) \times \cdots \times (\mathbb{Z}/l_{D-1} \mathbb{Z}).
\end{equation}
where $l_i$ refers to the length of the lattice in the $(i+1)^{\text{th}}$ component of $x$.
For a scalar field $\{\phi\}$, the variables will be denoted as $\phi_x$; for a vector field $\{U\}$, we denote the variables using $U_{x,\mu}$, where $\mu \in \{0,1,\cdots,D-1\}$. For simplicity, we use $x \pm \hat{\mu}$ to denote the multi-index that adds/subtracts the $\mu$th component of $x$ by $1$. For instance,
\begin{equation}
x + \hat{0} = (x_0+1, x_1, \cdots, x_{D-1}), \quad
x - \hat{1} = (x_0, x_1-1, \cdots, x_{D-1}).
\end{equation}
In what follows, three models in the lattice field theory will be studied. In order to achieve better results from regression, we will derive more suitable regression models for specific problems whenever possible.

\subsection{3D XY model}
For the 3D XY model, we have $D = 3$ and its action reads
\begin{equation}\label{eq:3dxy-action}
    S(\{\phi\}) = -\beta \sum_{x \in X} \sum_{\nu=0}^2 \cos(\phi_x - \phi_{x+\hat{\nu}} - \ii \mu \delta_{\nu,0}),
\end{equation}
where the field variables $\phi_x\in \mathbb{T}$, and $\mu$ is the chemical potential. Without loss of generality, we will set the length of the lattice $l$ in each component to be equal, that is $l = l_i$ for $i \in \{0,1,2\}$. For a scalar field, the size of the lattice, denoted by $N = l^3$, corresponds to the number of field variables we will be considering in our integration. When $\re \mu \neq 0$, the action becomes complex, and the complex Langevin method is applied to solve this model. This method fails even for small $\mu$ in this model \cite{aarts2010convergence, scherzer2020controlling}, and its failure was carefully analyzed in \cite{aarts2010convergence}, with the effect of the boundary terms begin discussed in \cite{scherzer2020controlling}.
Furthermore, according to our numerical experiments, the complex Langevin dynamics diverges even for a simple Euler-Maruyama method without special time-stepping techniques.
In addition, to impose the presented regularization method on this model, we simply add the regularization term $-\frac{s}{2}\sum_x\phi_x^2$ with $s>0$ as discussed in \Cref{sec:GenU1}.

For this model, the complex drift force corresponding to the variable $\phi_x$ upon regularization is given by 
\begin{equation}\label{eq:3dxy-drift}
    K_{x,s} = -\frac{\partial S_s}{\partial \phi_x} = -\beta \sum_{\nu=0}^2 \left[ \sin(\phi_x-\phi_{x+\hat\nu} -\ii \mu\delta_{\nu,0}) + \sin(\phi_x -\phi_{x-\hat\nu} + \ii \mu \delta_{\nu,0}) \right] - s \phi_x.
\end{equation}
We will mainly focus on two observables: the action density 
\begin{equation}\label{eq:3dxy-oa}
    \langle S \rangle = -\beta \frac{\partial\ln Z}{\partial \beta}  = -\beta\left\langle \sum_{x}\sum_{\nu=0}^2\cos(\phi_x-\phi_{x+\hat\nu}-\ii \mu\delta_{\nu,0})\right\rangle,
\end{equation}
and the number density
\begin{equation}\label{eq:3dxy-on}
    \langle n \rangle = \frac{\partial \ln Z}{\partial \mu} = \left\langle \ii \beta \sum_x \sin(\phi_x -\phi_{x+\hat 0}-\ii \mu) \right\rangle.
\end{equation} 

As part of the general framework for $\operatorname{U}(1)$ lattice field theory, we can expect Propositions \ref{3DXYProp1}, \ref{3DXYProp3}, \ref{3DXYProp4}, and \ref{3DXYProp5} to hold. What remains is to determine an equivalent interval of $s$ depending on a given $\beta$ and $\mu$ such that we are guaranteed correct convergence. In addition, for this specific model, we are able to derive a better regression model as compared to the general model in Proposition \ref{3DXYProp5}. The former is summarised in a Proposition that follows.

\begin{proposition}\label{3DXYProp2}
Let $\beta$ and $\mu$ be given. Let $\eta_0$ be the smallest real number such that the following inequality holds:
\begin{equation}\label{3DXY6}
(e^{|\mu|}+2)Y(\eta_0;\mu) - (e^{-|\mu|}+2)Y(\eta_0;\mu)^{-1} - \eta_0 \log \ml Y(\eta_0;\mu) \mr \leq 0
\end{equation}
whereby $$ Y(\eta;\mu) = \frac{\eta + \sqrt{\eta^2 - 4(e^{|\mu|}+2)(e^{-|\mu|}+2)}}{2(e^{|\mu|}+2)} .$$Then, if $s > 2\eta_0\beta$, The imaginary part of the field $\{\phi^I\}$ is bounded for any realization of the complex Langevin dynamics.
\end{proposition}

Here, by following the notation in \eqref{eq:real-Langevin} and \eqref{U(1)5}, the first index in the subscript of $K_{x,s}$ represents the drift term for the corresponding scalar field $\phi_x$, while $s$ in the second index indicates that this drift term is obtained from a regularized action.

\begin{proof}
We begin the proof by decomposing the drift term in \eqref{eq:3dxy-drift} as follows:
\begin{equation}\label{3DXY7}
\begin{aligned}
K^R_{x,s} =& -\beta \sum_{\nu = 0}^2 [\sin \ml \phi_x^R - \phi_{x+\hat\nu}^R\mr \cosh \ml \phi_x^I - \phi_{x+\hat\nu}^I - \mu \delta_{\nu,0}\mr \\
&+ \sin \ml \phi_x^R - \phi_{x-\hat\nu}^R\mr \cosh \ml \phi_x^I - \phi_{x-\hat\nu}^I + \mu \delta_{\nu,0}\mr] - s \phi_x^R, \\
K^I_{x,s} =& -\beta \sum_{\nu = 0}^2 [\cos \ml \phi_x^R - \phi_{x+\hat\nu}^R\mr \sinh \ml \phi_x^I - \phi_{x+\hat\nu}^I - \mu \delta_{\nu,0}\mr \\
&+ \cos \ml \phi_x^R - \phi_{x-\hat\nu}^R\mr \sinh \ml \phi_x^I - \phi_{x-\hat\nu}^I + \mu \delta_{\nu,0}\mr] - s \phi_x^I.
\end{aligned}
\end{equation}

Following a similar strategy as to Proposition \ref{U(1)Prop2} but in a generalized case, we need to show that for $s$ large enough, the support of the probability density in the imaginary variables $\phi^I$ is compact. Thus, we do so by first removing any dependence in $\{\phi^R \}$ by finding an upper bound for $K^I_{x,s}$ for a given $x \in X$ that holds for all $\phi_x^R$. An instructive upper bound for $K^I_{x,s}$ would be
\begin{equation}\label{3DXY8}
\begin{aligned}
K^I_{x,s} &\leq \beta \sum_{\nu = 0}^2 \ml \sinh(|\phi_x^I|+|\phi_{x+\hat\nu}^I| + |\mu\delta_{\nu,0}|) + \sinh(|\phi_x^I|  + |\phi_{x-\hat\nu}^I| + |\mu\delta_{\nu,0}|) \mr  - s\phi_x^I.
\end{aligned}
\end{equation}
Let $\overline{K^I_{x,s}}$ be the right-hand side of the inequality above, which only depends on the imaginary part of the field variables $\{\phi^I\}$. Note that $\overline{K^I_{x,s}}$ can be viewed as a function on $\mathbb{R}^N$. As a generalization to Proposition \ref{U(1)Prop2}, we will need to find an $N$-dimensional object $H$ that contains the origin\footnote{This condition is necessary since all our complex Langevin simulations always starts from the origin.} such that along the boundary of the object $\partial H$, we have\footnote{In the one-link model, the $1$- dimensional object is the line segment $[Y^-,Y^+]$ that contains the origin $0$, with the boundary being the points $y = Y^+$ and $y = Y^-$.}
\begin{equation}\label{3DXY29}
\sum_{x \in X} K_{x,s}^I \cdot \hat{n}_x|_{\partial H} < 0,
\end{equation}
where $\{\hat{n}\}$ represents the outward-oriented unit vector normal of the surface $\partial H$. 
For this proof, our choice of $H$ would be an $N$-dimensional hypercube centered at the origin with length $2C$ in each dimension. Here, we reserve the freedom of choice on $C > 0$ which would be chosen to close the proof for this proposition. We note that the surfaces of this hypercube are $(N-1)$-dimensional finite planes given by
\begin{equation}\label{3DXY9}
\Pi^{\pm}_x = \{ \{\phi^I\} \in \mathbb{R}^N | \phi_x^I = \pm C \text{ and } \forall y \neq x, |\phi^I_y| < C \}
\end{equation}
with a total of $2N$ of such planes, indexed by a sign on its superscript and $x \in X$ corresponding to the imaginary part of the field variable $\phi_x$ in which the value of $C$ is achieved. Note that since the outward oriented unit vector normal to $\Pi_x^\pm$ only has a component in the $x$-direction and that $K_{x,s}^I \leq \overline{K_{x,s}^I}$, it is sufficient to show that $\overline{K^I_{x,s}}(\{\phi|_{\Pi^+_x} \}) < 0$ if $\phi_x^I = + C$ and $\overline{K^I_{x,s}}(\{\phi|_{\Pi^+_x} \}) > 0$ if $\phi_x^I = - C$ for all $x$. Thus, applying the relevant inequalities from \eqref{3DXY9} for $\Pi_x^+$, we have
\begin{equation}\label{3DXY10}
\begin{aligned}
\overline{K^I_{x,s}}|_{\Pi_x^+} 
&< \beta \sum_{\nu = 0}^2 \ml \sinh(|\phi_x^I|+|\phi_{x+\hat\nu}^I| + |\mu\delta_{\nu,0}|) + \sinh(|\phi_x^I|-|\phi_{x-\hat\nu}^I| + |\mu\delta_{\nu,0}|) \mr  - s\phi_x^I\\
&<\beta (2 \sinh(2C + |\mu|) + 4 \sinh(2C)) - sC := K^I_{\text{upper}}.
\end{aligned}
\end{equation}
Thus, $\overline{K^I_{x,s}}|_{\Pi_x^+} < 0$ if $K^I_{\text{upper}} \leq 0$. Thus, for a given $\beta, s$ and $\mu$, we can conduct a one variable optimization on $C$ and deduce that $K^I_{\text{upper}}$ is minimized at $C = C_1$ with
\begin{equation}
C_1 = \frac{1}{2}\log\ml \frac{\frac{s}{2\beta} + \sqrt{\ml\frac{s}{2\beta}\mr^2-4(e^{|\mu|}+2)(e^{-|\mu|}+2)}}{2(e^{|\mu|}+2)} \mr.
\end{equation}
A remark here is that minimization is consistent with the fact that we can utilize the freedom of choice of $C$ for any given $\beta,\mu$ and $s$. Thus, we want to lower the upper bound of $K^I_{\text{upper}}$ as much as possible so that $K^I_{\text{upper}} < 0$ can be achieved with a smaller value of $s$. This is analogous to picking the choice of $y_0$ in the $\operatorname{U}(1)$ one-link model in Proposition \ref{U(1)Prop2}. Substituting $C_1$ to the expression of $\overline{K^I_{\text{upper}}}|_{\Pi_x^+}$ in \eqref{3DXY10} yields
\begin{equation}\label{3DXY11}
K^I_{\text{upper}}|_{\Pi_x^+} = \beta \ml (e^{|\mu|}+2)Y(\eta;\mu) - (e^{-|\mu|}+2)Y(\eta;\mu)^{-1} - \eta \log \ml Y(\eta;\mu) \mr \mr
\end{equation}
where $Y(\eta;\mu)$ is as defined in \eqref{3DXY6} and $\eta := \frac{s}{2\beta}.$ Thus, demanding $K^I_{\text{upper}} \leq 0$ is equivalent to solving the inequality as described in \eqref{3DXY6}. In other words, we describe the minimum value in which \eqref{3DXY6} holds as $\eta_0$. Thus, this is equivalent to $$ \eta = \frac{s}{2\beta} \geq \eta_0,$$ and thus $s \geq 2\eta_0 \beta$. Alternatively, we can just choose $s > 2\eta_0 \beta$. A symmetric and analogous argument holds for the case of $\Pi^-_x$ and for all $x$. This concludes the proof.
\end{proof}

Furthermore, as promised, we will attempt to obtain a better regression model as compared to that in \eqref{3DXY18}, summarized in the following proposition:

\begin{proposition}\label{3DXYProp6}
We consider the 3D XY model for the action density and the number density observables as defined in \eqref{eq:3dxy-oa} and \eqref{eq:3dxy-on}. For both observables, we can improve the representation of $\langle O \rangle_s$ to
\begin{equation}\label{3DXY19}
\langle O \rangle_s = \frac{\sum_{k=0}^\infty a_k e^{-\frac{k}{s}}}{\sum_{k=0}^\infty b_k e^{-\frac{k}{s}}}.
\end{equation}
\end{proposition}

\begin{proof}
It is obvious that both observables and the action are $C^{\infty}$ functions. Thus the derivations in the proof of \Cref{3DXYProp1} work in the 3D XY model. 
By observing that the action $S$ and the number density $n$ on a discrete lattice takes the form of a difference in two neighboring field variables $\phi_x - \phi_{x+\hat{\nu}}$ as compared to individual field variables, we thus rework some of the steps in \eqref{3DXY3}. For the 3D XY model, with the action given in \eqref{eq:3dxy-action}, we can rewrite the exponentiation of the negative of it as 
\begin{equation}\label{3DXY20}
\begin{aligned}
e^{-S(\{\phi\})} &= \prod_{x\in X} \prod_{\nu=0}^{2} e^{\beta \cos(\phi_x - \phi_{x+\hat\nu} - \ii \mu \delta_{\nu,0}) } \\
&= \prod_{x\in X} \prod_{\nu=0}^{2} \sum_{m \in \mathbb{Z}} [\ii^m e^{m\mu \delta_{\nu,0}} J_m(-\ii \beta)] e^{\ii m(\phi_x - \phi_{x+\hat\nu})}
=  \sum_{\vec{k} \in \mathcal{K}} \beta_{\vec k} e^{\ii\sum_{x\in X}k_x \phi_x},
\end{aligned}
\end{equation}
where we have again used the Jacobi-Anger expansion \eqref{U(1)9} and
the resultant Fourier series with its coefficients represented by $\beta_{\vec k}$ is analogous to that in the second line of equation \eqref{3DXY3}.
Here the range of the summation $\mathcal{K}$ is given as a proper subset of $\mathbb{Z}^N$ defined by
\begin{displaymath}
\mathcal{K} = \{ \vec{k} \in \mathbb{Z}^N: \sum_{x\in X} k_x = 0 \}.
\end{displaymath}
The appearance of $\mathcal{K}$ can be explained as follows. Since the inner sum represents different Fourier frequency modes for $\phi_x - \phi_{x+\hat{\nu}}$, with each satisfying the fact that $\sum_{x \in X}k_x = 0$, and the fact that the product of exponentials is given by the exponential of the sum of the individual arguments, such a property is preserved in the resulting Fourier series, and thus have the given property for $\mathcal{K}$.
Here, we note that the observables $n$ or $S$ follow a similar structure. This thus implies that, analogous to \eqref{3DXY14} and \eqref{3DXY15}, we have
\begin{equation}\label{3DXY21}
\langle O \rangle_s = \frac{1}{Z_s} \sum_{\vec{m} \in \mathcal{K}}\sum_{\vec{k} \in \mathcal{K}}\widehat{O}_{\vec m} \beta_{\vec k} \exp \left(-\frac{1}{2s} \sum_{x \in X} (m_x+k_x)^2 \right), \qquad
Z_s = \sum_{\vec{k} \in \mathcal{K}}\beta_{\vec k}  \exp \left( -\frac{1}{2s} \sum_{x \in X} k_x^2 \right).
\end{equation}

For every $\vec{k} \in \mathbb{Z}^N$, the sum $\sum_{x\in X} k_x^2$ is odd/even if and only if $\sum_{x \in X} k_x$ is odd/even since $k_x^2$ and $k_x$ always share the same parity.
Then for $\vec{k} \in \mathcal{K}$, we know that $\sum_{x \in X} k_x^2$ is even so that in \eqref{3DXY21}, each exponential term in the series expansion of $Z_s$ has the form $\exp(-l/s)$ for some $l \in \mathbb{Z}$, where the factor $2$ drops off by reduction of the fraction. The same can be done to the expansion of $\langle O \rangle_s$ since $\vec{m} + \vec{k} \in \mathcal{K}$, which leaves with us the expression in \eqref{3DXY19}. 
\end{proof}
\begin{remark}
We note that the concept of counting the parity of $k_x$ and $m_x$ is inspired by the representation of the partition function into bonds as proposed by the Worm Algorithm in \cite{prokofev2010worm} and \cite{banerjee2010worm}.
\end{remark}

With reference to the above proposition, the appropriate regression model is thus given by 

\begin{equation}\label{3DXY23}
\langle O \rangle_s = \frac{\sum_{k=0}^M a_k e^{-\frac{k}{s}}}{1 + \sum_{k=1}^M b_k e^{-\frac{k}{s}}}
\end{equation}
for a given integer $M$.

We apply the above regression model to two examples on a $8\times 8\times 8$ lattice with parameters $\beta=0.2,\mu=\sqrt{0.1}$ and $\beta=0.7,\mu=\sqrt{0.1}$ respectively, for the two observables of interest, the number density scaled with the chemical potential $\mu n$, and the action density $S$. Similar to the $U(1)$ one-link model, we will compare the numerical results obtained with that from standard methods from current literature. The results for the $2R$ method are summarized in \Cref{fig:3dxy-reg-extra} and Table \ref{tab:3DXYNumResults}.

\begin{table}[htbp]
\centering
\caption{Estimates for $\re{\mu \langle n \rangle}$ and $\re{\langle S\rangle }$ using various methods from the current literature and from our $2R$ method as in \Cref{fig:3dxy-reg-extra}.}
\label{tab:3DXYNumResults}
\begin{tabular}{|c|c|c|c|c|}
\hline 
    &\multicolumn{2}{c|}{ $\beta = 0.2$, $\mu = \sqrt{0.1}$} &    \multicolumn{2}{c|}{ $\beta = 0.7$, $\mu = \sqrt{0.1}$} \\
\hline
   & $\re{\mu\langle n \rangle}$ & $\re{\langle S \rangle}$ & $\re{\mu\langle n \rangle}$ & $\re{\langle S \rangle}$\\  
\hline
    Original complex Langevin method from \cite{scherzer2020controlling} & $0.001840$  & $-0.07929$ & $0.04926$ & $-1.5268$\\
    Corrected complex Langevin Method from \cite{scherzer2020controlling} & $0.0003858$ & $-0.06716$ & $0.04905$ & $-1.5240$\\
    Worldline Method from \cite{scherzer2020controlling} & $1.5495 \times 10^{-7}$ &  $-0.06230$ & $0.04898$ & $ -1.5240$ \\
    Best for $2R$ Method & $0.001641$ & $-0.08791$ & $0.05234$ & $-1.6052$ \\
\hline
    &\multicolumn{4}{c|}{$2R$ Method with different values of $M$} \\
\hline
    $M = 3$ & $0.001641$ & $-0.10397$ & $0.05307$ & $-1.6170$\\ 
    $M = 4$ & $0.001797$ & $-0.1046$ & $0.05303$ & $-1.6191$\\
    $M = 5$ & $0.001849$ & $-0.1069$ & $0.05277$ & $-1.6052$\\
    $M = 6$ & $0.002455$ & $-0.1097$ & $0.05315$ & $-1.6151$\\
    $M = 7$ & $0.002793$ & $-0.08791$ & $0.05234$ & $-1.6126$\\
\hline
\end{tabular}
\end{table}

Here, we summarize some of the key points from Table \ref{tab:3DXYNumResults} as follows. First, we note that the range of $s$ used might not be consistent with what is obtained in Proposition \ref{3DXYProp5}; $s_0 > 2 \eta_0 \beta  = 2(7.8)(0.2) = 3.12$ for $\beta = 0.2$ and $\mu = \sqrt{0.1}$ while we have included points as close as $s = 0.6$ in our regression. Next, we note that the results obtained from our $2R$ method are similar to those obtained from the original complex Langevin method. Although they might not do as well as compared to the corrected complex Langevin method from \cite{scherzer2020controlling} and the worldline method, the results obtained are still not too far off from these methods. Furthermore, we note that the results obtained using different values of $M$ for the $2R$ method are generally stable, while according to our simulations, the original complex Langevin method for $s = 0$ suffers from instabilities, and the result of the original complex Langevin method from \cite{scherzer2020controlling} may require adaptive time-stepping. Nonetheless, one should note that by picking $s = 0.6$, a value which is way off from our guaranteed region of $s > s_0$, there might be an unquantifiable bias that possibly grows as we pick values of $s$ closer to $0$. Such a possibility is inferred for the case of the $U(1)$ one-link model, in which the difference between the true value and the complex Langevin method grows significantly as $s$ falls below a certain threshold $s_1$ ($\approx 0.4 < s_0$) and gets larger as it approaches $0$. Despite the inability to obtain a result of better accuracy as compared to the worldline and the corrected complex Langevin method, the relatively simple structure and the improved generalizability of the $2R$ method might still serve as a method that we can use to corroborate with alternative methods in the current literature.

\begin{figure}
    \centering
    \includegraphics[width=0.45\textwidth]{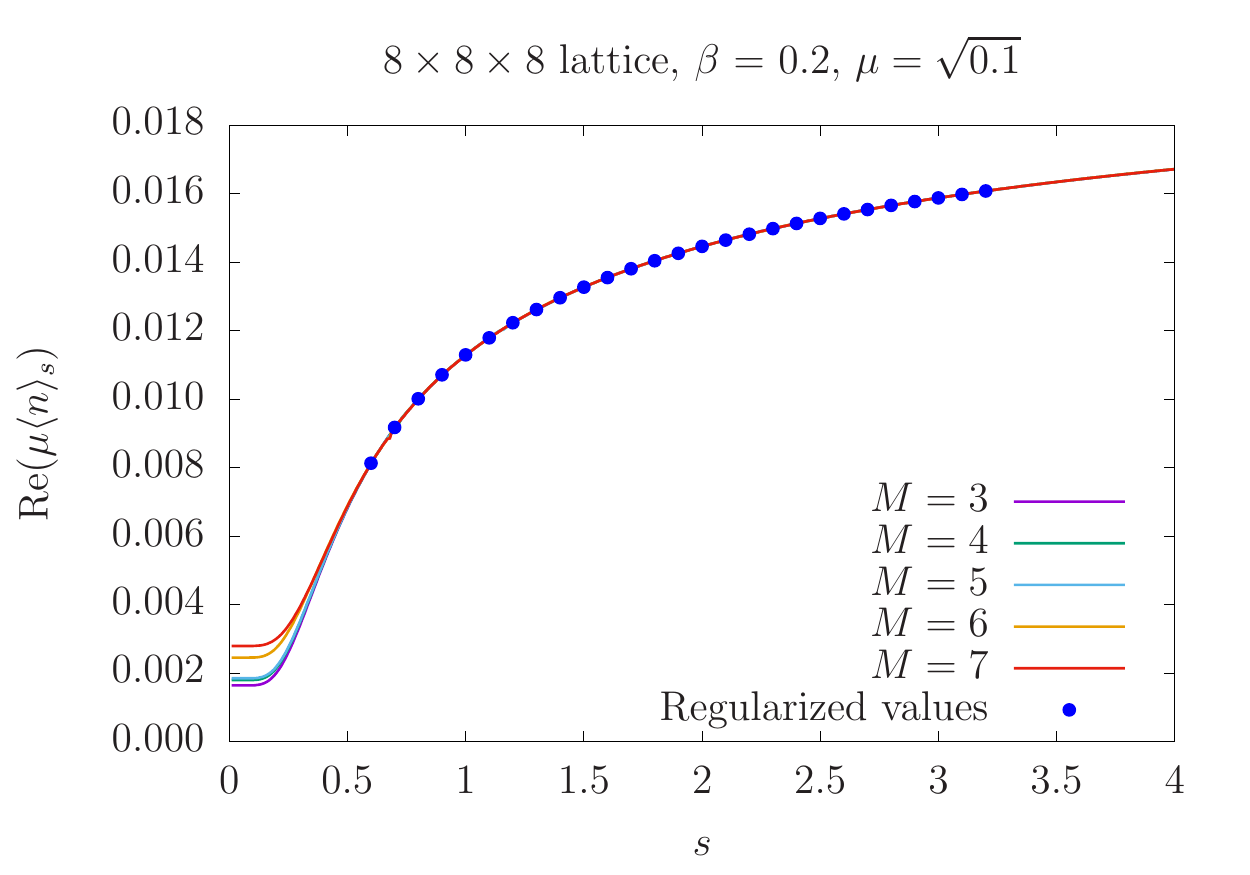}
    \includegraphics[width=0.45\textwidth]{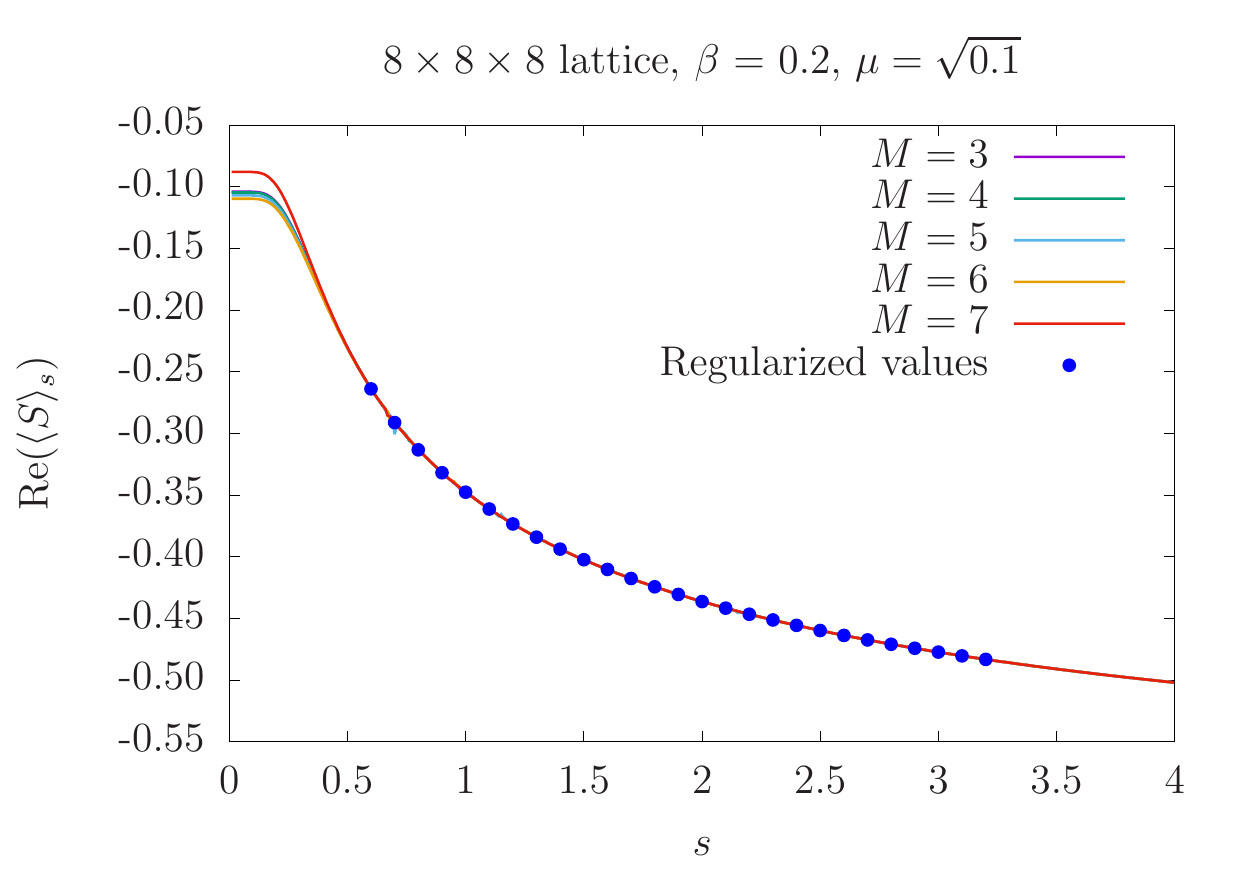}    
    \includegraphics[width=0.45\textwidth]{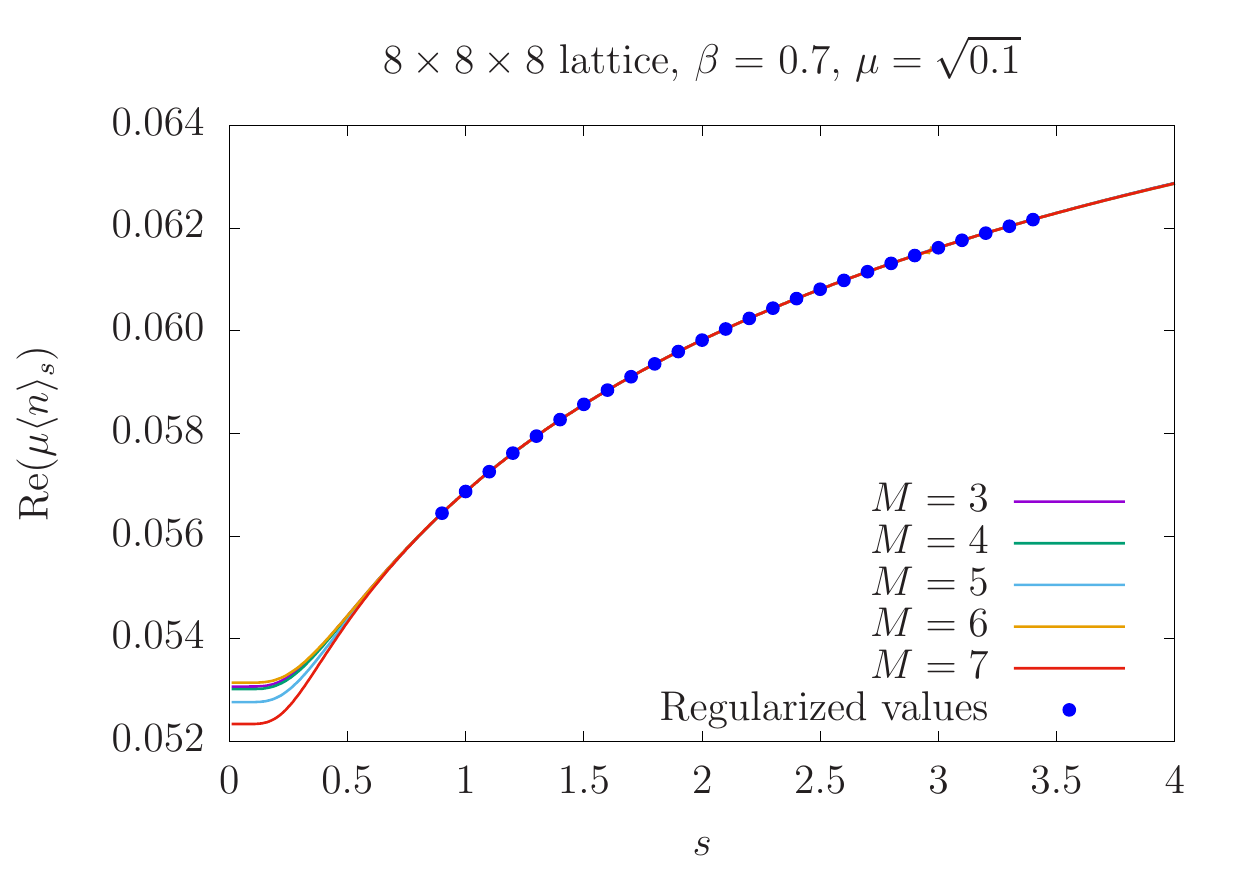}
    \includegraphics[width=0.45\textwidth]{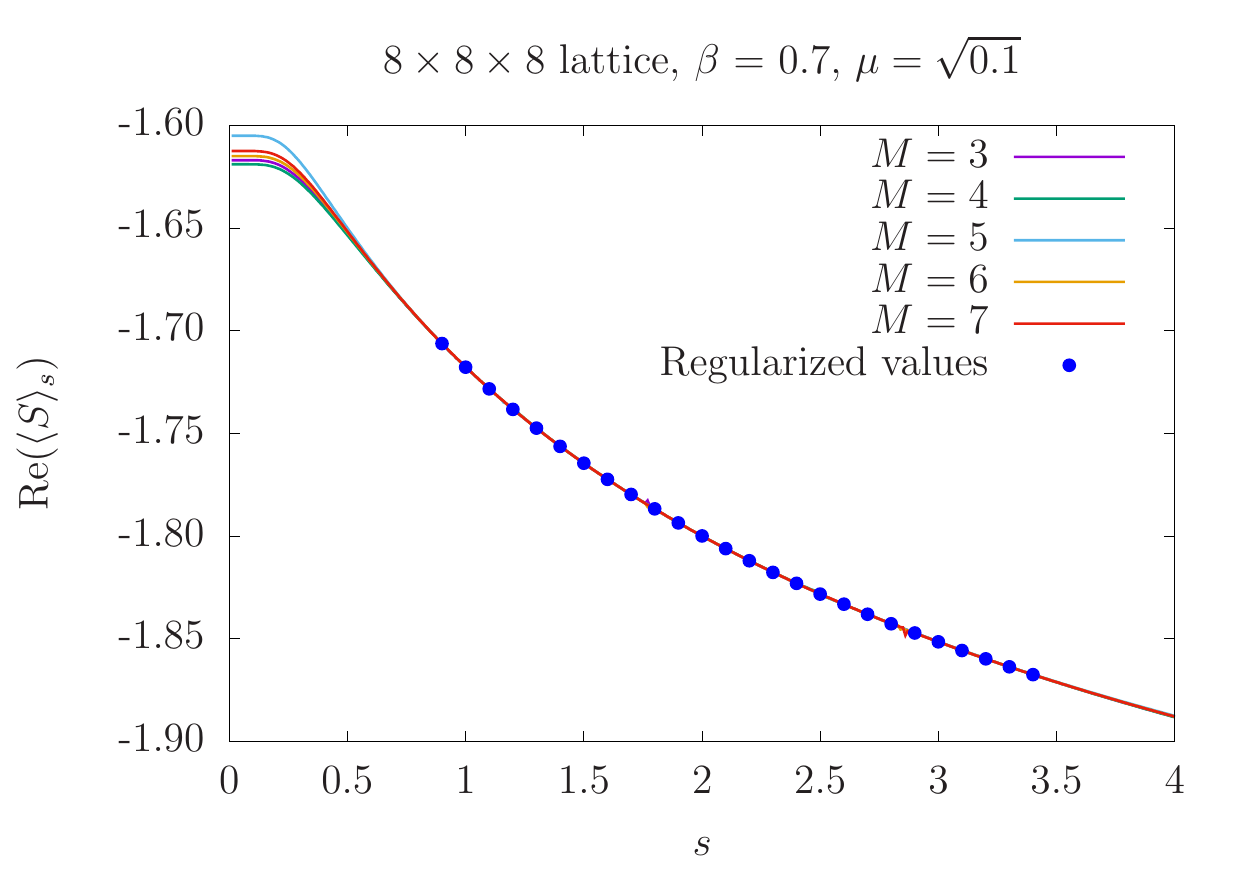}  
    \caption{Number density (left) and action density (right) for $\beta=0.2$ (top) and $0.7$ (bottom) of 3D XY model. $M$ represents the number of expansion terms used for extrapolation \eqref{3DXY19}.  \label{fig:3dxy-reg-extra}}
\end{figure}

Next, we shall explain the rationale of choosing the $2R$ method over the $3R$ method despite the latter having better success with the $U(1)$ one-link model. Recall that for our $3R$ method, we will have to choose a reference $s_0$. From \Cref{fig:3dxy-rew}, we can observe that the estimates for $\re{\mu \langle n \rangle}$ without regression is better for $s_0 = 0.6$ at $0.005666$ as compared to that in $s_0 = 3.2$ at $0.01498$. Here, better is defined as how close our results are to the results generated from the worldline method, at $\sim 0$ for $\re{ \mu \langle n\rangle}$ at $\beta = 0.2$, $\mu = \sqrt{0.1}$ and an estimate without regression is obtained by quoting the value at $s = 0.1$ directly for an estimate as regression including this point will likely not predict the value at $s = 0$ to be too far off from it. The reason for this difference might be as follows. We note that there is a trade-off between errors arising from regression and errors arising from having $s_0$ that is not sufficiently large enough to guarantee possible correct convergence. In this case here, $s_0 = 3.2$ is too far off from our point of interest at $s = 0$, and thus may result in a large error resulting from regression. This error might be larger than the error arising from inaccurate simulations with $s_0 = 0.6$ and thus accounts for such a phenomenon. Nonetheless, the values obtained for both choices of $s_0$ fail to surpass that obtained from the $2R$ method with even the worst value at $M = 7$ at $0.002793$.

\begin{figure}
    \centering
    \includegraphics[width=0.45\textwidth]{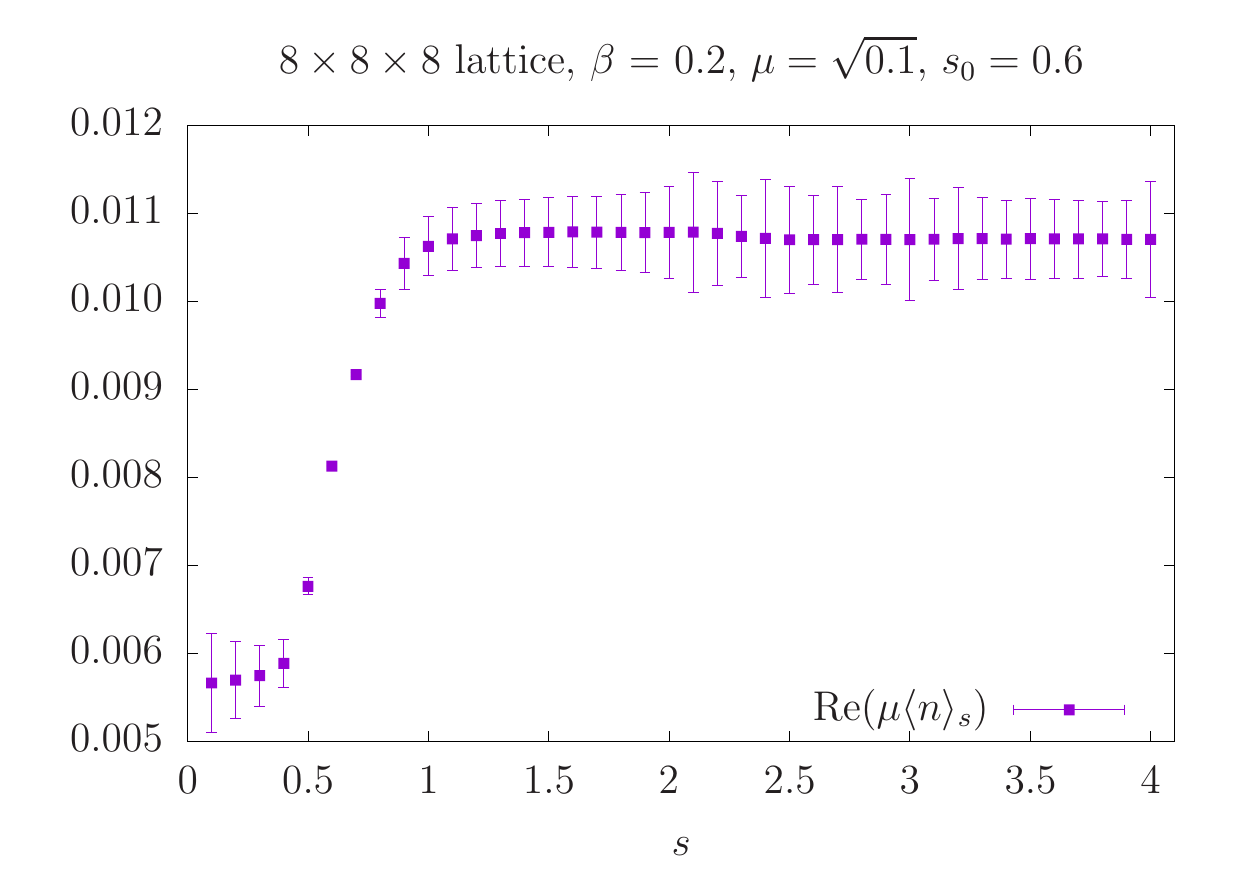}
    \includegraphics[width=0.45\textwidth]{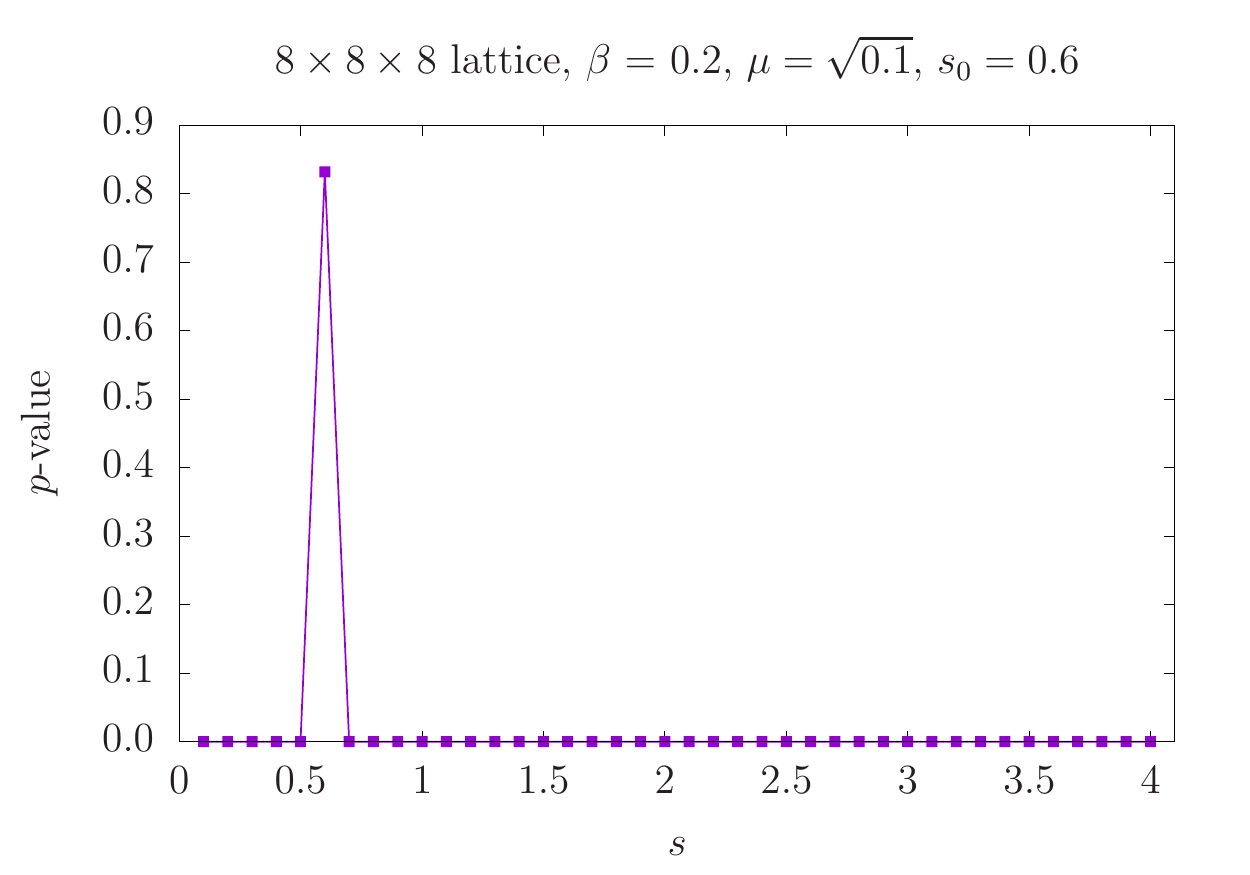}
    \includegraphics[width=0.45\textwidth]{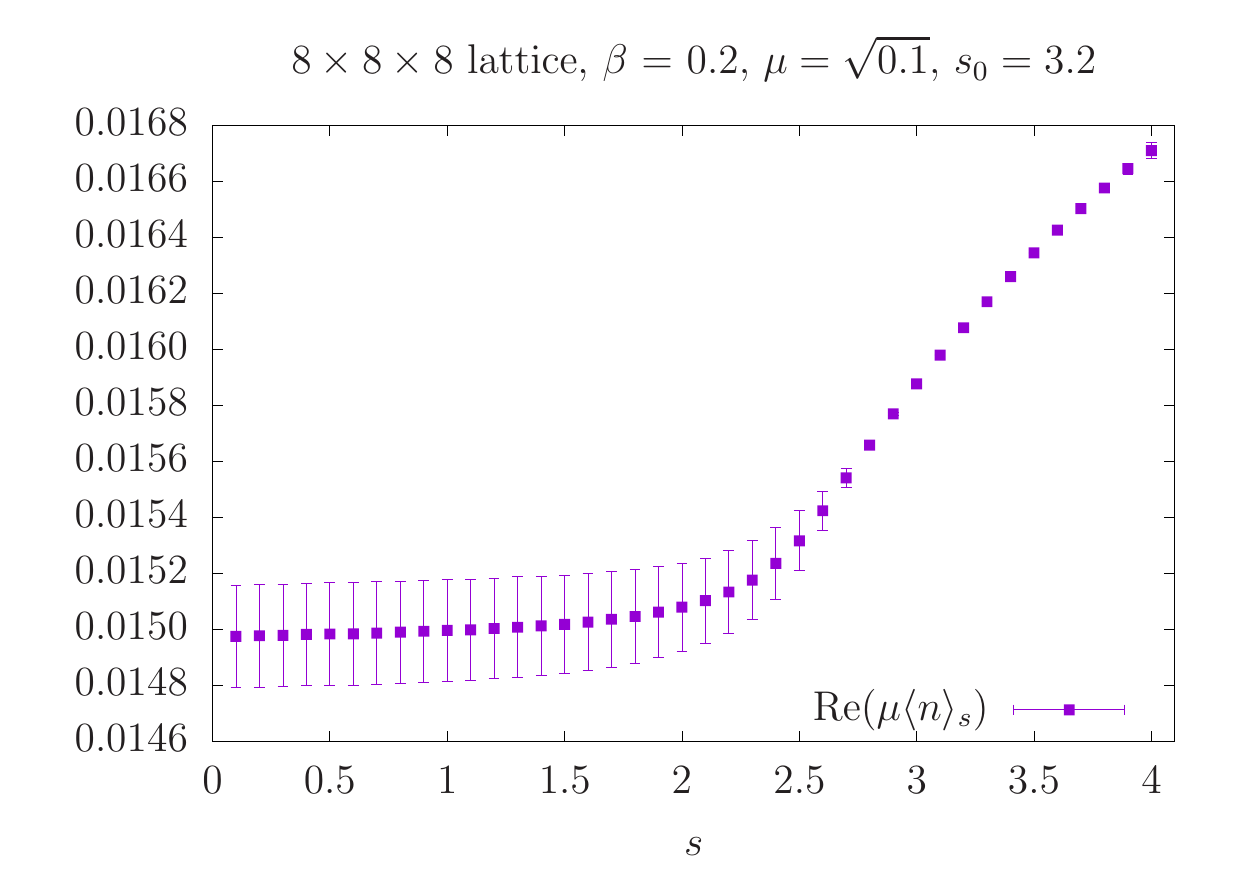}
    \includegraphics[width=0.45\textwidth]{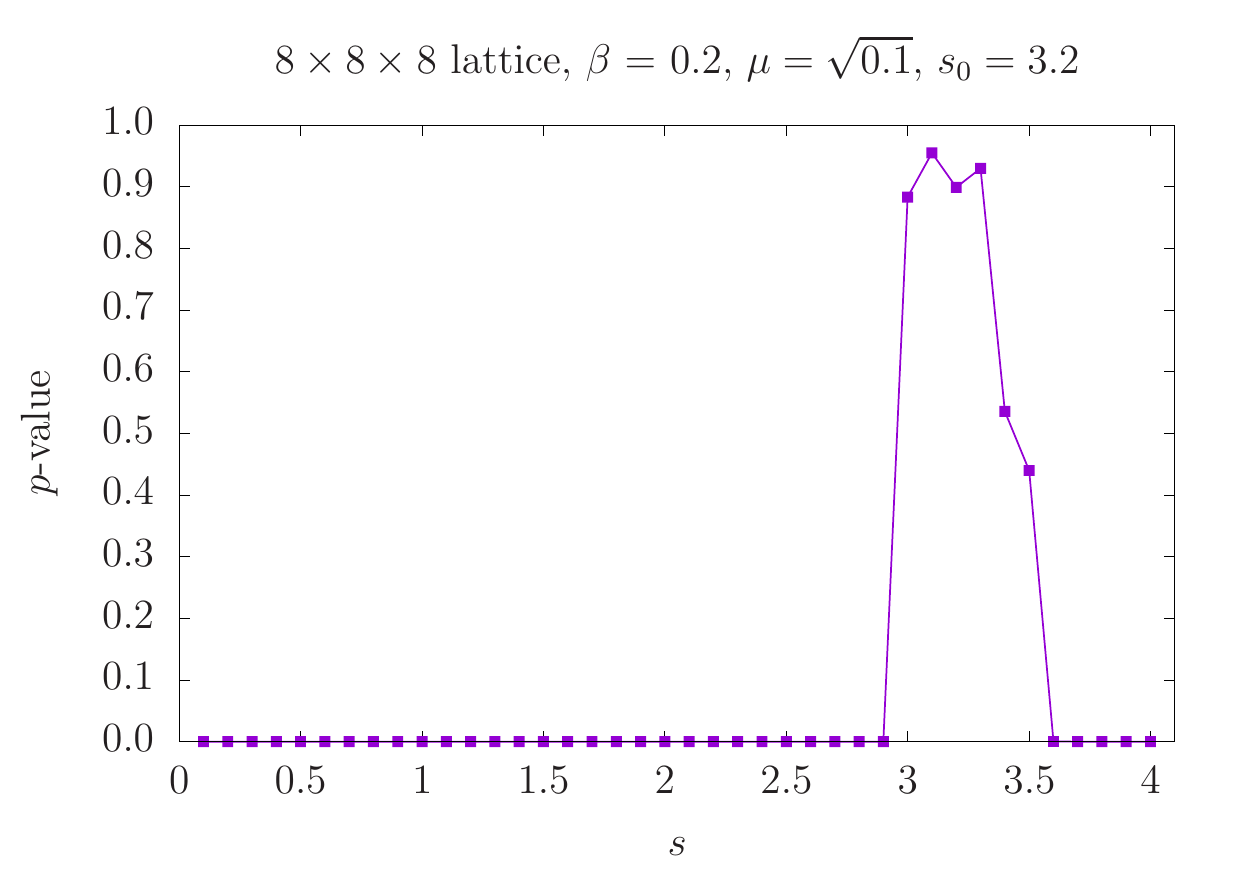}
    \caption{Reweighted results for $\beta=0.2$ with parameter $s_0 = 0.6$ (top) and $3.2$ (bottom).\label{fig:3dxy-rew}}
\end{figure}

On top of that, we note that the benefits of regression might be limited for the 3D XY model. As seen from the p-values on the right diagrams in \Cref{fig:3dxy-rew}, the phenomenon of infinite variance appears relatively quick for both cases, with p-value dropping to a value close to $0$ as soon as $s$ hits $0.5$ for $s_0 = 0.6$ and a similar phenomenon at ar $s = 2.9$ for $s_0 = 3.2$. Thus, restricting our regression points for which the p-value is at least $0.05$ would heavily restrict the number of points that can be used, and therefore reduces the prediction ability of the model at $s = 0$. This is also why we did not perform regression for this model, as the number of points might not be sufficient to determine the regression coefficients.

In both the $2R$ and the $3R$ method, a key ingredient would be the regression method in the final step. Thus, in view of improving results obtained from regression, we list down a few plausible explorations. One includes improving the regression method itself as it was done using ordinary regression techniques on highly non-linear models as such that in \eqref{3DXY23}, where the dependent variable only appears after a ratio of sums of the exponential of the inverse of its independent variable. Another possible exploration would be to reduce the distance of $s_0$ from $0$ by evaluating the expectation of a modified observable. Yet another possible exploration would be to combine the use of the corrected CL method with our $2R$ or $3R$ method. The interested reader is invited to try out some of these explorations in view of improving the accuracy of the proposed method.

\subsection{Polyakov chain model} 
In this subsection, we discuss the results for the one-dimensional Polyakov chain model \cite{cai2020how}, whose action is given by
\begin{equation}\label{eq:su3-polyakov}
S(\{U\})=-\operatorname{tr}\left(\beta_{1} U_{1} \cdots U_{N}+\beta_{2} U_{N}^{-1} \cdots U_{1}^{-1}\right), \quad \beta_1,\beta_2 \in \mathbb{R}.
\end{equation}
The observable of interest is $O_l(\{U\}) = \tr ([U_1\cdots U_N]^l)$ for $l \in \mathbb{Z}_+$. Due to the gauge invariance, this model can actually be reduced to the one-link model ($N = 1$) by gauge fixing \cite{seiler2013gauge}. Thus, to test the performance of the 2R method, we choose to simulate the original problem without fixing the gauge.

Note that this model works for both $\operatorname{U}(1)$ and $\SU{n}$ theories.
For the $\operatorname{U}(1)$ theory, the trace operator reduces to the identity operator.
Let $U_k = \exp(\ii \theta_k)$, $k = 1,\cdots,N$. Then, the expectation of the regularized observable can be represented by
\begin{equation} \label{eq:Ol_s}
\langle O_l \rangle_s = \frac{1}{Z_s} \int_{\mathbb{R}^N} e^{\ii l (\theta_1 + \cdots + \theta_N)} \exp \left(\beta_1 e^{\ii(\theta_1+\cdots+\theta_N)}+\beta_2 e^{-\ii(\theta_1+\cdots+\theta_N)} -\frac{s}{2} (\theta_1^2 + \cdots + \theta_N^2) \right) \mathrm{d}\theta_1 \cdots \mathrm{d}\theta_N,
\end{equation}
where
\begin{displaymath}
Z_s = \int_{\mathbb{R}^N} \exp \left(\beta_1 e^{\ii(\theta_1+\cdots+\theta_N)}+\beta_2 e^{-\ii(\theta_1+\cdots+\theta_N)} - \frac{s}{2} (\theta_1^2 + \cdots + \theta_N^2) \right) \mathrm{d}\theta_1 \cdots \mathrm{d}\theta_N.
\end{displaymath}
By the series expansion of the exponential function, we get that
\begin{displaymath}
Z_s = \int_{\mathbb{R}^N} \sum_{j=-\infty}^{+\infty} \alpha_j e^{\ii j (\theta_1 + \cdots + \theta_N)} \exp\left(-\frac{s}{2} (\theta_1^2 + \cdots + \theta_N^2) \right) \mathrm{d}\theta_1 \cdots \mathrm{d}\theta_N = \left( \frac{2\pi}{s} \right)^{N/2} \sum_{j=-\infty}^{+\infty} \alpha_j \exp\left( -\frac{j^2 N}{2s} \right),
\end{displaymath}
where
\begin{displaymath}
\alpha_j = \sum_{k=\max(0,j)}^{+\infty} \frac{\beta_1^k \beta_2^{k-j}}{k! (k-j)!}, \qquad j \in \mathbb{Z}.
\end{displaymath}
Similarly, the regularized observable $\langle O_l\rangle_s$  given in \eqref{eq:Ol_s} can be expanded as
\begin{equation}
\langle O_l \rangle_s = \frac{1}{Z_s} \left( \frac{2\pi}{s} \right)^{N/2} \sum_{j=-\infty}^{+\infty} \alpha_j \exp\left( -\frac{(j+l)^2 N}{2s} \right).
\end{equation}
Inspired by the analysis above, one can choose the expression
\begin{equation} \label{eq:polyakov-U1}
\langle O_l \rangle_s = \frac{\sum_{k=0}^M a_k e^{-\frac{N k^2}{2s}}}{\sum_{k=0}^M b_k e^{-\frac{N k^2}{2s}}}
\end{equation}
in the regression. For $\SU{n}$ theories, a similar expression will be used, while $s$ will be replaced with $2s$ due to the reason stated in \Cref{sec:sun-reg}.

For the general $\SU{n}$ theory, the Lie derivative of \eqref{eq:su3-polyakov} can be derived as
\begin{equation*}
  D_{a,k}S(\{U\}) = -\ii\beta_1\operatorname{tr}(U_1 \cdots U_{k-1} \lambda_aU_k\cdots U_N) + \ii\beta_2
  \operatorname{tr}(U_N^{-1}\cdots U_k^{-1} \lambda_a U_{k-1}^{-1}\cdots U_1^{-1}).
\end{equation*}
In our simulation, we would like to focus on the $\SU{3}$ model, which has also been studied in \cite{seiler2013gauge, cai2020how}. Following \cite{seiler2013gauge}, we choose $\beta_1$ and $\beta_2$ to be $2.27$ and $2.04$, respectively.
In this case, the exact values of $\langle O_l \rangle$ for $l = 1,2,3$ obtained were
\begin{displaymath}
\langle O_1 \rangle = 2.0957, \quad
\langle O_2 \rangle = 0.3761, \quad
\langle O_3 \rangle = -0.5269,
\end{displaymath}
which have been calculated in \cite{cai2020how}. For all the numerical tests, we chose the fixed time step $\Delta t = 5 \times 10^{-4}$ and simulate the complex Langevin dynamics up to $T = 2$. Then, we drew one sample for every 10 time steps until 20 million samples were collected. These samples were used to estimate the expectations of the observables.

Following \cite{dong2020alternating}, we use the quantity
\begin{equation} \label{eq:DF}
  \Delta F = \frac{1}{N}\sum_{i=1}^N \tr (U_i U_i^{\dagger} - I)
\end{equation}
to measure the extent of excursion away from $[\SU{3}]^N$.
Here $I$ stands for the $3 \times 3$ identity matrix.
Results for some values of $s$ and $N$ are given in \Cref{fig:polyakov-deltaF}. For instance, when $N = 16$, the complex Langevin dynamics quickly diverges if no regularization is applied. Even when $s = 4$, we can still observe a few spikes of the curve at the magnitude of $10^{-4}$, which may indicate convergent but biased results. For $s = 8$ and $16$, the deviation from $\SU{3}$ is well suppressed, so that the regularized observables computed from the samples are likely to be reliable. However, as $N$ increases, the regularization for $s = 8$ may no longer be sufficient. The middle diagram of \Cref{fig:polyakov-deltaF} shows that for $s = 8$, the complex Langevin dynamics fails to converge when $N = 64$ and $128$. Even with $N = 32$,  the few spikes on the curve of $\Delta F$ may imply possibly biased results. A reasonable modification appears to be setting $s$ to be proportional to $N$, as displayed in the right diagram of \Cref{fig:polyakov-deltaF}. This agrees with the analysis for \eqref{eq:polyakov-U1}, in which $s$ also scales with $N$.

\begin{figure}[!ht]
    \centering
    \includegraphics[width=0.32\textwidth]{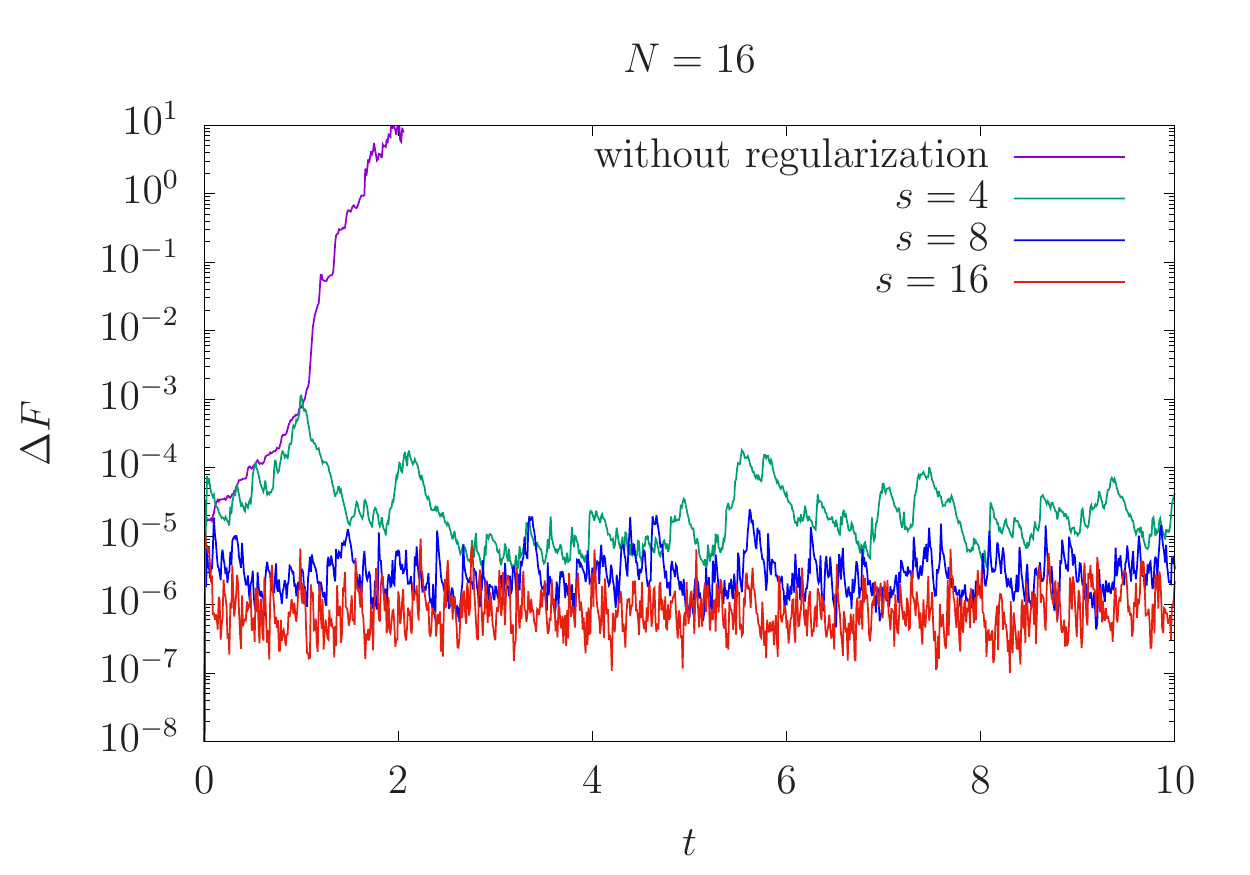}
    \includegraphics[width=0.32\textwidth]{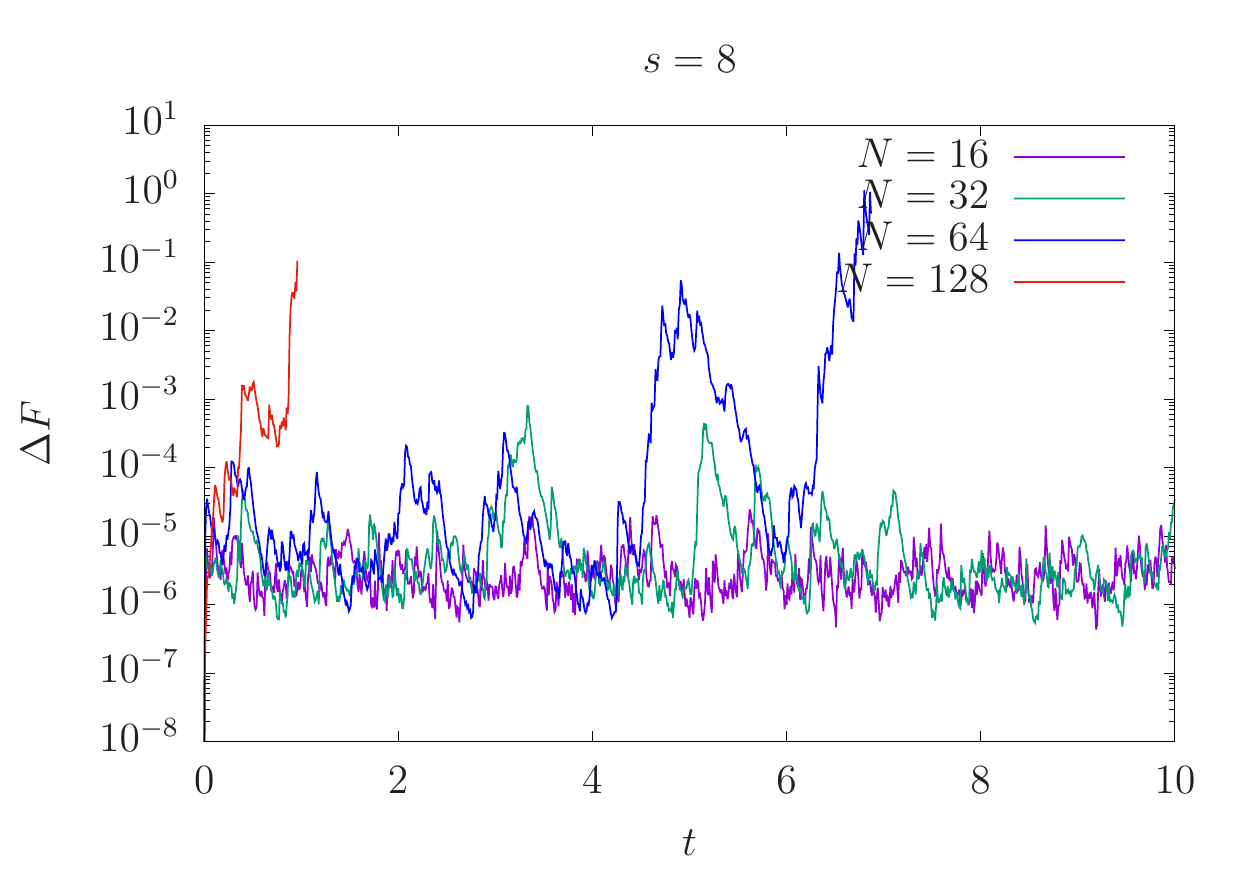}
    \includegraphics[width=0.32\textwidth]{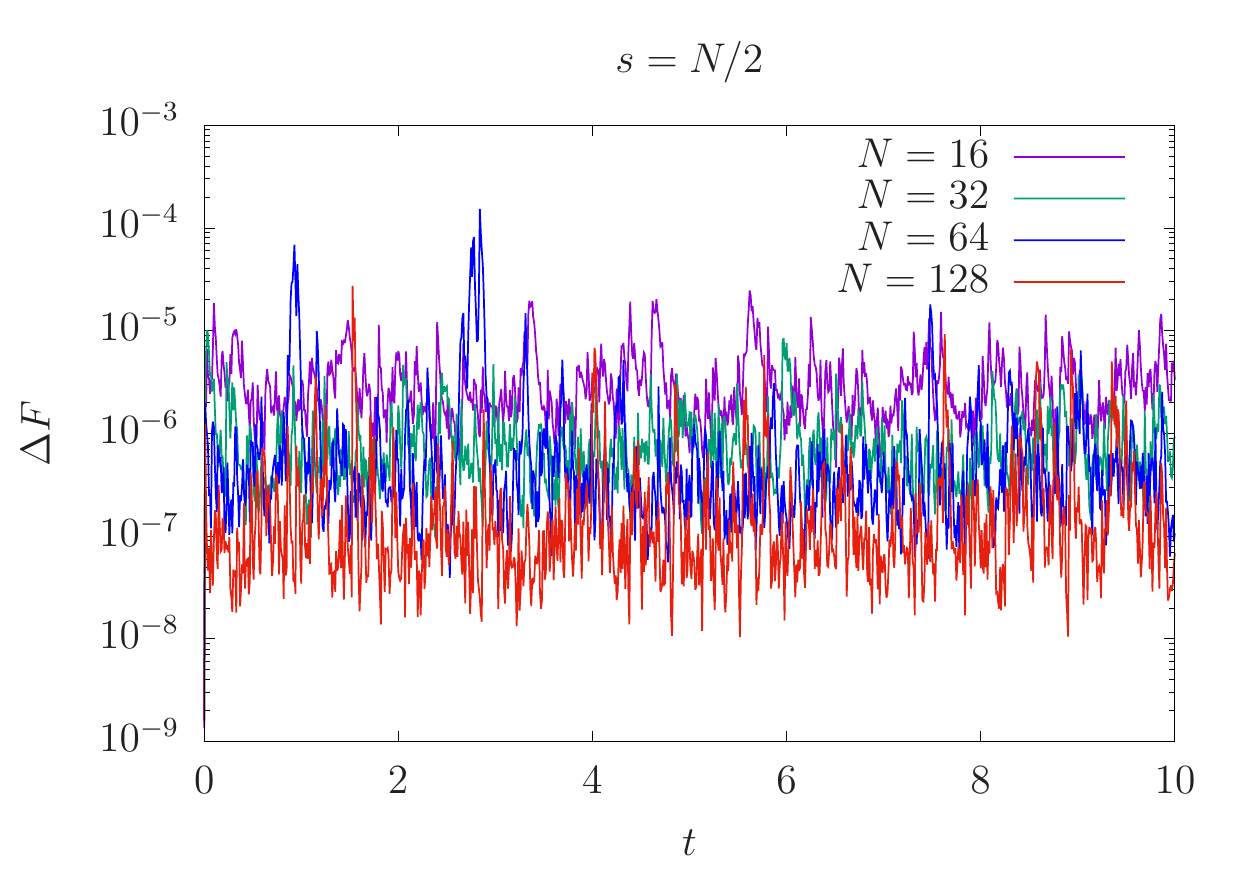}
    \caption{Restriction of regularization method $\Delta F$. Left: different $s$ for $N=16$. Middle: Fixed $s=8$ for $N=16,32,64,128$. Right: $s=N/2$.}
    \label{fig:polyakov-deltaF}
\end{figure}

We now focus on the case $N = 16$. For the $2R$ method, we display the results for $s$ ranging from $4$ to $34$ in \Cref{fig:polyakov-reg}. The dashed horizontal line denotes the reference solution. As previously observed, smaller values of $s$ will lead to unstable complex Langevin dynamics. In these examples, the estimated values of $\langle O_l \rangle_s$ for $s = 4$ and $l = 1,2,3$ are very close to the exact solution. These indicate the existence of an example where the regularization can work well without further corrections. However, the reliability of this method is hard to judge if the exact solution is unknown. 

\begin{figure}[!ht]
    \centering
    \includegraphics[width=0.32\textwidth]{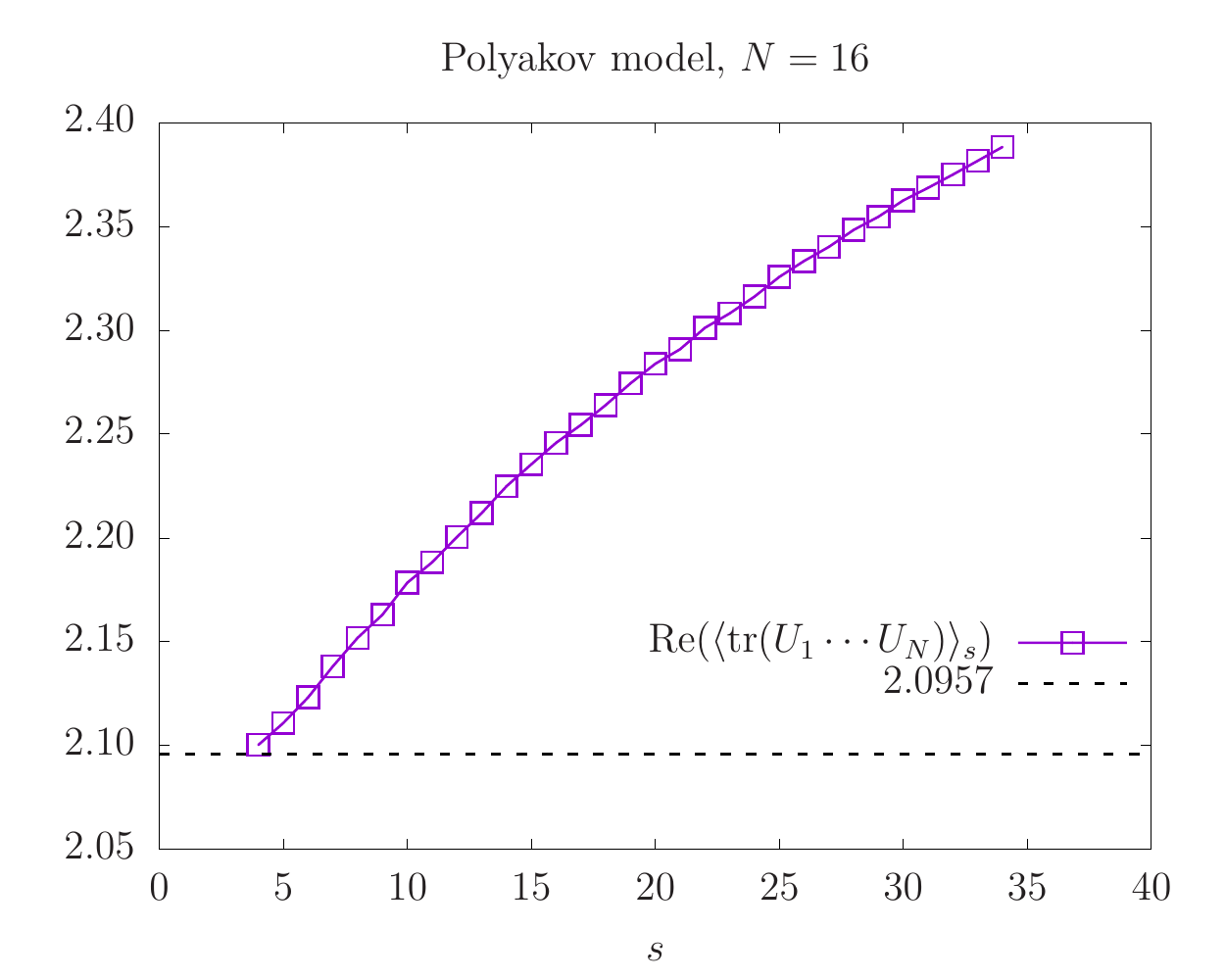}
    \includegraphics[width=0.32\textwidth]{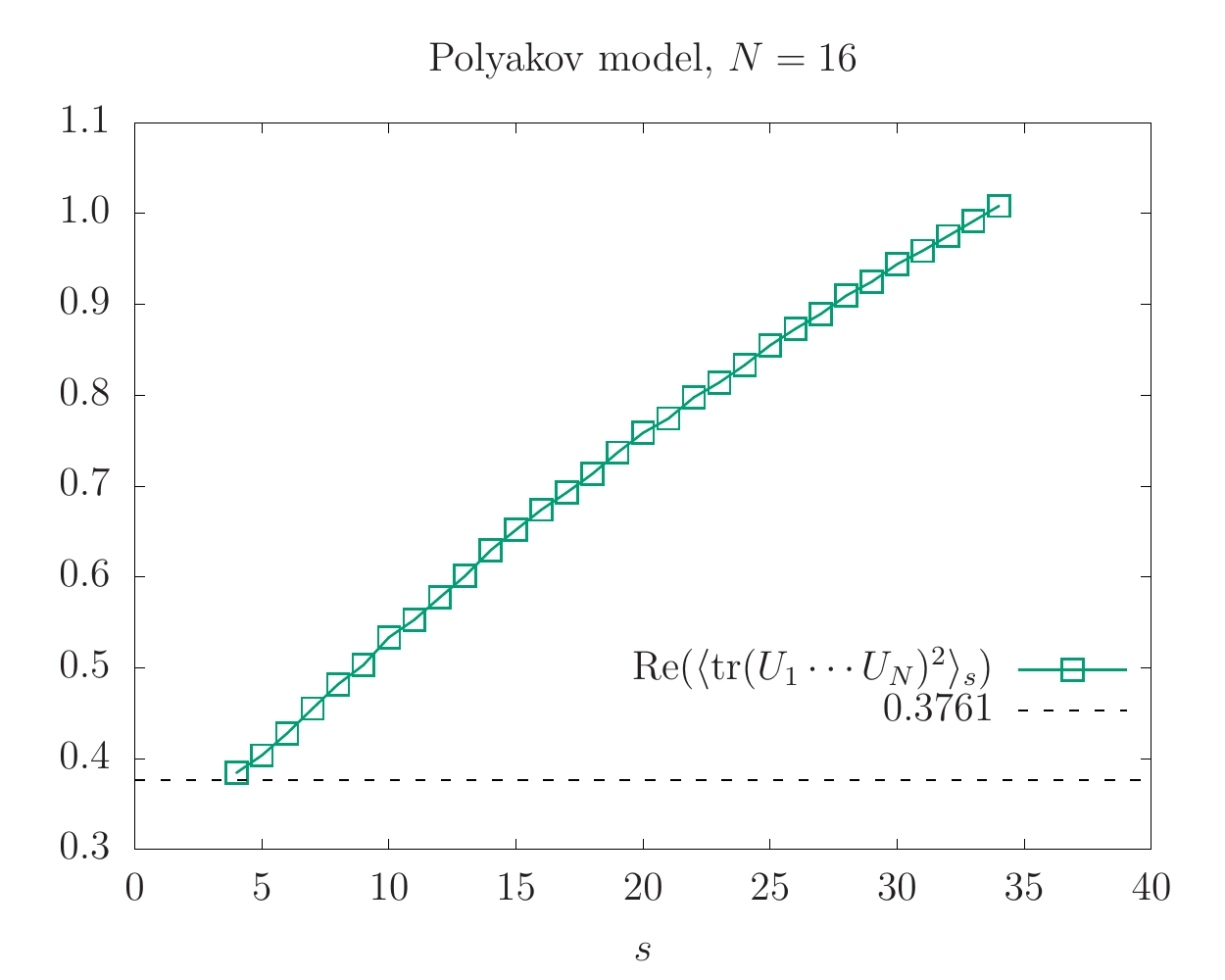}
    \includegraphics[width=0.32\textwidth]{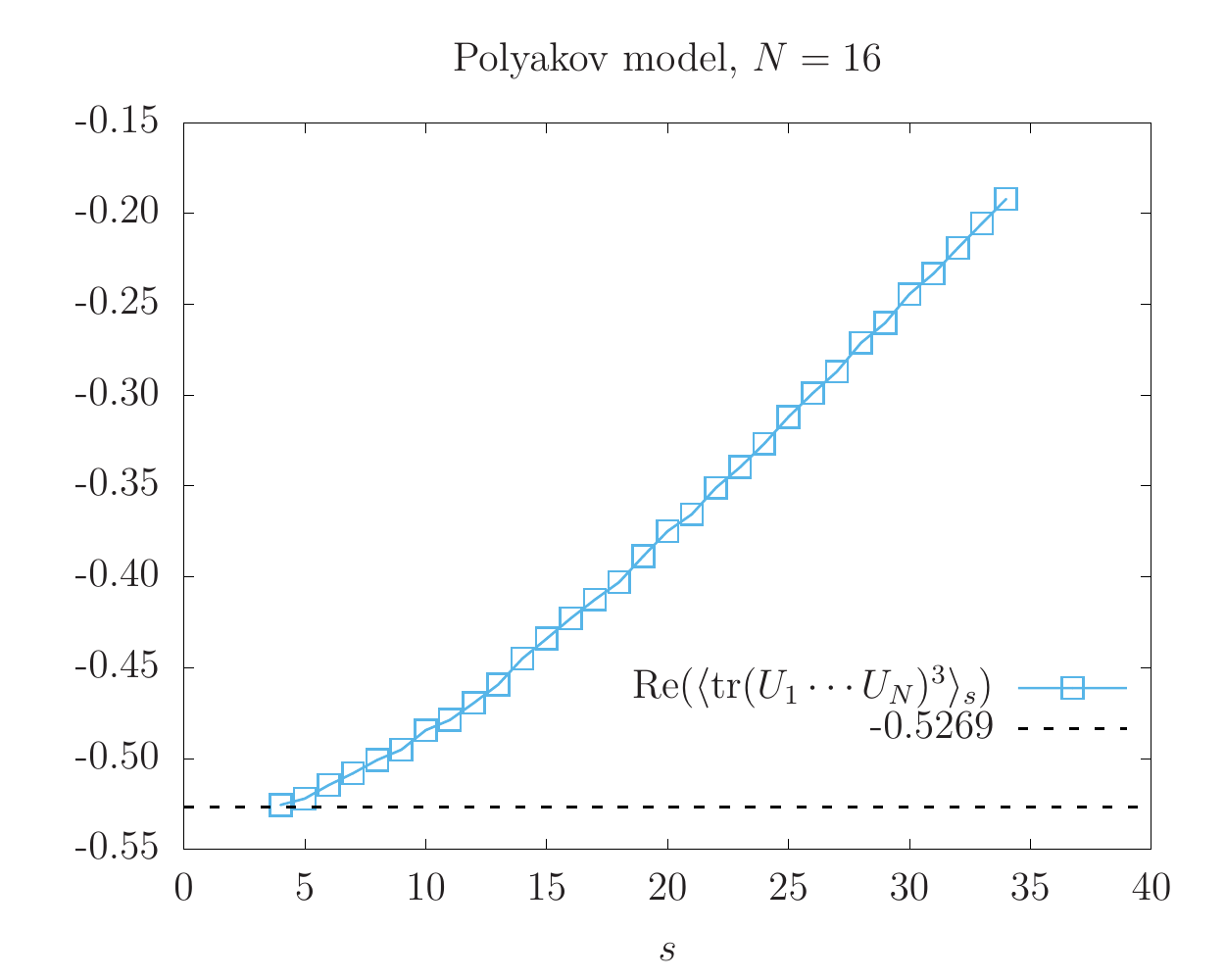}
    \caption{Regularized results for the problem \eqref{eq:su3-polyakov}.}
    \label{fig:polyakov-reg}
\end{figure}

We then consider the $3R$ method and display the results in \Cref{fig:ployakov-rew}. Here we select $s_0 = 8$ and $s_0 = 16$ for which the results are likely to be accurate according to the evolution of $\Delta F$. In all the cases, we observe that when $s$ decreases from $s_0$ to $0$, the numerical results first move toward and then deviate from the exact solutions, which behave similarly to the $\operatorname{U}(1)$ one-link model shown in \Cref{fig:U(1)RewDivFig}. This again confirms that for the high-dimensional integrals, reweighting might fail to provide desired solutions.

\begin{figure}[!ht]
    \centering
    \includegraphics[width=0.32\textwidth]{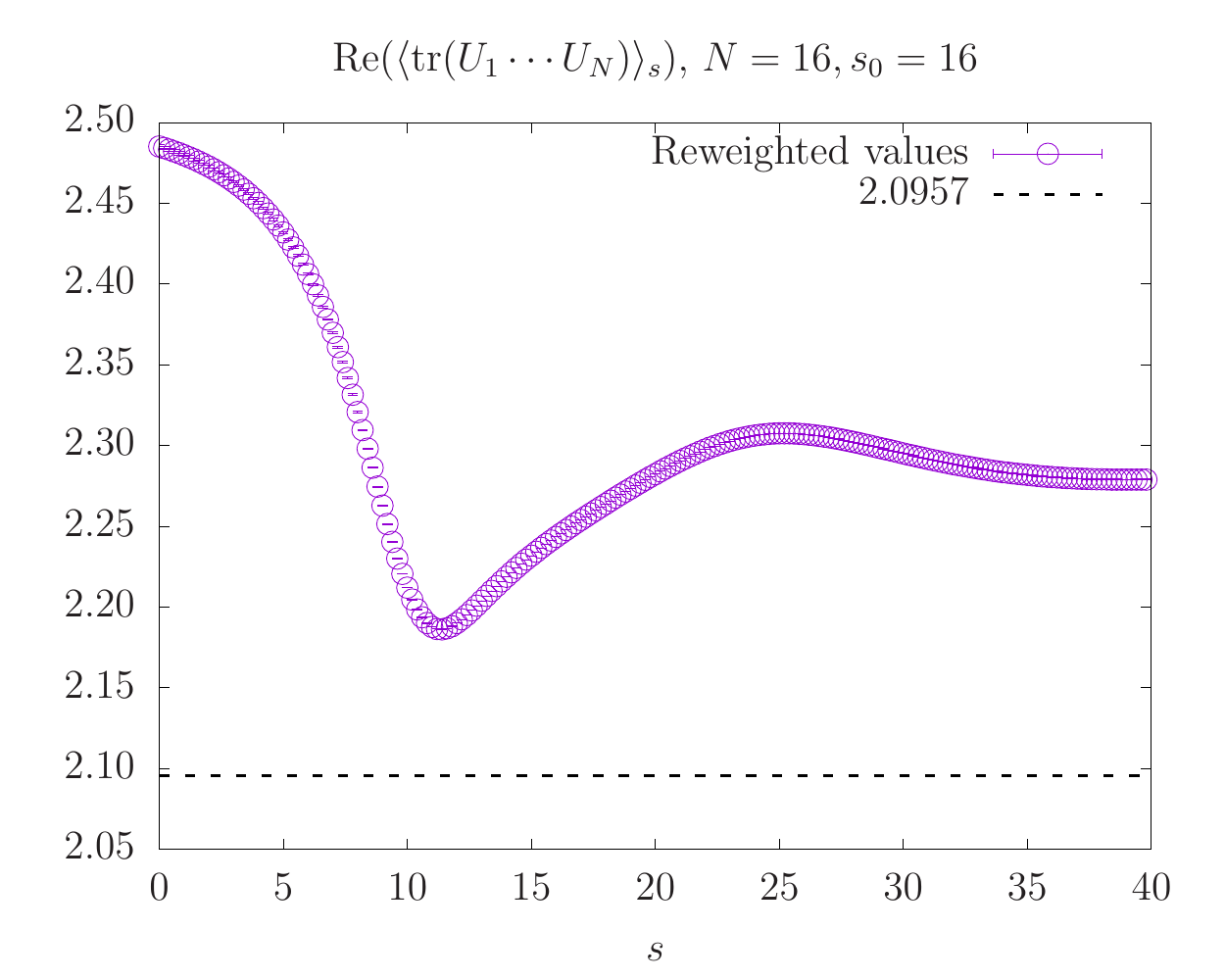}
    \includegraphics[width=0.32\textwidth]{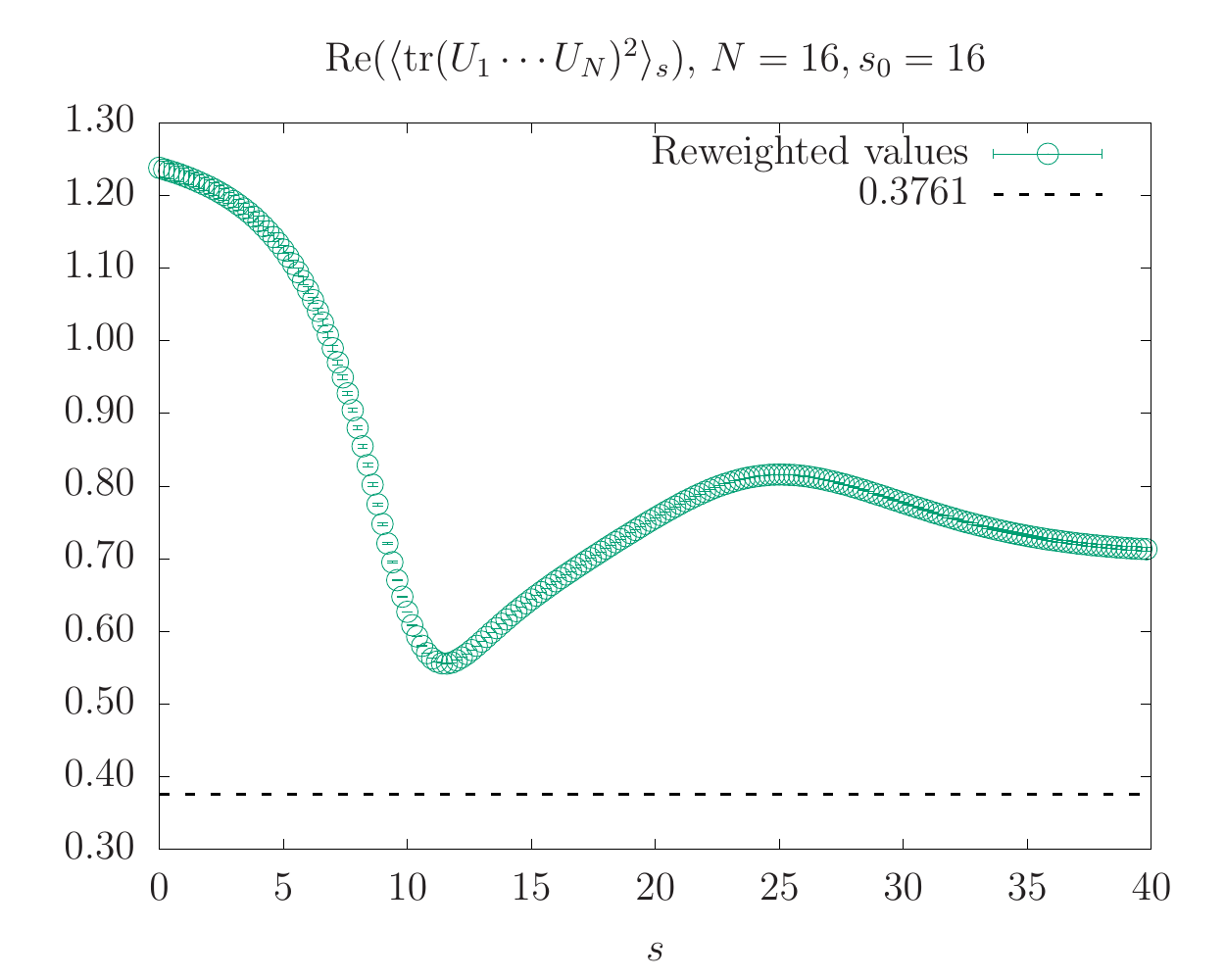}
    \includegraphics[width=0.32\textwidth]{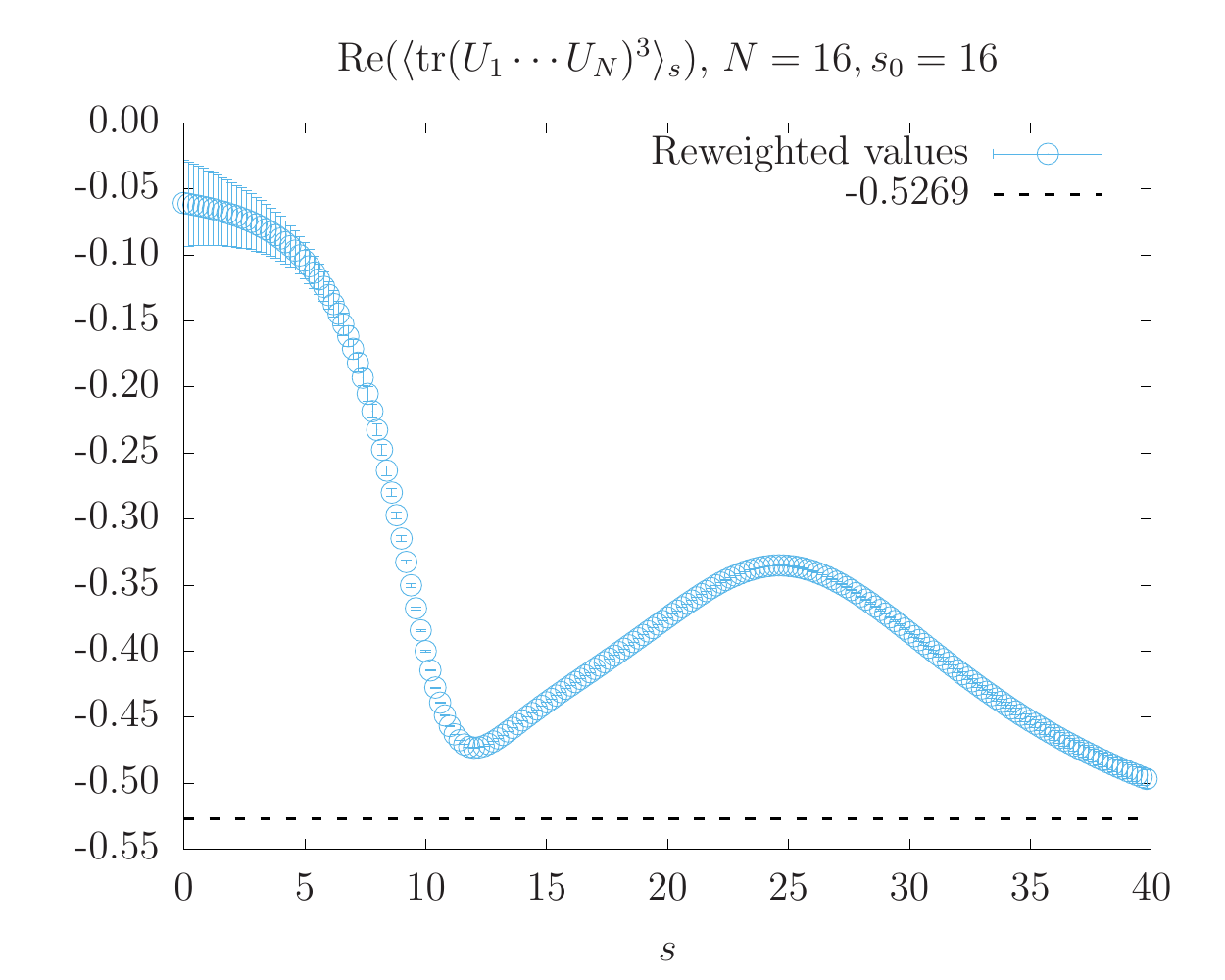}
    \includegraphics[width=0.32\textwidth]{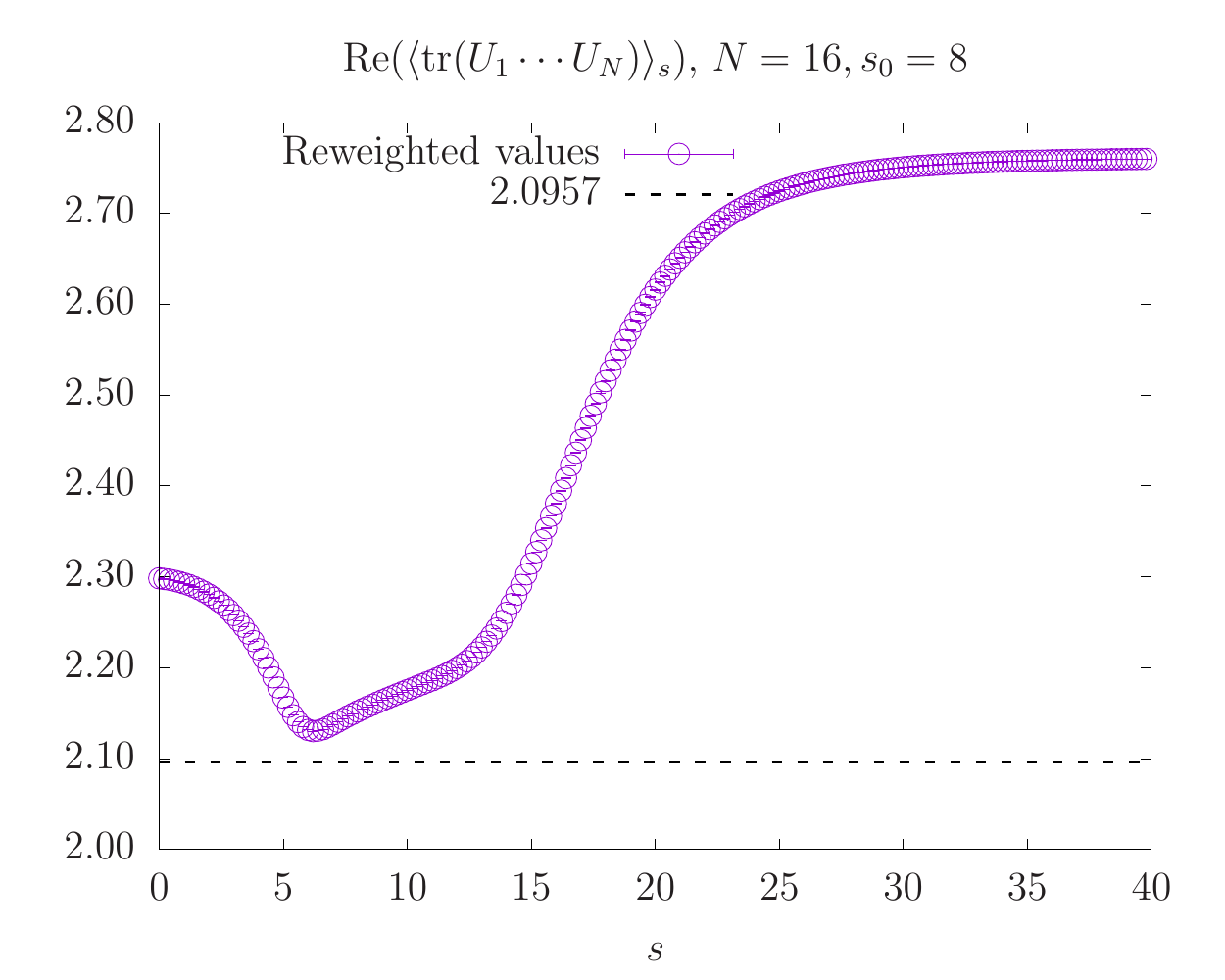}
    \includegraphics[width=0.32\textwidth]{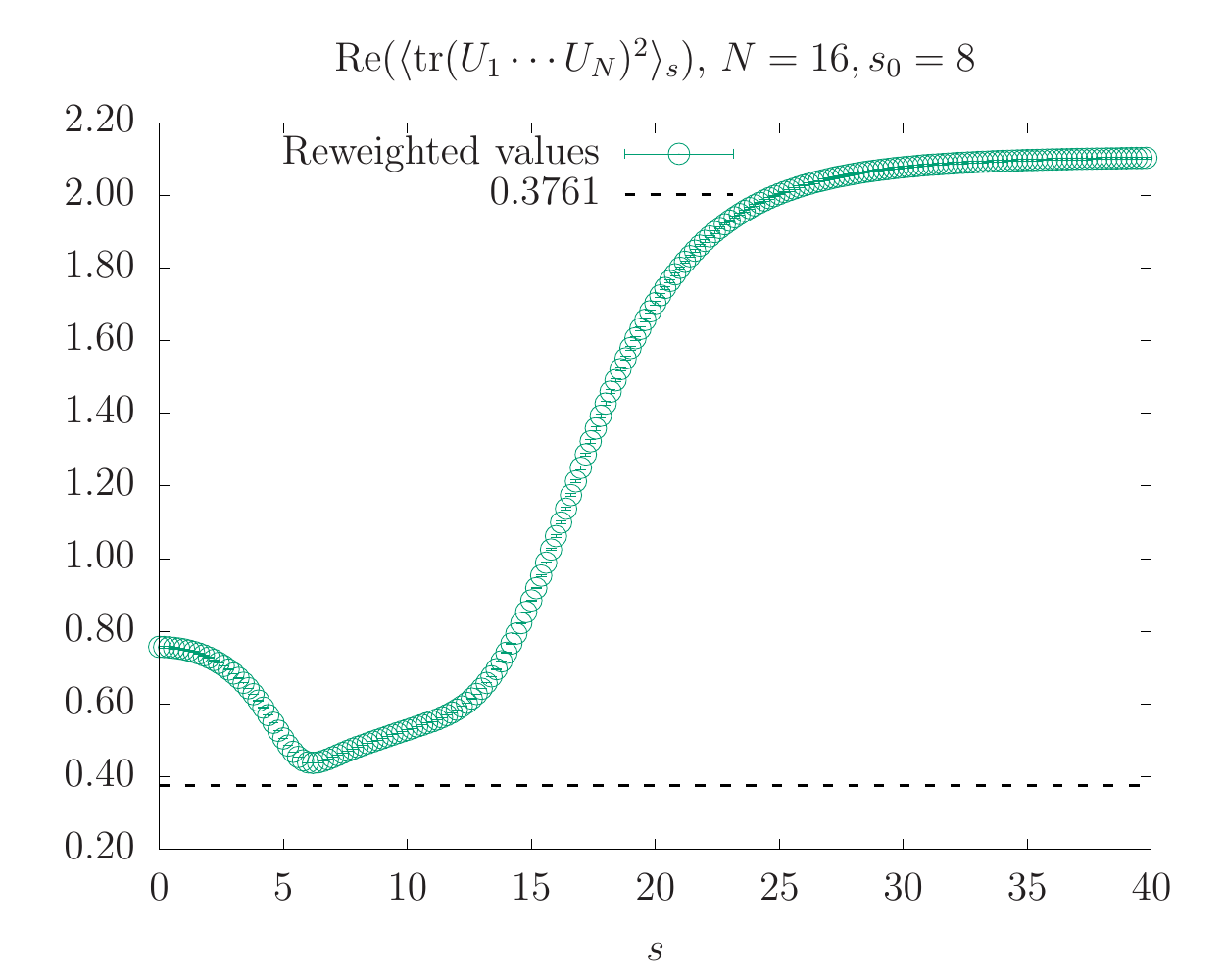}
    \includegraphics[width=0.32\textwidth]{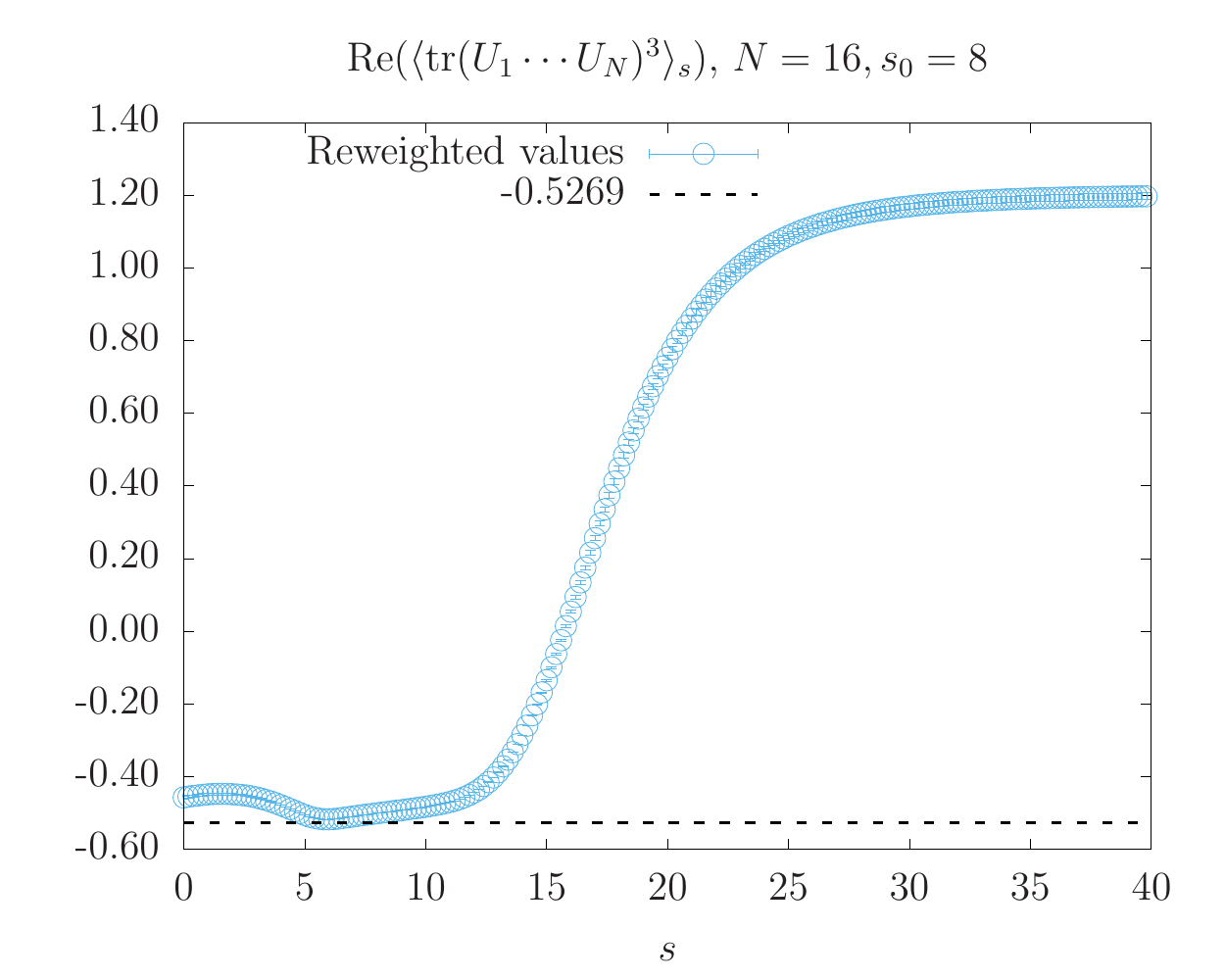}
    \caption{Reweighted regularization of Polyakov model. $s_0=16$ (top) and $s_0=8$ (bottom)}
    \label{fig:ployakov-rew}
\end{figure}

In light of the observations above, we thus consider the use of the $2R$ method to extrapolate the observables. As discussed previously, the proposed regression model is given by
\begin{equation}
\langle O_l \rangle_s = \frac{\sum_{k=0}^M a_k e^{-\frac{N k^2}{4s}}}{1 + \sum_{k=1}^M b_k e^{-\frac{N k^2}{4s}}}.
\end{equation}
We will use data points with $s \in [10, 34]$ in the extrapolation. Note that some of the points closer to $s = 0$ in \Cref{fig:ployakov-rew} were discarded to avoid including points with significant biases. The extrapolations obtained using $M = 3,4,5,6,7$ were plotted in \Cref{fig:polyakov-2R}. For all three observables, the results obtained were generally acceptable, though the values at $s = 0$ were slightly underestimated.
\begin{figure}[!ht]
    \centering
    \includegraphics[width=0.32\textwidth]{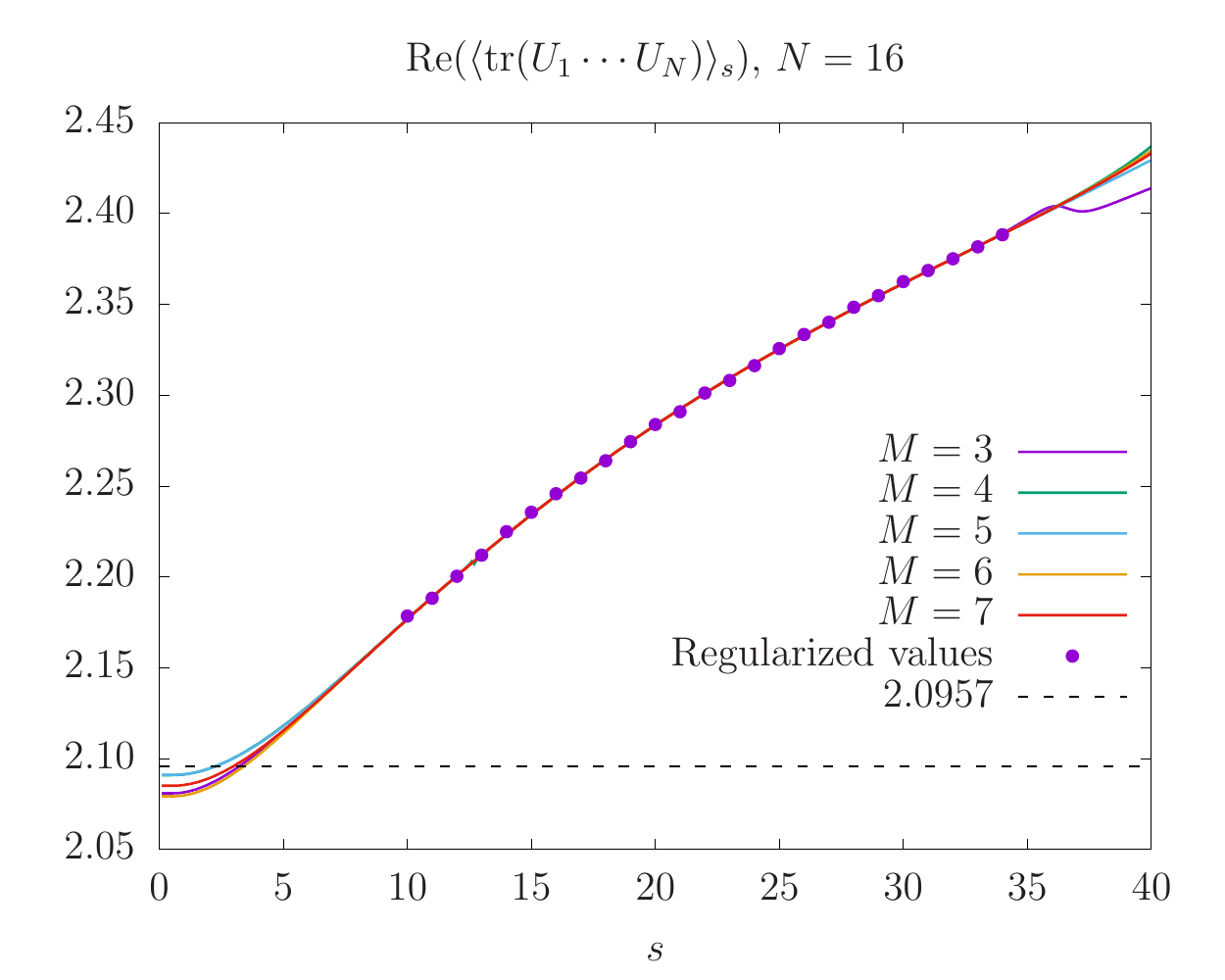}
    \includegraphics[width=0.32\textwidth]{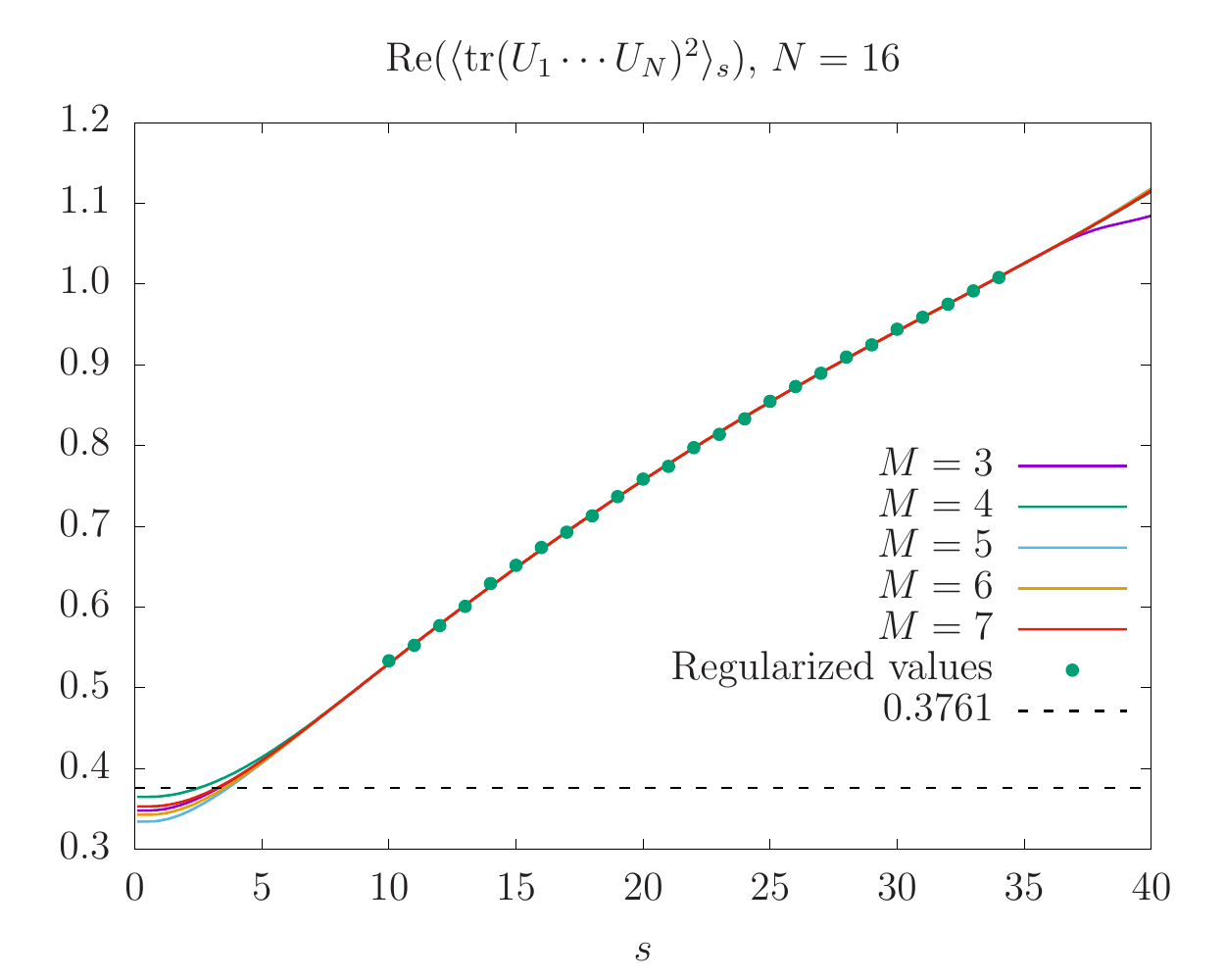}
    \includegraphics[width=0.32\textwidth]{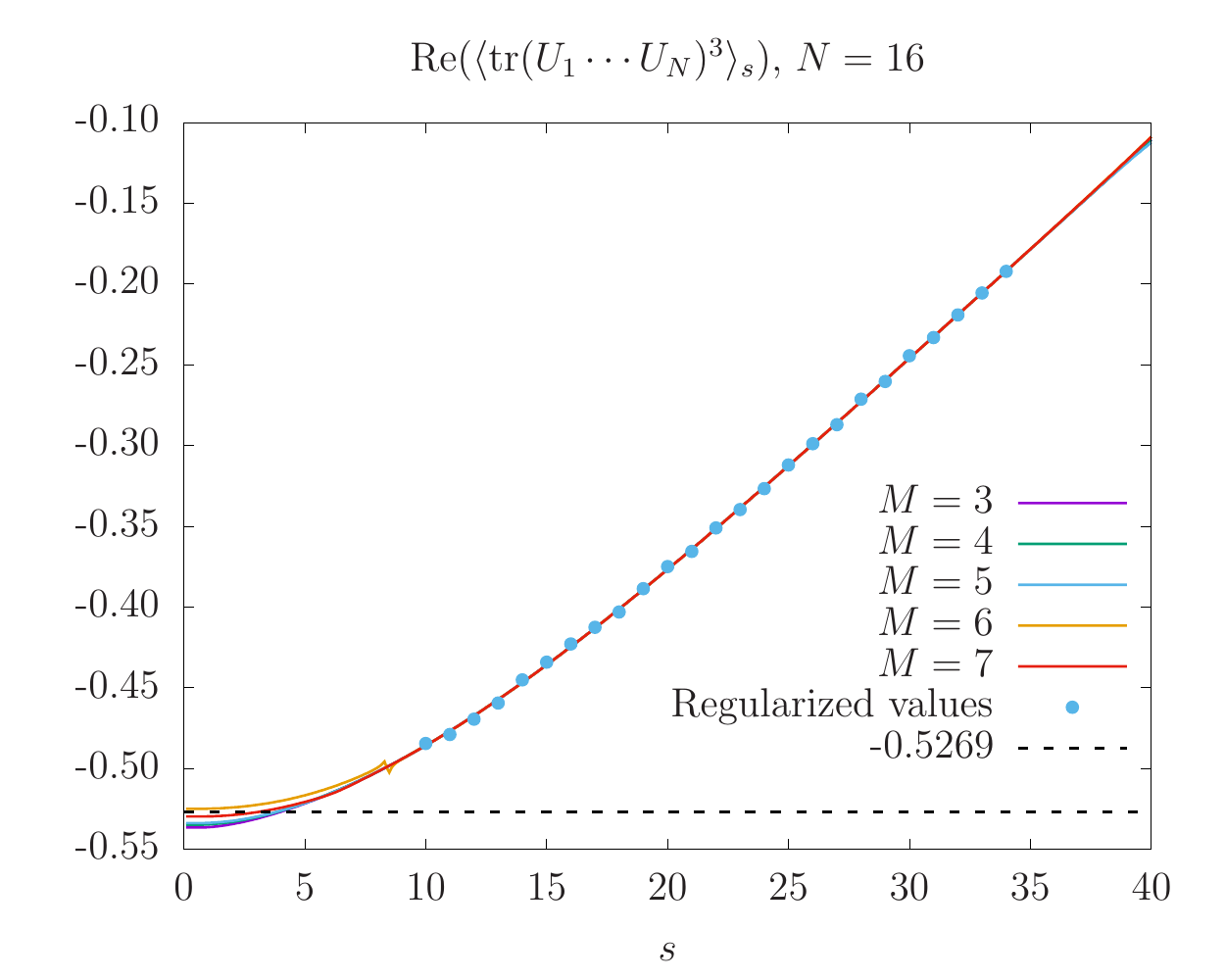}
    \caption{Results of Polyakov model using the $2R$ method}
    \label{fig:polyakov-2R}
\end{figure}

\subsection{Heavy Dense QCD} 
In this section, we consider a more realistic numerical example originating from the heavy dense QCD model at finite potential, studied in several works \cite{aarts2008stochastic, aarts2016qcd, dong2020alternating}.
The field is discretized on a four-dimensional lattice indexed by $x := (x_0, \mathbf{x}) \in X$, where $\mathbf{x} = (x_1, x_2, x_3)$ (see also \eqref{eq:X}).
We consider a vector field $\{U\}$ defined on the lattice, and each $U_{x,\mu}$ is a member of the $\SU{3}$ group.
The action of the heavy dense QCD model is given by
\begin{equation*}
S(\{U\}) = - \ln \det M_{\mu}(\{U\}) + S_B(\{U\}),
\end{equation*}
where $S_B(\{U\})$ is defined by
\begin{equation*}
S_B(\{U\}) = -\beta \sum_{x \in X} \sum_{\nu_1<\nu_2} \left( \frac{1}{6} \left[\tr (U_{x,\nu_1} U_{x+\hat{\nu}_1, \nu_2} U_{x+\hat{\nu}_2, \nu_1}^{-1}  U_{x,\nu_2}^{-1}) + \tr (U_{x,\nu_2} U_{x+\hat{\nu}_2, \nu_1} U_{x+\hat{\nu}_1, \nu_2}^{-1} U_{x,\nu_1}^{-1})\right] - 1 \right),
\end{equation*}
and $\det M_{\mu}$ is the fermionic determinant with $\mu$ being the chemical potential, whose definition is
\begin{equation*}
\det M_{\mu}(\{U\}) = \prod_{\mathbf{x}} \det(I + C \mathcal{P}_\mathbf{x}(\{U\}))^2 \det(I + C' [\mathcal{P}_\mathbf{x}(\{U\})]^{-1})^2
\end{equation*}
where $C = [2 \kappa \exp(\mu)]^{l_0}$, $C' = [2 \kappa \exp(-\mu)]^{l_0}$ with $\kappa$ being the hopping parameter, and
\begin{equation*}
\mathcal{P}_\mathbf{x}(\{U\}) = \prod_{t=1}^{l_0} U_{(t,\mathbf{x}),0}.
\end{equation*}
We refer the readers to \cite{aarts2008stochastic} for the Lie derivatives of this action. 

In our simulation, we worked with two sets of parameters, namely $\mu=2.0, \beta=5.0$ and $\mu=1.4,\beta=5.9$. The value of $\kappa$ was set to be $0.12$. The observable of interest is given by:
\begin{equation*}
P(\{U\}) = \frac{1}{3 l_1 l_2 l_3} \sum_{\mathbf{x}} \tr \mathcal{P}_\mathbf{x}.
\end{equation*}
The parameters of the lattice used in our numerical tests is given by $l_0 = 6$, $l_1 = l_2 = l_3= 8$, and the time step is fixed to be $\Delta t = 2\times 10^{-5}$.

For the heavy dense QCD model, we only considered the $2R$ method. To determine the range of $s$ to be adopted in the regression model, we again study the evolution of the deviation from $\SU{n}$. This is defined in a similar way as \eqref{eq:DF}, given by
\begin{displaymath}
\Delta F = \frac{1}{4l_0 l_1 l_2 l_3} \sum_{\mu=0}^3 \sum_{x \in X} \tr(U_{x,\mu} U_{x,\mu}^{\dagger} - I).
\end{displaymath}
Once again, we observe that larger values of $s$ result in smaller deviations. In both cases, using $s = 1$ reduce the deviation to the magnitude of $10^{-5}$, for which we expect that the results may contain sufficiently small biases and can be used in the regression model.

\begin{figure}[!ht]
\centering
\includegraphics[width=0.45\textwidth]{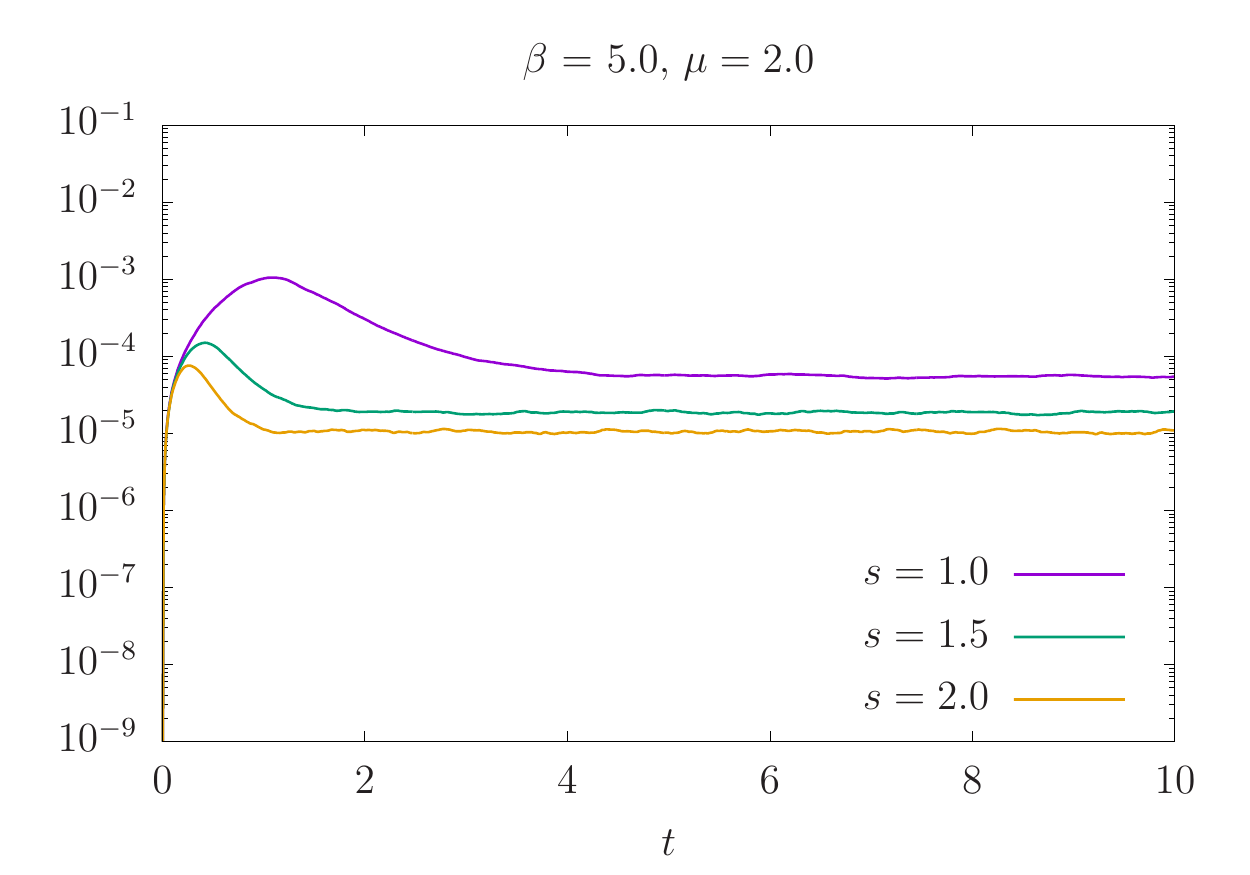}\quad
\includegraphics[width=0.45\textwidth]{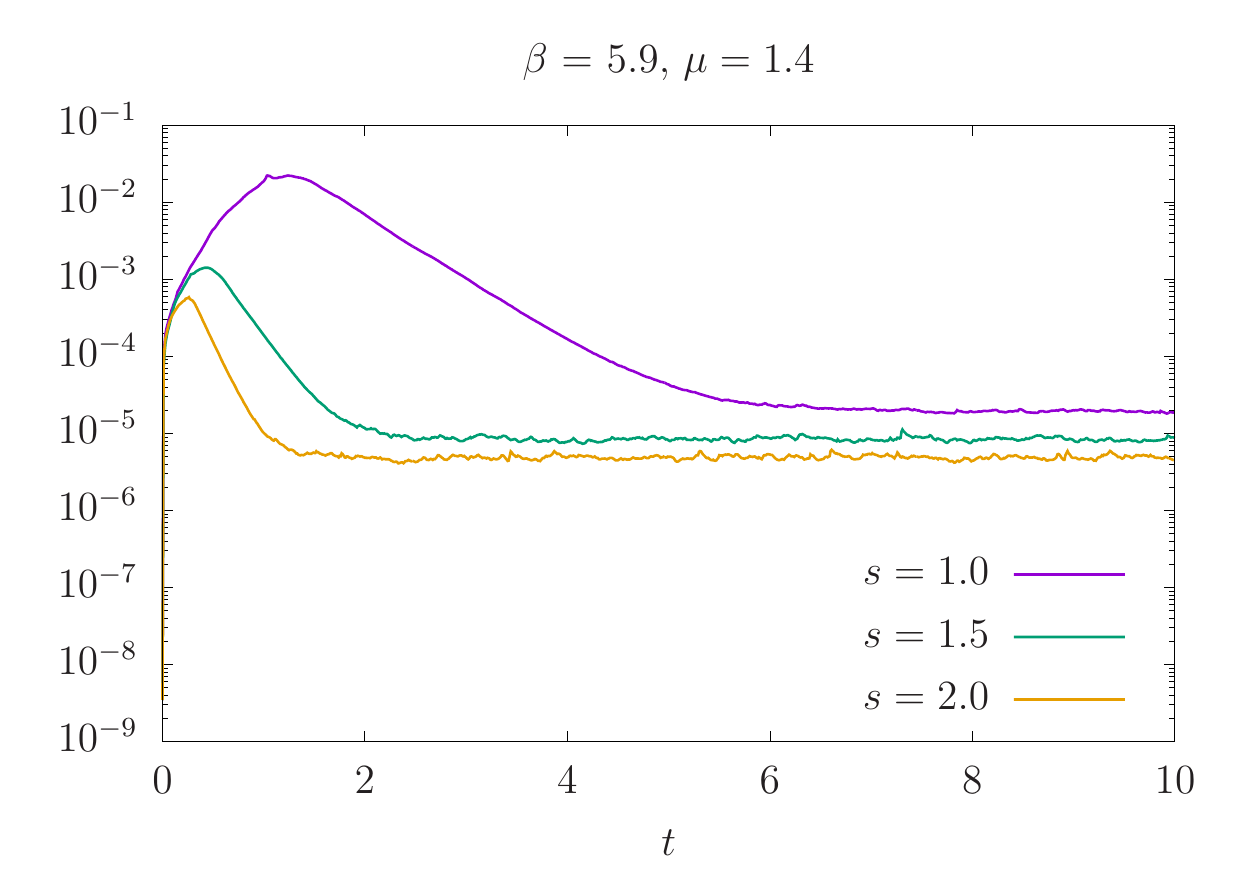}
\caption{The evolution of $\Delta F$. Left: $\beta=5.0$, $\mu = 2.0$; Right:  $\beta=5.9$, $\mu = 1.4$. \label{fig:regdf}}
\end{figure}

To estimate the expectation, we used 6.4 million samples in our tests.
By computing the regularized observable $\langle P \rangle_s$ for $s$ ranging from $1$ to $3.5$, we perform extrapolation based on the expression \eqref{eq:sun-regression} with $b_0 = 1$.
The results are provided in \Cref{fig:su3-reg-extra}. In both cases, the regression results using $M = 6$ and $M = 7$ give similar estimates, whose values at $s = 0$ can be considered as approximations of $\langle P \rangle$.
Note that for $\beta = 5.0$, the complex Langevin method is generally considered to be not applicable in the standard literature. Thus, the result of our approximation obtained by the $2R$ method remains to be validated.
For the case with $\beta = 5.9$, our estimate agrees with that in \cite{seiler2013gauge}, in which gauge cooling is applied to stabilize the dynamics.

\begin{figure}[htbp]
\centering
\includegraphics[width=0.45\textwidth]{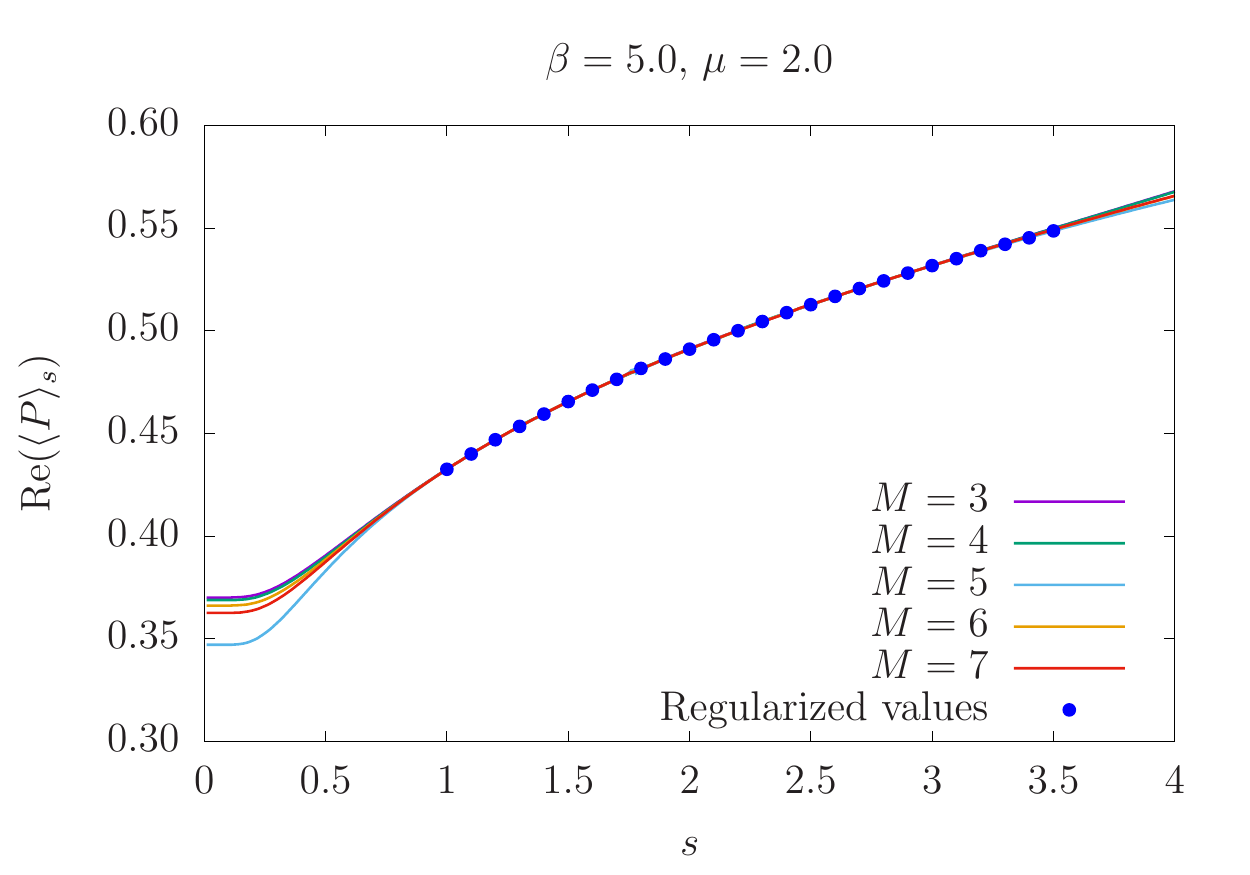}\quad
\includegraphics[width=0.45\textwidth]{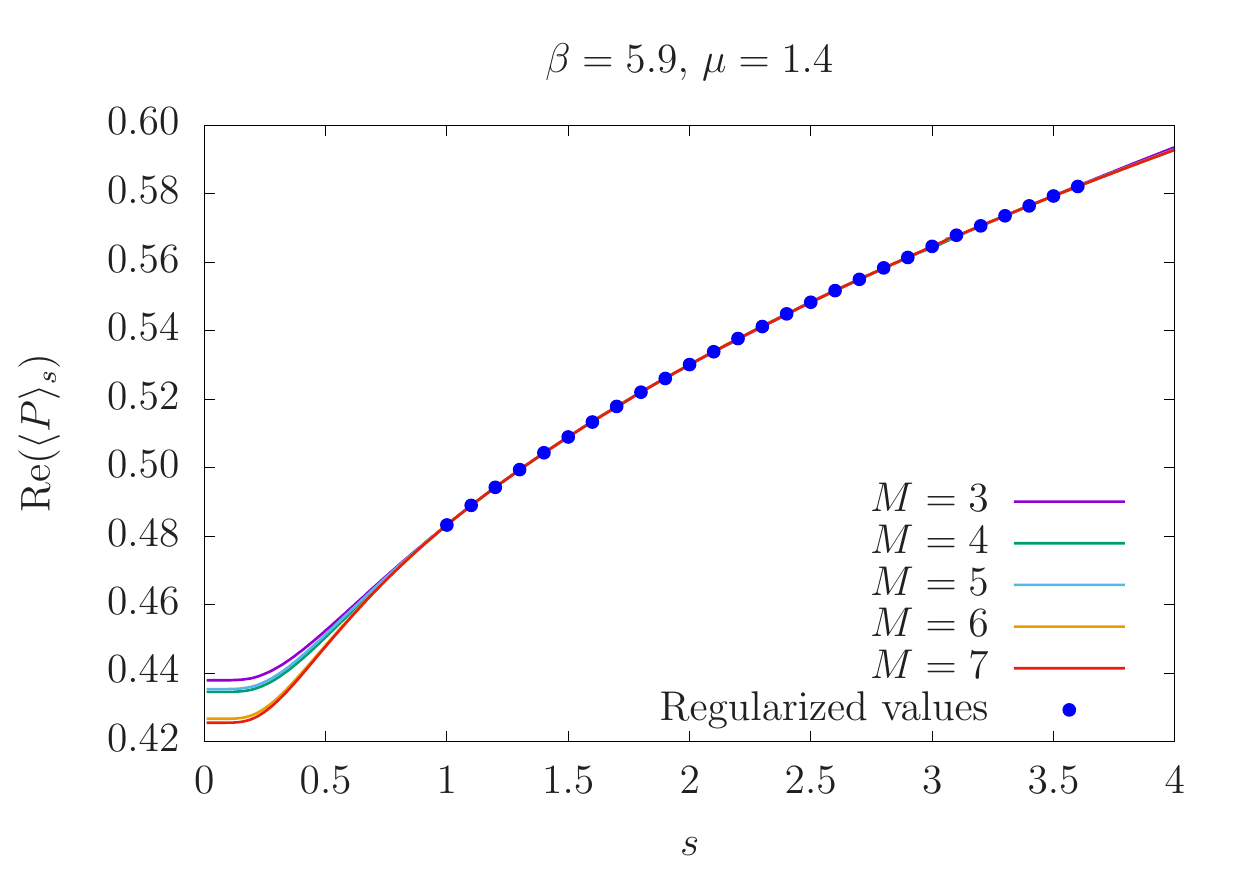}
\caption{Numerical results for the heavy dense QCD model. Left: $\beta=5.0$, $\mu = 2.0$; Right:  $\beta=5.9$, $\mu = 1.4$. \label{fig:su3-reg-extra}}
\end{figure}

\section{Conclusion}\label{sec:con}
We have performed an in-depth study of the regularization of the complex Langevin method.
It is demonstrated that the regularization can produce significant bias in some cases, and we have proposed a few extensions to the regularized complex Langevin method:
\begin{enumerate}
\item The $3R$ method, which performs regression based on the results of the reweighting method proposed in \cite{jacques2017reweight};
\item The $2R$ method, which performs regression based on the regularized results with a number of different parameters.
\end{enumerate}
The computational cost of the $2R$ method is higher than the other two approaches, since multiple complex Langevin dynamics have to be simulated for different regularization constants.
The reweighting method and its $3R$ extension works well in the one-link toy model.
However, it is observed that in the high-dimensional case, the results of the reweighting were reliable only for a very small range of parameters.
The best results are obtained from the $2R$ method, which has successfully simulated one example in lattice QCD for which the original complex Langevin method was known to be inapplicable. We expect that this approach can also be applied to the actions with poles, which will be studied in our future works.

\bibliographystyle{amsplain}
\bibliography{CL}
\end{document}